\documentclass[11pt]{article}
\usepackage[margin=3cm]{geometry}

\usepackage{amsmath,amsthm,amsfonts,amssymb}
\usepackage{thmtools, thm-restate}
\usepackage{tikz}
\usepackage{enumitem}
\usepackage{titling}
\usepackage{tabularx}
\usepackage{graphicx}
\usepackage{appendix}
\usepackage{lipsum}
\usepackage{mathtools}
\usepackage[round]{natbib}
\usepackage{color}
\definecolor{myurlcolor}{rgb}{0,0.35,0}
\definecolor{mycitecolor}{rgb}{0,0,0.55}
\definecolor{myrefcolor}{rgb}{0.55,0,0}
\usepackage[pagebackref,draft=false]{hyperref}
\hypersetup{colorlinks,
linkcolor=myrefcolor,
citecolor=mycitecolor,
urlcolor=myurlcolor}

\usepackage{cancel}
\usepackage{bbm}
\usepackage{enumitem}
\setlist[enumerate]{label=(\arabic*), ref=\arabic*}
\usepackage{cite}
\usepackage{url}
\usepackage{tikz-cd}
\usepackage{bbm}
\usepackage{todonotes}
\usepackage[capitalize,noabbrev]{cleveref}
\usepackage{fancyhdr}
\usepackage{fancyhdr}
\usepackage{pifont}   % for \ding etc.
\usepackage{array}    % extended tabular column types

\tikzstyle{ndash} = [draw, dashed, shape=circle, minimum size=20pt,inner sep=0pt]
\tikzstyle{ndexo} = [draw, dashed, shape=circle, minimum size=20pt,inner sep=0pt, fill=lightgray]
\tikzstyle{ndsel} = [draw, semithick, shape=regular polygon, regular polygon sides=3,
                     minimum size=25pt,inner sep=0pt, fill=lightgray,shape border rotate=180]

\providecommand{\pequ}{\stackrel{p}{=} }
\providecommand{\abe}{\mathrm{abe}}
\providecommand{\SA}{\mathrm{SA}}
\providecommand{\rp}{\mathrm{ep}} % kept as in the original first template

\newtheorem{theorem}{Theorem}[section]
\newtheorem{proposition}[theorem]{Proposition}
\newtheorem{lemma}[theorem]{Lemma}
\newtheorem{corollary}[theorem]{Corollary}
\newtheorem{definition}[theorem]{Definition}

\newtheorem{notation}[theorem]{Notation}
\theoremstyle{definition}
\newtheorem{example}[theorem]{Example}
\newtheorem{remark}[theorem]{Remark}

\newtheorem{assumption}[theorem]{Assumption}

\newtheorem{algorithm}[theorem]{Algorithm}

\newtheorem{defprop}[theorem]{Proposition/Definition}
\crefname{definitionthm}{Definition/Theorem}{Definitions/Theorems}
\crefname{notation}{Notation}{Notations}
\crefname{question}{Question}{Questions}

\newcommand{\Si}{\mathcal{S}}
\newcommand{\Ac}{\mathcal{A}}
\newcommand{\Bc}{\mathcal{B}}
\newcommand{\Cc}{\mathcal{C}}
\newcommand{\Dc}{\mathcal{D}}

\newcommand{\Fc}{\mathcal{F}}
\newcommand{\Gc}{\mathcal{G}}
\newcommand{\Hc}{\mathcal{H}}
\newcommand{\Ic}{\mathcal{I}}
\newcommand{\Jc}{\mathcal{J}}

\newcommand{\Lc}{\mathcal{L}}

\newcommand{\Nc}{\mathcal{N}}

\newcommand{\Pc}{\mathcal{P}}

\newcommand{\Sc}{\mathcal{S}}

\newcommand{\Uc}{\mathcal{U}}

\newcommand{\Xc}{\mathcal{X}}

\newcommand{\Pf}{\mathfrak{P}}

\newcommand{\Prb}{\mathrm{P}}
\newcommand{\Qr}{\mathrm{Q}}

\newcommand{\Mb}{\mathbb{M}}
\newcommand{\Nb}{\mathbb{N}}

\newcommand{\Rb}{\mathbb{R}}

\usepackage{centernot}
\usepackage{stmaryrd}

\usepackage{algorithm}
\usepackage{algpseudocode}

\let\hat\widehat
\let\tilde\widetilde

\newcommand{\ol}{\overline}
\newcommand{\ul}{\underline}
\newcommand{\cf}{cf.\ }
\newcommand{\pr}{\operatorname{pr}}
\newcommand{\sd}{\,\triangle\,}
\newcommand{\Uni}{\mathrm{Uni}}

\newcommand{\notind}{\not\!\perp\!\!\!\perp}
\newcommand{\notsep}{\not\perp}
\newcommand{\ind}{\perp\!\!\!\perp}

\newcommand{\Do}{\mathrm{do}}
\newcommand{\anc}{\mathrm{Anc}}
\newcommand{\sib}{\mathrm{Sib}}
\newcommand{\ch}{\mathrm{Ch}}

\newcommand{\pa}{\mathrm{Pa}}

\newcommand{\twin}{\mathrm{twin}}

\newcommand{\SCM}{(V,W,\Xc,\Prb,f)}

\newcommand{\sm}{\setminus}
\newcommand{\dcup}{\,\dot{\cup}\,}
\newcommand{\wrt}{w.r.t.\ }

\newcommand{\almostall}{a.a.\ }
\newcommand{\Ibbm}{\mathbbm{1} }
\newcommand{\resp}{resp.\ }

\newcommand{\sep}[2]{\underset{#2}{\stackrel{#1}{\perp}}}

\usepackage{tikz}
\usepackage{tikz-cd}
\usetikzlibrary{arrows,arrows.meta,calc,fit}
\usetikzlibrary{positioning,patterns,matrix}
\usetikzlibrary{decorations.markings,decorations.pathreplacing,decorations.pathmorphing}
\usetikzlibrary{shapes,shapes.arrows,shapes.geometric,shapes.multipart}
\tikzstyle{ndout} = [draw, semithick, shape=circle, minimum size=20pt,inner sep=0pt]
\tikzstyle{ndash} = [draw, dashed, shape=circle, minimum size=20pt,inner sep=0pt]
\tikzstyle{ndexo} = [draw, dashed, shape=circle, minimum size=20pt,inner sep=0pt, fill=lightgray]
\tikzstyle{ndlat} = [draw, semithick, shape=circle, minimum size=20pt,inner sep=0pt, fill=lightgray]
\tikzstyle{ndsel} = [draw, semithick, shape=regular polygon, regular polygon sides=3, minimum size=25pt,inner sep=0pt, fill=lightgray,shape border rotate=180]
\tikzstyle{arout} = [style={->,>=Latex}]
\tikzstyle{arlout} = [style={->,>=Latex}]
\tikzstyle{arlat} = [style={<->,>=Latex}]
\newcommand{\arrhead}{{Latex}}
\newcommand{\arrtail}{{}}
\newcommand{\arrstar}{Rays[n=6]}

\newcommand*{\hut}[1][]{\mathrel{\tikz [baseline=-0.25ex,\arrhead-\arrtail, #1] \draw [#1] (0pt,0.5ex) -- (1.3em,0.5ex);}}

\newcommand*{\tuh}[1][]{\mathrel{\tikz [baseline=-0.25ex,\arrtail-\arrhead, #1] \draw [#1] (0pt,0.5ex) -- (1.3em,0.5ex);}}

\newcommand*{\huh}[1][]{\mathrel{\tikz [baseline=-0.25ex,\arrhead-\arrhead, #1] \draw [#1] (0pt,0.5ex) -- (1.3em,0.5ex);}}

\newcommand*{\hus}[1][]{\mathrel{\tikz [baseline=-0.25ex,\arrhead-\arrstar, #1] \draw [#1] (0pt,0.5ex) -- (1.3em,0.5ex);}}

\newcommand*{\suh}[1][]{\mathrel{\tikz [baseline=-0.25ex,\arrstar-\arrhead, #1] \draw [#1] (0pt,0.5ex) -- (1.3em,0.5ex);}}

\newcommand*{\sus}[1][]{\mathrel{\tikz [baseline=-0.25ex,\arrstar-\arrstar, #1] \draw [#1] (0pt,0.5ex) -- (1.3em,0.5ex);}}

\tikzstyle{ndint} = [draw, semithick, shape=rectangle, minimum size=20pt,inner sep=0pt]
\tikzstyle{ndout} = [draw, semithick, shape=circle, minimum size=20pt,inner sep=0pt]
\tikzstyle{ndlat} = [draw, semithick, shape=circle, minimum size=20pt,inner sep=0pt, fill=lightgray]
\tikzstyle{arint} = [style={->,>=Latex,thick}] %teal
\tikzstyle{arout} = [style={->,>=Latex,thick}] %blue
\tikzstyle{arlout} = [style={->,>=Latex,thick}] %red
\tikzstyle{arlat} = [style={<->,>=Latex,thick}] %red

\tikzstyle{out} = [style={o->,>=Latex,style=semithick}]
\tikzstyle{hut} = [style={<-,>=Latex,style=semithick}]
\tikzstyle{tut} = [style=semithick]
\tikzstyle{tuh} = [style={->,>=Latex,style=semithick}]
\tikzstyle{tuo} = [style={-o,style=semithick}]
\tikzstyle{huh} = [style={<->,>=Latex,style=semithick}]
\tikzstyle{ouo} = [style={o-o,style=semithick}]
\tikzstyle{huo} = [style={<-o,>=Latex,style=semithick}]
\tikzstyle{ouh} = [style={o->,>=Latex,style=semithick}]
\tikzstyle{ous} = [style={o-{Rays[n=6]},style=semithick}]
\tikzstyle{hus} = [style={<-{Rays[n=6]},>=Latex,style=semithick}]
\tikzstyle{tus} = [style={-{Rays[n=6]},style=semithick}]
\tikzstyle{sut} = [style={{Rays[n=6]}-,style=semithick}]
\tikzstyle{suh} = [style={{Rays[n=6]}->,>=Latex,style=semithick}]
\tikzstyle{suo} = [style={{Rays[n=6]}-o,style=semithick}]
\tikzstyle{sus} = [style={{Rays[n=6]}-{Rays[n=6]},style=semithick}]

\long\def\acks#1{\vskip 0.3in\noindent{\large\bf Acknowledgments and Disclosure of Funding}\vskip 0.2in
\noindent #1}

\newenvironment{keywords}
{\bgroup\leftskip 20pt\rightskip 20pt \small\noindent{\bf Keywords:} }%
{\par\egroup\vskip 0.25ex}

\title{Foundations of Structural Causal Models with Latent Selection}
\author{Leihao Chen\thanks{Korteweg-de Vries Institute for Mathematics, University of Amsterdam, Amsterdam, the Netherlands; {\tt l.chen2@uva.nl}},~~~
        Onno Zoeter\thanks{Booking.com, The Netherlands; {\tt onno.zoeter@booking.com}}, ~~and~~
        Joris M.~Mooij\thanks{Korteweg-de Vries Institute for Mathematics, University of Amsterdam, Amsterdam, the Netherlands; {\tt j.m.mooij@uva.nl}}}

\begin{document}
\maketitle

\date{\today}

\begin{abstract}
Three distinct phenomena complicate statistical causal analysis: latent common causes, causal cycles, and latent selection. Foundational works on Structural Causal Models (SCMs), e.g., Bongers et al. (2021, Ann. Stat., 49(5): 2885-2915), treat cycles and latent variables, while an analogous account of latent selection is missing. The goal of this article is to develop a theoretical foundation for modeling latent selection with SCMs. To achieve that, we introduce a conditioning operation for SCMs: it maps an SCM with explicit selection mechanisms to one without them while preserving the causal semantics of the selected subpopulation. Graphically, in Directed Mixed Graphs we extend bidirected edges—beyond latent common causes—to also encode latent selection. We prove that the conditioning operation preserves simplicity, acyclicity, and linearity of SCMs, and interacts well with marginalization, conditioning, and interventions. These properties make those three operations valuable tools for causal modeling, reasoning, and learning after abstracting away latent details (latent common causes and selection). Examples show how this abstraction streamlines analysis and clarifies when standard tools (e.g., adjustment, causal calculus, instrumental variables) remain valid under selection bias. We hope that these results deepen the SCM-based understanding of selection bias and become part of the standard causal modeling toolbox to build more reliable causal analysis.
\end{abstract}


\newpage
\tableofcontents

\newpage

\section{Introduction}

\citet{bongers2021foundations} provide a general measure-theoretic foundational theory for causal modeling with Structural Causal Models (SCMs) with cycles and latent variables, but an analogous treatment for latent selection is still absent. Addressing (latent) selection bias remains a significant challenge \citep{wald1943method,Heckman79SampleSB,zhao21bets,fryer2019empirical,cooper95causal_discovery_selecion}. For example,  in some cases unconsciously selecting samples can induce ``spurious dependency'' among collected samples, and therefore the famous Berkson's paradox arises \citep{berkson1946limitations,Munafo2016ColliderSW}. There are many types of selection bias without universally accepted definitions and various methods to address them \citep{lu2022toward,smith2020selection}. In this work, we focus on ``truncating selection bias," which occurs when an underlying (unobserved) filtering process, denoted ``$X_S \in \Si$", selects individual samples where the variable $X_S$ takes values within a set $\Si$. In probability theory, this can be modeled as conditioning on the event $\{X_S \in \Si\}$.

%  where the former is an informal notation for ``$X_S$ takes values in $\Si$'' and the latter is a set $\{\omega \in \Omega: X_S(\omega)\in \Si\}$ with the underlying probability space $(\Omega,\mathcal{F},\mathbb{P})$

To understand its structural behavior, one approach is to model selection bias via a causal model that explicitly describes the selection mechanism \citep{pearl2009causality,bareinboim2012controlling,daniel2012using,Abouei24sID,hernan2004structural}. This necessitates detailed knowledge of the selection mechanism. However, in many situations, the selection mechanism is unobserved \citep{cooper95causal_discovery_selecion}, which makes such knowledge unavailable and introduces a layer of complexity with infinitely many possibilities. The goal of the current work is to study how to model latent selection by effectively abstracting away its details in an SCM \citep{pearl2009causality,bongers2021foundations}.

\subsection{Motivation}

Marginalization of causal models is a powerful tool for abstracting away latent details, which makes causal modeling more manageable and trustworthy \citep{bongers2021foundations,pearl2009causality,Evans16graphs_margin}. By marginalizing out latent variables, we use one simplified model to represent infinitely many complex models, abstracting away unnecessary latent details while preserving essential causal information such as observational/interventional/counterfactual distributions, $d$-separations or $\sigma$-separations \citep{pearl2009causality,Forre2017markov}, and ancestral relationships among the observed variables. The SCM marginalization and causal graph marginalization interact well and part of the nice properties can be compactly expressed via Figure~\ref{fig:diagram1}.\footnote{Replacing $M$ with the intervened model $M_{\Do(X_T=x_T)}$ where we intervene the variable $X_T$ to take on the value $x_T$ (\cf Definition \ref{def:scm_hard_intervention}), we get the corresponding results for interventional distributions and intervened graphs.}

\begin{figure}[ht]
    \centering
    \begin{tikzcd}[cells={nodes={scale=0.9}}]
   X_A \underset{\Prb_M(X_{V})}{\ind} X_B\mid X_C \arrow[d, Leftrightarrow] \arrow[r, "\text{\tiny{$d/\sigma$-Faithful}}"', Rightarrow, shift right, shift right] & A\underset{G(M)}{\overset{d/\sigma}{\perp}}B\mid C   \arrow[d, Leftrightarrow] \arrow[l, "\text{\tiny{$d/\sigma$-Markov}}"', Rightarrow, shift right, shift right]\\
   X_A \underset{\Prb_{M_{\sm L}}(X_{V\sm L})}{\ind} X_B\mid X_C  \arrow[r, "\text{\tiny{$d/\sigma$-Faithful}}"', Rightarrow, shift right, shift right] & A\underset{G(M)_{\sm L}}{\overset{d/\sigma}{\perp}}B\mid C \arrow[l, "\text{\tiny{$d/\sigma$-Markov}}"', Rightarrow, shift right, shift right]
  \end{tikzcd}
    \caption{The logical relations between stochastic conditional independence in simple SCM $M$ and marginalized model $M_{\sm L}$, and graphical separation ($d$- or $\sigma$-separation) in causal graph $G(M)$ and marginalized graph $G(M)_{\sm L}$. SCM $M$ has endogenous variables $X_V$ with $X_L$ latent, and  $A,B,C\subseteq V\sm L$. The terms ``$d/\sigma$-Markov'' and ``$d/\sigma$-Faithful'' represent $d$- or $\sigma$-Markov property and $d$- or $\sigma$-faithfulness, respectively, regarding $\Prb_M(X_V)$ and $G(M)$ (top), and $\Prb_{M_{\sm L}}(X_{V\sm L})$ and $G(M)_{\sm L}$ (bottom).}
    \label{fig:diagram1}
\end{figure}

For instance, the model $G$ in Figure \ref{fig:marg_abstract} effectively abstracts models $G^i$ for $i=1,\ldots, 5,\ldots$, yielding the same identification result (under discreteness and positivity assumptions) regardless of the latent structure (front-door criterion \citep{pearl2009causality}):\footnote{More generally, the ID-algorithm was first proved to be sound and complete for models with bidirected edges and then the results can be translated to the case with arbitrary latent structures via marginalization \citep{tian02general,Huang2008OnTC,richardson2023nested}.}
\[
\Prb(C=c\mid \Do(S=s))=\sum_{t}\Prb(C=c\mid T=t)\Prb(T=t\mid S=s).
\]

\begin{figure}[ht]
\centering
\begin{tikzpicture}[scale=0.6, transform shape]
\begin{scope}[xshift=-2cm]
    \node[ndout] (S) at (0,0) {$S$};
    \node[ndout] (T) at (1.5,0) {$T$};
    \node[ndout] (C) at (3,0) {$C$};
    \draw[arout] (S) to (T);
    \draw[arout] (T) to (C);
    \draw[arlat, bend left] (S) to (C);
    \node at (1.5,-1) {$G$};
\end{scope}
\node[align=center] at (3,0) {\Large Abstracts:};
\draw[thick, decorate, decoration={brace, amplitude=10pt}] (5,-3) -- (5,3) node[midway, below=15pt] {};
    \begin{scope}[xshift=6cm,yshift=2cm]
    \node[ndout] (S) at (0,0) {$S$};
    \node[ndout] (T) at (1.5,0) {$T$};
    \node[ndout] (C) at (3,0) {$C$};
    \node[ndlat] (L) at (1.5,1) {$L$};
    \draw[arout] (S) to (T);
    \draw[arout] (T) to (C);
    \draw[arout] (L) to (S);
    \draw[arout] (L) to (C);
    \node at (1.5,-1) {$G^1$};
\end{scope}

\begin{scope}[xshift=11cm,yshift=2cm]
    \node[ndout] (S) at (0,0) {$S$};
    \node[ndout] (T) at (1.5,0) {$T$};
    \node[ndout] (C) at (3,0) {$C$};
    \node[ndlat] (L1) at (1.5,1) {$L_1$};
    \node[ndlat] (L2) at (0,1.5) {$L_2$};
    \draw[arout] (S) to (T);
    \draw[arout] (T) to (C);
    \draw[arout] (L1) to (S);
    \draw[arout] (L1) to (C);
    \draw[arout] (L2) to (S);
    \draw[arout] (L2) to (L1);
    \node at (1.5,-1) {$G^2$};
\end{scope}

\begin{scope}[xshift=16cm,yshift=2cm]
    \node[ndout] (S) at (0,0) {$S$};
    \node[ndout] (T) at (1.5,0) {$T$};
    \node[ndout] (C) at (3,0) {$C$};
    \node[ndlat] (L1) at (1.5,1) {$L_1$};
    \node[ndlat] (L2) at (0,1.5) {$L_2$};
    \node[ndlat] (L3) at (3,1.5) {$L_3$};
    \draw[arout] (S) to (T);
    \draw[arout] (T) to (C);
    \draw[arout] (L1) to (S);
    \draw[arout] (L1) to (C);
    \draw[arout] (L2) to (S);
    \draw[arout] (L2) to (L1);
    \draw[arout] (L3) to (L1);
    \draw[arout] (L3) to (C);
    \node at (1.5,-1) {$G^3$};
\end{scope}

\begin{scope}[xshift=6cm,yshift=-3cm]
    \node[ndout] (S) at (0,0) {$S$};
    \node[ndout] (T) at (1.5,0) {$T$};
    \node[ndout] (C) at (3,0) {$C$};
    \node[ndlat] (L1) at (1.5,1) {$L_1$};
    \node[ndlat] (L2) at (0,1.5) {$L_2$};
    \node[ndlat] (L3) at (3,1.5) {$L_3$};
    \node[ndlat] (L4) at (1.5,3) {$L_4$};
    \draw[arout] (S) to (T);
    \draw[arout] (T) to (C);
    \draw[arout] (L1) to (S);
    \draw[arout] (L1) to (C);
    \draw[arout] (L2) to (S);
    \draw[arout] (L2) to (L1);
    \draw[arout] (L3) to (L1);
    \draw[arout] (L3) to (C);
    \draw[arout] (L4) to (L2);
    \draw[arout] (L4) to (L1);
    \draw[arout] (L4) to (L3);
    \node at (1.5,-1) {$G^4$};
\end{scope}
\node[align=center] at (19,-2) {\Large etc.};
\begin{scope}[xshift=11cm,yshift=-3cm]
    \node[ndout] (S) at (0,0) {$S$};
    \node[ndout] (T) at (3,0) {$T$};
    \node[ndout] (C) at (6,0) {$C$};
    \node[ndlat] (L1) at (3,1) {$L_1$};
    \node[ndlat] (L2) at (1.5,1.5) {$L_2$};
    \node[ndlat] (L3) at (4.5,1.5) {$L_3$};
    \node[ndlat] (L4) at (1.5,0) {$L_4$};
    \node[ndlat] (L5) at (4.5,0) {$L_5$};
    \node[ndlat] (L6) at (0,1.5) {$L_6$};
    \node[ndlat] (L7) at (6,1.5) {$L_7$};
    \draw[arout] (S) to (L4);
    \draw[arout] (T) to (L5);
    \draw[arout] (L4) to (T);
    \draw[arout] (L5) to (C);
    \draw[arout] (L1) to (S);
    \draw[arout] (L1) to (C);
    \draw[arout] (L2) to (S);
    \draw[arout] (L2) to (L1);
    \draw[arout] (L3) to (L1);
    \draw[arout] (L3) to (C);
    \draw[arout] (C) to (L7);
    \draw[arout] (L6) to (S);
    \node at (3,-1) {$G^5$};
\end{scope}
\end{tikzpicture}
\caption{$G$ effectively abstracts $G^i$ for $i=1,\ldots, 5,\ldots$}
\label{fig:marg_abstract}
\end{figure}
\paragraph*{Motivating questions} Selection bias is ubiquitous, often latent, and can lead to biased results; therefore, not taking it into account may lead to an untrustworthy model. Unfortunately, marginalization is not able to deal with latent selection bias (\cf Example~\ref{ex:car_mechanic}).
Considering this, the following questions arise naturally:

\begin{enumerate}[label=Q\arabic*, ref=Q\arabic*]

    \item\label{Q1} Given an SCM $M$ with a selection mechanism $X_S\in \Si$, can we always find an SCM without a selection mechanism to faithfully represent $(M,X_S\in \Si)$?\footnote{Note that this question is trickier than it seems. Finding an SCM without a selection mechanism to represent another SCM with a selection mechanism is, to some extent, analogous to the problem of finding a DAG to represent the marginalized model of another DAG, which is impossible in general \citep{richardson2002ancestral}. See also \citet{Blom19beyond} for some data-generating processes with clear causal interpretation that cannot be modeled by SCMs.}  (\cf Appendix \ref{app:impossible})

    \item\label{Q2} If not, which part of the causal semantics of $(M,X_S\in \Si)$ can be represented by an SCM in general? Can we construct a transformation that transforms $(M,X_S\in \Si)$ into an SCM $M_{|X_S\in \Si}$ so that $M_{|X_S\in \Si}$ encodes this part of the causal semantics? What properties does it have? (\cf \cref{thm:causal_semantics,def:cdSCM,sec:3.2,prop:impossible})

    \item\label{Q3} Can we similarly construct a transformation on causal graphs such that it is compatible with the transformation at the level of SCMs? What properties does it have and what is the relation between the ``conditioned SCM'' $M_{|X_S\in \Si}$ and the ``conditioned causal graph'' $G(M)_{|S}$ (Figure~\ref{fig:diagram2})?  (\cf Definition~\ref{def:cond_dmg}, Section~\ref{sec:condDMG})
\end{enumerate}
We answer these questions in the current manuscript.
\begin{figure}[ht]
    \centering
\begin{tikzcd}[cells={nodes={scale=0.9}}]
 X_A \underset{\Prb_M(X_{V})}{\ind} X_B\mid X_C, X_S\in \Si \arrow[d,"?"', Leftrightarrow] \arrow[r, "?"', Rightarrow, shift right] & A\underset{G(M)}{\overset{d/\sigma}{\perp}}B\mid C\cup S \arrow[l, "?"', Rightarrow, shift right]  \arrow[d,"?"', Leftrightarrow] \\
   X_A \underset{\Prb_{M_{|X_S\in \Si}}(X_{V\sm S})}{\ind} X_B\mid X_C  \arrow[r, "?"', Rightarrow, shift right] & A\underset{G(M)_{|S}}{\overset{d/\sigma}{\perp}}B\mid C \arrow[l, "?"', Rightarrow, shift right]
\end{tikzcd}
\caption{What are the relations between stochastic conditional independence in $M_{|X_S\in \Si}$ and graphical separation in $G(M)_{|S}$? The answer is shown in \cref{fig:cond‐diagram}.}
\label{fig:diagram2}
\end{figure}

% \subsection{How can SCMs Model Latent Selection Bias}

% If yes, how to construct such transformation and what properties do this transformation have? If no, is it possible to characterize the largest part of causal semantics that can be preserved?

% Indeed, . We can think of an SCM $M$ with endogenous variables $X_V$ taking values in $\Xc_V$ as a collection of distributions $\Pc=\{\Prb(X_{V\sm T}\mid \Do(X_T=x_T)):T\subseteq V, x_T\in \Xc_T\}$, which satisfies some constraints given by the SCM $M$. Once can find a \footnote{It is similar to the fact that classical probability theory cannot model quantum phenomena \citep{aspect81bell}.} The question we are asking is given $\Pc^s=\{\Prb(X_{V\sm (T\cup S)}\mid \Do(X_T=x_T),X_S\in \Si):T\subseteq V, x_T\in \Xc_T\}$ (Notation \ref{nota:prob}), i.e., a collection of distributions generated by an SCM $M$ with a selection mechanism $X_S\in \Si$, can we find another SCM $\tilde{M}$ compatible with $\Pc^s$?

\paragraph*{A motivating example}

To illustrate, we first discuss a toy example, demonstrating that marginalization is \emph{not} appropriate for abstracting away selection bias and how to obtain correct results without assuming any specific details about the latent selection mechanism.

\begin{example}[Car mechanic]\label{ex:car_mechanic}
A car starts successfully if its battery is charged and its start engine is operating.
Introduce latent binary endogenous variables $B_0$ (``battery''), $E_0$ (``start engine'') and $S_0$ (``car starts'') measured at time $t_0$ and observed variables $B_1,E_1$ and $S_1$ with a similar meaning for the same car but measured at time $t_1$ with $t_1>t_0$. We model this\footnote{For illustration, we assume such a simplified model. One can add more (endogenous or exogenous random) variables to the model.} by the following SCM $M$ and denote by $M^*$ its marginalized model on observed endogenous variables $B_1$, $E_1$, and $S_1$.

\begin{minipage}{0.48\linewidth}
\vspace{2 pt}
\begin{equation*}
M:\left\{\begin{array}{l}
U_B \sim \operatorname{Ber}(1-\delta),
U_E \sim \operatorname{Ber}(1-\epsilon), \\
B_0 = U_B, E_0 = U_E, S_0=B_0 \wedge E_0, \\
B_1=B_0, E_1=E_0, S_1=B_1 \wedge E_1,
\end{array}\right.
\end{equation*}
\vspace{2 pt}
\end{minipage}
\begin{minipage}{0.48\linewidth}
\vspace{2 pt}
\begin{equation*}
M^*:\left\{\begin{array}{l}
U_B \sim \operatorname{Ber}(1-\delta),
U_E \sim \operatorname{Ber}(1-\epsilon), \\
B_1=U_B, E_1=U_E, S_1=B_1 \wedge E_1,
\end{array}\right.
\end{equation*}
\vspace{2 pt}
\end{minipage}
where $U_B$ and $U_E$ are latent exogenous independent Bernoulli-distributed random variables with parameters $1-\delta$ and $1-\epsilon$. Their graphs are shown in Figure~\ref{intro}.
%{\color{gray}This model assumes independence of the battery status and the start engine status (which is not unreasonable to assume in the population of all cars).}

%  {\color{gray}Given values for $\delta, \epsilon$, the car mechanic can decide to take the action that most likely will fix the car (without dorng a detailed diagnosis), or make use of utilities (such as the cost or time spent on the reparation) to choose the actions that optimize the utilities.}

%{\color{gray}She wonders why $B_1$ and $E_1$ would be strongly dependent, but realizes that it is in line with her experience: if she observes that the battery is charged ($B_1=1$), then she can infer that the start engine must be broken ($E_1=0$), and similarly, if she observes that the start engine is operational ($E_1=1$), she knows that the battery must be empty ($B_1=0$).\footnote{In case she followed a course on causality, she may vaguely remember Reichenbach's principle, and wonder whether this dependence implies that there is a latent common cause of $B_1$ and $E_1$. According to model $M$, there is no such latent common cause, but this is just Berkson's paradox, or ``$M$-bias''.}}

The question is whether there exists an SCM with variables $B_1, E_1, S_1$ encoding the causal semantics of $M$ for the subpopulation of cars for which $S_0=0$.
Consider the SCM $\tilde{M}$, whose graph is depicted in Figure~\ref{intro}, given by
\\[0.5ex]
\begin{minipage}{0.5\linewidth}
\begin{equation*}
\tilde{M}:\left\{\begin{array}{l}
\left(U_B, U_E\right) \sim \tilde{\Prb}(U_B,U_E)  \\
B_1=U_B,
E_1=U_E,
S_1=B_1 \wedge E_1
\end{array}\right.
\end{equation*}
\end{minipage}
\begin{minipage}{0.5\linewidth}
\begin{tabularx}{0.9\textwidth} {
      >{\raggedright\arraybackslash}X
      | >{\centering\arraybackslash}X
      >{\centering\arraybackslash}X }
        \hline\\[-2.5ex]
        $\tilde{\Prb}(U_B,U_E)$ & $U_E=0$ & $U_E=1$ \\
        \hline
        $U_B=0$  & $\frac{\delta\epsilon}{\delta + (1-\delta)\epsilon}$  & $\frac{\delta(1-\epsilon)}{\delta + (1-\delta)\epsilon}$  \\
        $U_B=1$  & $\frac{(1-\delta)\epsilon}{\delta + (1-\delta)\epsilon}$  & $0$  \\
        \hline
    \end{tabularx}
\end{minipage}
\\[0.5ex]
As one can check,
\[
\begin{aligned}
    \Prb_{\tilde{M}}(B_1,E_1,S_1) &= \Prb_M(B_1,E_1,S_1 \mid S_0=0)\ne \Prb_{M^*}(B_1,E_1,S_1),\\
    \Prb_{\tilde{M}}(S_1=1 \mid \Do(B_1=1)) &=\Prb_M(S_1=1\mid \Do(B_1=1),S_0=0 )\ne \Prb_{M^*}(S_1=1\mid \Do(B_1=1)), \\
    \Prb_{\tilde{M}}(S_1=1 \mid \Do(E_1=1)) &=\Prb_M(S_1=1\mid \Do(E_1=1),S_0=0)\ne \Prb_{M^*}(S_1=1\mid \Do(E_1=1)).
\end{aligned}
\]
The car mechanic is only interested in cars that failed to start at an early time $t_0$ and are sent to the studio at a later time $t_1$. So, the car mechanic (who might not even be aware of the latent selection mechanism $S_0=0$) can still use an SCM as an accurate causal model to predict the effects of interventions on the subpopulation of cars that are of her concern.
Note that the marginalized model $M^*$ does not possess the correct causal semantics of the subpopulation. Furthermore, the graph $G(\tilde{M})$ correctly expresses that $B_1$ and $E_1$ might be dependent in the subpopulation (given $S_0=0$) via the $d$-separation criterion for acyclic directed mixed graphs \citep{richardson03markov_admg}, while the graph $G(M^*)$ wrongly claims that $B_1$ and $E_1$ are independent.
Therefore, $\tilde{M}$ effectively abstracts away irrelevant latent modeling details: (i) the latent variables $B_0$, $E_0$ and $S_0$, (ii) their causal mechanisms, and (iii) the explicit selection step on $S_0=0$. However, the marginalized model $M^*$ does not, which shows that marginalization alone cannot abstract latent selection mechanisms.

\begin{figure}[ht]
\centering
\begin{tikzpicture}[scale=0.8, transform shape]
  \begin{scope}[xshift=0]
    \node[ndlat] (B0) at (1,4.5) {$B_0$};
    \node[ndlat] (E0) at (3,4.5) {$E_0$};
    \node[ndout] (B1) at (1,3) {$B_1$};
    \node[ndout] (E1) at (3,3) {$E_1$};
    \node[ndsel] (S0) at (2,3.5) {$S_0$};
    \node[ndout] (S1) at (2,2) {$S_1$};
    \draw[arout] (B0) to (S0);
    \draw[arout] (E0) to (S0);
    \draw[arout] (B0) to (B1);
    \draw[arout] (E0) to (E1);
    \draw[arout] (B1) to (S1);
    \draw[arout] (E1) to (S1);
    \node at (2,1) {$G(M)$};
  \end{scope}

%  \begin{scope}[xshift=4cm]
%    \node[ndout] (B1) at (1,3) {$B_1$};
%    \node[ndout] (E1) at (3,3) {$E_1$};
%    \node[ndout] (S1) at (2,2) {$S_1$};
%    \draw[arout] (B1) to (S1);
%    \draw[arout] (E1) to (S1);
%    \draw[gray, line width=2mm, fill opacity=.30, draw opacity=.30] (0.5,1.5) edge (3.5,3.5);
%    \draw[gray, line width=2mm, fill opacity=.30, draw opacity=.30] (0.5,3.5) edge (3.5,1.5);
%    \node at (2,1) {$G(\tilde{M})$};
%  \end{scope}

  \begin{scope}[xshift=6cm]
    \node[ndout] (B1) at (1,3) {$B_1$};
    \node[ndout] (E1) at (3,3) {$E_1$};
    \node[ndout] (S1) at (2,2) {$S_1$};
    \draw[arout] (B1) to (S1);
    \draw[arout] (E1) to (S1);
    \node at (2,1) {$G(M^*)$};
  \end{scope}

   \begin{scope}[xshift=12cm]
   \node[ndout] (B1) at (1,3) {$B_1$};
    \node[ndout] (E1) at (3,3) {$E_1$};
    \node[ndout] (S1) at (2,2) {$S_1$};
    \draw[arout] (B1) to (S1);
    \draw[arout] (E1) to (S1);
    \draw[arlat, bend left] (B1) to (E1);
    \node at (2,1) {$G(\tilde{M})=G(M)_{|S_0}=G(M_{|S_0=0})$};
  \end{scope}

\end{tikzpicture}
\caption{The causal graphs of the SCMs $M$, $M^*$ and $\tilde{M}$ in Example~\ref{ex:car_mechanic}. The gray nodes are latent and the triangle means conditioning on $S_0$ to take some specific values (\cf Notation~\ref{notation:graph}, Definition~\ref{def:SCM_selection}). Marginalizing out all the latent variables yields $M^*$, while conditioning out $S_0$ (\cf Definition~\ref{def:cdSCM}) and marginalizing out the remaining latent variables yields $\tilde{M}$.   }
\label{intro}
\end{figure}
\end{example}

%The augmented causal graph $G^a(M)$ of $M$ is given in Figure \ref{intro}.\footnote{See Definition 2.7 of \citet{bongers2021foundations} for the definitions of $G^a(M)$ and $G(M)$.} The goal is to compute causal effects of interventions on $B_1$ and $E_1$ to help make a repairing decision.
%As we shall show in Section \ref{s3}, a transformation $(\cdot)_{\cd(S_0=0)}: M\mapsto M_{\cd(S_0=0)}$ will give rise to an SCM
%\[
%M_{\cd(S_0=0)}:\left\{\begin{array}{l}
%\left(B_0, E_0\right) \sim \Prb\left(B_0, E_0\mid S_0=0\right) \\
%B_1=B_0 \\
%E_1=E_0 \\
%S_1=B_1 \wedge E_1.
%\end{array}\right.
%\]
%The causal graph $G(M_{\cd(S_0=0)})$ is shown in Figure \ref{intro}, which indicates correctly that $B_1$ and $E_1$ are dependent by the d-separation criterion.\footnote{Or m-separation criterion.} Also, performing casual reasoning in $M_{\cd(S_0=0)}$ will give us the correct result
%$\Prb_{M_{\cd(S_0=0)}}(S_1=1\mid \Do(B_1=1))=\frac{9}{14}$ and $\Prb_{M_{\cd(S_0=0)}}(S_1=1\mid \Do(E_1=1))=\frac{4}{14}$. This shows that marginalization is not enough for dorng model abstraction when there are latent selection present.

Note that in Example~\ref{ex:car_mechanic},  we can obtain the model $\tilde{M}$ directly from $M$ (\cf \cref{def:cdSCM,ex:car_mechanic_continue}), by
(i) merging $U_B$ and $U_E$, (ii) replacing $\Prb_M(U_B,U_E)$ by $\Prb_M(U_B,U_E \mid S_0=0)$, and (iii) marginalizing out $B_0, E_0$ and $S_0$ (by substituting the structural equations for $B_0, E_0$ and $S_0$ in the remaining structural equations and then removing these variables from the model). In fact, this procedure can be generalized and gives the desired transformation of Question \ref{Q2}, as we show in Section \ref{sec:cond}.

\subsection{Our contribution}

\citet[p.163]{pearl2009causality} claims that (when doing causal modeling):``...bidirected arcs should be assumed to exist, by default, between any two nodes in the diagram. They should be deleted only by well-motivated
justifications, such as the unlikely existence of a common cause for the two variables \emph{and the unlikely existence of selection bias}.'' Although marginalization makes it clear how bidirected edges can represent latent common causes, a rigorous approach to representing (latent) selection bias with bidirected edges has not been formalized yet.

Our main contribution is that we provide a rigorous approach to representing \emph{(latent) selection bias} with \emph{bidirected edges}. To be more precise, given a Structural Causal Model $M$ with a selection mechanism $X_S\in \Si$ (\cf Definition~\ref{def:SCM_selection}) where variable $X_S$ takes values in a measurable subset $\Si$, we define a transformation that maps $(M,X_S\in\Si)$ to a ``conditioned'' SCM $M_{|X_S\in\Si}$ without any accompanying selection mechanism, so that $M_{|X_S\in\Si}$ is an \emph{effective abstraction} of $M$ \wrt the selection $X_S\in\Si$ in the sense that:
\begin{enumerate}[label=(\roman*)]

    \item  the conditioned SCM $M_{|X_S\in\Si}$ correctly encodes as much \emph{causal semantics} (observational, interventional and counterfactual) of $M$ of the \emph{subpopulation} $X_S\in \Si$ as possible (\cf \cref{thm:causal_semantics,thm:as_much_as_possible,app:impossible});

    \item  the conditioning operation preserves important model classes, e.g., linear, acyclic and simple SCMs (\cf Proposition \ref{prop:modelcalss});

    \item this conditioning operation interacts well with other operations on SCMs, e.g., intervention, marginalization, and the conditioning operation itself  (\cf \cref{lem:int}, \cref{prop:iterative_cond,prop:cond_marg});

    \item  one can read off \emph{qualitative causal information} about $M$ under the selected subpopulation from the \emph{causal graph} of $M_{|X_S\in\Si}$ and the \emph{conditioned graph} $G(M)_{|S}$ (\cf Definition \ref{def:cond_dmg}, \cref{thm:dmarkov,thm:gmarkov}, Corollary \ref{cor:markov}).
\end{enumerate}

In Section \ref{sec:cond}, we will introduce the rigorous mathematical definition of the conditioning operation (Definition \ref{def:cdSCM}) and demonstrate that it possesses all the aforementioned properties.

The significance of this conditioning operation lies in the fact that we can take $M_{|X_S\in\Si}$ as a simplified ``proxy'' for $M$ \wrt the selection $X_S\in\Si$, which effectively abstracts away details about latent selection (i.e., satisfying the properties listed previously). This makes it a versatile tool for causal inference tasks with latent selection bias. Specifically:

\begin{enumerate}
    \item Causal Reasoning: One can directly apply all the causal inference tools for SCMs on $M_{|X_S\in\Si}$, e.g.,  identification results (adjustment criterion and Pearl's do-calculus), ID-algorithm and instrumental inequality, which simplifies causal reasoning tasks under latent selection bias (\cf \cref{ex:identif,ex:inst,ex:ID-algorithm}).

    \item Causal Modeling: Utilizing the marginalization and conditioning operation, we can represent infinitely many SCMs with a single marginalized conditioned SCM. This significantly streamlines causal modeling, eliminating the need to enumerate all possibilities with different latent structures and selection mechanisms. Moreover, it improves the robustness and trustworthiness of the model by reducing the sensitivity to various causal assumptions (\cf  \cref{ex:rei,ex:causalassum}).

    \item Causal Discovery: Many algorithms exclude selection bias by assumption—an often unrealistic idealization \citep{cooper95causal_discovery_selecion}. Nevertheless, methods originally designed for latent common causes without selection bias can, under suitable conditions, be applied directly to selected data, with their outputs interpreted as learning the conditioned model $M_{|X_S\in \Si}$ (\cf \cref{ex:discovery}). This requires no redesign of the algorithm: the unmodified procedure still admits certain causal interpretation of its output in the presence of selection.\footnote{Due to the abstraction nature of conditioned SCM, we lose some information in this process and see \cref{sec:caveat,thm:causal_semantics,thm:as_much_as_possible,app:impossible} for the subtlety of causal interpretation of conditioned SCM.}
\end{enumerate}

It is worth mentioning that many of our results rely on two facts: bidirected edges in DMGs can be used to represent latent selection bias, and DMGs admit an interpretation as causal graphs of SCMs. There is one subtlety, though: not all the endogenous variables of the conditioned SCM retain their causal interpretation. Interventions targeting such nodes yield predictions that are typically incompatible with those of the corresponding interventions on the original SCM in the presence of the selection mechanism (see also \cref{sec:caveat}). However, these ``non-intervenable" endogenous variables are easily identified as the ancestors of the selection variables.

\subsection{Connections to related work}

In a series of papers \citep{bareinboim2012controlling,bareinboim2015recovering}, the authors explored the `s-recoverability' problem, aiming to recover causal quantities for the whole population from selected data. This investigation operated under qualitative causal assumptions on the selection nodes, explicitly expressed in terms of causal graphs. However, such knowledge about selection nodes is not always available \citep[Footnote 11]{richardson2013single}. In the current work, we focus on the problem of how to \emph{model selection bias with an SCM without explicitly modeling the selection mechanism} and draw (causal) conclusions for the selected subpopulation.

There are graphical models with well-behaved marginalization and conditioning operations such as maximal ancestral graphs (MAGs) \citep{richardson2002ancestral}, $d$-connection graphs \citep{hyttinen2014constraint} and $\sigma$-connection graphs \citep{forre2018constraint}. Among them, MAGs were originally developed as a model class representing the conditional independence models of the marginalized conditioned conditional independence models of DAGs. By summarizing the common causal features of causal DAGs represented by a MAG, one can give a causal interpretation to MAGs and call them causal MAGs. One single causal MAG can represent infinitely many SCMs with different graphs but the same conditional independences among observed variables. Interpreting a graph as a causal graph of an SCM and as a causal MAG respectively will not give the same causal conclusions in general.\footnote{For example, consider a graph consisting of $A\tuh B\tuh C$ and $A\tuh C$. If it is a causal graph of an SCM, then we can conclude that variable $A$ has a direct causal effect on $C$ according to this model and we can identify $\Prb(C=c\mid \Do(A=a))=\Prb(C=c\mid A=a)$ under the positivity and discreteness assumption. However, if it is a MAG, then we cannot obtain the above two conclusions.}
Due to the nature of model abstraction, MAGs are well suited for causal discovery, and one can further draw \emph{some} causal conclusions from MAGs \citep{spirtes95causal, richardson03markov_admg, zhang2008causal, MooijClaassen20constraint}. However, MAGs are not always suitable for \emph{causal modeling} under selection bias in some cases, since: (i) it is not clear how to read off causal relationships (direct causal relations, confounding) from MAGs; (ii) there are no causal identification results for MAGs under selection bias and causal cycles yet; (iii) currently the standard theory of MAGs cannot deal with counterfactual reasoning. On the other hand, our conditioning operation transforms an SCM with selection mechanisms to an ordinary SCM, which carries an intuitive causal interpretation. All the theory for SCMs (causal identification, cycles, counterfactual reasoning) can be directly extended to the case with selection bias via the conditioning operation. Therefore, our results can address causal inference tasks such as fairness analysis  \citep{kusner2017counterfactual, zhang2018fairness}, causal modeling of dynamical systems \citep{bongers2018causal, peters2022causal} and biological systems with feedback loops \citep{versteeg2022local} under selection bias (cf. Definition~\ref{def:scm}, \cref{rem:back-door,rem:mag_instrumental}, \cref{ex:ID-algorithm,ex:mediation,ex:causalassum}). Another subtle difference between SCM conditioning and MAG conditioning is that they consider different forms of conditioning (cf. Example~\ref{ex:counterex_ci}).

Although causal graphs provide a means to differentiate selection bias from confounding due to common causes \citep{hernan2004structural,cooper95causal_discovery_selecion}, the potential outcome community tends to amalgamate the two \citep{richardson2013single,hernan2020causal}. In many cases, one can be sure about the existence of ``non-causal dependency'', but cannot be sure whether it is induced by a latent common cause or latent selection bias or the combination of the two (see e.g., \citet[Footnote~11]{richardson2013single} and \citet[p.163]{pearl2009causality}). Our conditioning operation formalizes this ambiguity within SCMs. Graphically, we employ bidirected edges to symbolize the dependence of two variables arising from either unmeasured common causes, latent selection bias, or any intricate combination of the two. Therefore, in causal modeling, our work allows the modeler to be able to represent such non-causal dependency abstractly via bidirected edges.

Some work considers the abstraction of causal models from the perspective of grouping low-level variables to high-level variables and merging values of variables \citep{rubenstein17,Bec2019abstract}. \citet{Geiger23CausalAF} study the so-called ``constructive abstraction'' of causal models. They show that it can be characterized as a composition of clustering sets of variables, merging values of variables, and marginalization. Our conditioning operation does not fall under the umbrella of ``constructive abstraction'' of \citet{Geiger23CausalAF}.

\subsection{Outline}

In Section~\ref{sec:pre}, we review basic notions of SCMs and fix the notation used throughout the article; additional preliminaries are deferred to \cref{app:def} to save space. In \cref{sec:sSCM}, we give a formal definition of SCMs with selection mechanisms. In \cref{sec:condSCM,sec:condDMG}, we introduce the conditioning operations for SCMs and DMGs, respectively, and study their mathematical properties and mutual relations; \cref{thm:causal_semantics,thm:dsep,prop:cond_scm_dmg} contain the main results. In \cref{sec:caveat}, we discuss important caveats concerning the interpretation of conditioned SCMs. Further remarks and examples related to \cref{sec:cond} are collected in \cref{app:ex}, while all proofs of the results in \cref{sec:cond} are provided in \cref{app:pf}. 

We illustrate the applicability of the conditioning operation through a series of examples in \cref{sec:ex}, including generalized versions of Reichenbach’s principle, the back-door theorem, the ID-algorithm, instrumental variables, causal model learning, mediation analysis, and a real-world causal modeling exercise on COVID-19. In \cref{app:impossible}, we show that SCMs without selection mechanisms are, in general, not flexible enough to represent SCMs with selection mechanisms, which answers Questions~\ref{Q1}–\ref{Q3} in combination with the discussion in \cref{sec:cond}. Finally, in \cref{app:variants_cond,sec:cond_iSCM}, we explore alternative conditioning operations, including variants based on different decompositions of exogenous variables, a conditioning operation for causal Bayesian networks, and a conditioning operation for SCMs with exogenous non-stochastic input variables (\cf \cref{def:cdiSCM}).

\section{Preliminaries and notation}\label{sec:pre}
This section provides the necessary background on SCMs and introduces the notions of common cause and confounding. To save space, additional preliminaries on SCMs are deferred to \cref{app:def}. We follow the formal setup of \citet{bongers2021foundations}, which allows us to formulate the theory for “simple’’ SCMs, a class that includes acyclic SCMs as well as well-behaved cyclic SCMs. However, we also depart from \citet{bongers2021foundations} in several respects—for instance, we allow for non-intervenable variables and introduce new node types (dashed nodes and triangle nodes). In this section we define (i) SCMs, (ii) (hard) interventions, (iii) SCM solution functions, (iv) simple SCMs, (v) marginalization, and (vi) basic causal relationships such as common cause and confounding (\cf \cref{def:scm,def:scm_hard_intervention,def:scm_solution_functions,def:simple_SCM,def:marg,def:causal_relation}). We also fix notation for causal graphs and (conditional) interventional distributions (\cf \cref{notation:graph,nota:prob}).

\subsection{Structural Causal Model (SCM)}

\begin{definition}[Structural Causal Model]\label{def:scm}
    A \textbf{Structural Causal Model (SCM)} is a tuple $M=(V,W,\Xc,\Prb,f)$ such that
    \begin{enumerate}[itemsep=-.5ex,label=(\roman*)]
      \item   $V,W$ are disjoint finite sets of labels for the \textbf{endogenous variables} and the \textbf{latent exogenous random variables}, respectively;
      \item   the \textbf{state space} $\Xc = \prod_{i\in V \dcup W} \Xc_i$ is a product of standard measurable spaces $\Xc_i$;
      \item   the \textbf{exogenous distribution} $\Prb$ is a probability distribution on $\Xc_W$ that factorizes as a product $\Prb = \bigotimes_{w\in W} \Prb(X_w)$ of probability distributions $\Prb(X_w)$ on $\Xc_w$;
      \item   the \textbf{causal mechanism} is specified by the measurable mapping $f : \Xc \rightarrow \Xc_V$.
\end{enumerate}
\end{definition}

\begin{definition}[Hard intervention]\label{def:scm_hard_intervention}
  Given an SCM $M$, an intervention target $T \subseteq V $ and an intervention value $x_T \in \Xc_T$, we define the \textbf{intervened SCM} $$M_{\Do(X_T=x_T)} \coloneqq  (V,W,\Xc,\Prb,(f_{V\setminus T},x_T)).$$
\end{definition}

This replaces the targeted endogenous variables with specified values. In this work, we do \emph{not} assume that all the endogenous variables in an SCM can be intervened on, which deviates from the standard modeling assumption. If an endogenous variable is modeled as ``intervenable'', then we say that we model it as causal or that it has a causal interpretation.\footnote{Although some variables are modeled as ``non-intervenable'', one can mathematically define an intervention on them. However, one should be careful with the causal interpretation \citep{Pearl19interpretation,pearl2015conditioning}. Similar problems arise in the work of causal model abstraction such as \citet{rubenstein17} and \citet{Bec2019abstract}. An intervention on the ``high-level'' variables in the abstracted models may not correspond to a well-defined intervention on the ``low-level'' variables in the detailed models. One can keep track of the ``allowed intervention targets''  $\mathcal{I}\subseteq V$ and augment $M$ to $(M,\mathcal{I})$. Similarly, one can also encode the information about which variables are latent or not in the definition of an SCM. This would introduce four types of endogenous nodes, which makes the notation quite heavy. Note that mathematically they can often be treated equally and the difference comes only at the phase of modeling.  Therefore, we do not distinguish these nodes in the definition of SCMs and only mark them informally with different types of nodes in the causal graphs (\cf \cref{notation:graph}). Another method is to introduce the so-called regime indicators \citep{Dawid02influence,Dawid21decision} to indicate on the graphs which variables are causal and which are purely probabilistic. This usually makes causal graphs much more inflated and requires us to introduce a conditioning operation for SCMs with exogenous non-stochastic input variables (\cf \cref{def:cdiSCM}).  To ease notation and the reader's mental burden, we do not adopt this approach, either.} To avoid confusion, most of the time we will only consider interventions $\Do(X_T=x_T)$ with intervention target $T$ a subset of the intervenable nodes in $V$. We consider all exogenous random variables as non-intervenable. Other types of interventions can be defined, such as soft or stochastic ones \citep{Correa2020stochastic}.

Given an SCM $M$, one can define its causal graph $G(M)$ and its augmented causal graph $G^a(M)$ to give intuitive and compact graphical representations of the causal model (see \cref{def:gra}). One can read off useful causal information directly from the causal graph without knowing the details of the underlying SCM.

\begin{notation}[Causal graphs]\label{notation:graph}
  In all the causal graphs, we use gray nodes to represent latent variables. Dashed nodes represent non-intervenable variables, and solid nodes represent intervenable variables. Exogenous variables are assumed latent and non-intervenable, so they are marked gray and dashed. Triangle nodes mean that there are selection mechanisms conditioning on the corresponding variables to take some specific values.\footnote{A triangle looks like a funnel, which means that we filter out some samples based on the values of variable $X_S$.} We sometimes abuse the notation by identifying the label and random variables in causal graphs.
\end{notation}

\begin{definition}[Solution function of an SCM]\label{def:scm_solution_functions}
   Given an SCM $M$, we call a measurable mapping
    $g^S :\Xc_{V\sm S}\times\Xc_{W} \rightarrow \Xc_{S}$ a \textbf{solution function of $M$ \wrt $S\subseteq V$} if
    for $\Prb(X_W)$-\almostall $x_W\in\Xc_W$ and for all $x_{V\sm S}\in\Xc_{V\sm S}$,\footnote{The ordering of the two quantifiers does matter and cannot be changed in the definition. See e.g.,  \citet[Lemma~F.11]{bongers2021foundations} for more details.} one has that $g^S(x_{V\sm S},x_W)$ satisfies the structural equations for $S$, i.e.,
    \[
      g^S(x_{V\sm S},x_W)=f_{S}(x_{V\sm S},g^S(x_{V\sm S},x_W),x_W).
    \]
    When $S=V$, we denote $g^V$ by $g$, and call $g$ a \textbf{solution function of} $M$.
\end{definition}

\begin{definition}[Unique solvability]\label{def:scm_unique_solvability}
  An SCM $M$ is called \textbf{uniquely solvable} \textbf{w.r.t}.\ $S\subseteq V$ if it has a solution function \wrt $S$ that is \textbf{essentially unique} in the sense that if $g^{S}$ and $\tilde{g}^{S}$ both satisfy the structural equations for $S$, then for $\Prb(X_W)$-\almostall $x_W\in\Xc_W$ and for all
  $x_{V\sm S}\in\Xc_{V\sm S}$, one has $g^{S}(x_{V\sm S},x_W)=\tilde{g}^{S}(x_{V\sm S},x_W)$.
  If $M$ has an essentially unique solution function \wrt $V$, we call it  \textbf{uniquely solvable}.
\end{definition}

Note that a subset $S$ does not inherit the unique solvability from the unique solvability of any of its supersets in general \citep[Appendix B.2]{bongers2021foundations}.

\begin{definition}[Simple SCMs]\label{def:simple_SCM}
  An SCM $M$ is called a \textbf{simple SCM} if it is uniquely solvable \wrt each subset $S\subseteq V$.
\end{definition}

Note that \emph{all acyclic SCMs are simple} \citep[Proposition 3.4]{bongers2021foundations}. One benefit of introducing the class of simple SCMs is that it preserves the most convenient properties of acyclic SCMs but allows for weak cycles. We focus on simple SCMs in this work so that we can avoid mathematical technicalities and focus on conceptual issues (\cf \cref{ass,rem:assum}).

For a simple SCM, we can plug the solution function of one component into other parts of the model so that we can get a simple SCM that ``marginalizes'' it while preserving the causal semantics of the remaining variables \citep{bongers2021foundations}.

\begin{definition}[Marginalization]\label{def:marg}
  Let $M$ be a simple SCM, $L\subseteq V$, and $g^L$ be a solution function of $M$ \wrt $L$. Then we call $M_{\sm L}=(V\sm L,W,\Xc_{V\sm L}\times \Xc_W,\Prb,\tilde{f})$ with
  \[
    \tilde{f}(x_{V\sm L},x_W)=f_{V\sm L}(x_{V\sm L},g^L(x_{V\sm L},x_W),x_W)
  \]
  a \textbf{marginalization} of $M$ over $L$.
\end{definition}

For SCMs, one can introduce a hierarchy of equivalence relations. Observational, interventional, and counterfactual equivalence mean that two SCMs have the same observational, interventional, and counterfactual semantics, respectively (see \citet[Definitions 4.1, 4.3 and 4.5]{bongers2021foundations} or Definition~\ref{def:cio_equivalence}). Counterfactual equivalence is strictly stronger than interventional equivalence, and interventional equivalence is strictly stronger than observational equivalence. \emph{Equivalence of SCMs} is also an equivalence notion stronger than the three equivalence notions mentioned above.

\begin{definition}[Equivalence]\label{def:equiv} An SCM $M=\SCM$ is \textbf{equivalent} to an SCM $\tilde{M}=(V,W,\Xc,\Prb,\tilde{f})$ if for all $v\in V$, for $\Prb$-a.a.\ $x_W\in\Xc_W$ and for all $x_V\in\Xc_V$,
  \[
    x_v=f_v(x_V,x_W)\quad \Longleftrightarrow \quad x_v=\tilde{f}_v(x_V,x_W).
  \]
  If $M$ and $\tilde{M}$ are equivalent, we write $M\equiv\tilde{M}$.\footnote{For defining equivalence of SCMs, one does not need to assume that $M$ and $\tilde{M}$ have the same sets of exogenous nodes but only needs the two sets to be isomorphic. For simplicity, we do not specify this in detail.}
\end{definition} 

A simple SCM induces a collection of distributions that includes its observational distribution and interventional distributions. Besides, one can describe counterfactual semantics of an SCM by performing interventions in its twin SCM (see Definition \ref{def:twin}). Potential outcomes are also used to express counterfactual semantics \citep{hernan2020causal,rubin1974estimating}. We can define potential outcomes via simple SCMs \citep{bongers2021foundations}.\footnote{In most of the potential outcome literature, potential outcomes are taken as primitives and are not induced by an underlying SCM.}

\begin{definition}[Potential outcome]\label{def:po}
    Let $M=\SCM$ be a simple SCM, $T\subseteq V$ be a subset, and $x_T\in \Xc_T$ be a value. The potential outcome under the perfect intervention $\Do(X_T=x_T)$ is defined as $X_{V}(x_T)\coloneqq (g^{V\sm T}(x_T,X_W),x_T)$, where $g^{V\sm T}: \Xc_{T} \times \Xc_{W}\rightarrow \Xc_{V\sm T} $ is the (essentially unique) solution function of $M$ \wrt $V\sm T$ and $X_W$ is a (fixed) random variable such that $X_W\sim \Prb$.
\end{definition}

\begin{definition}[Potential-outcome equivalence]\label{def:po_equiv}
    We say two SCMs $M^1=(V,W^1,\Xc^1,\Prb^1,f^1)$ and $M^2=(V,W^2,\Xc^2,\Prb^2,f^2)$ are \textbf{potential-outcome equivalent} if $\Xc^1_V=\Xc^2_V$ and 
    \[
    \Prb_{M^1}(\{X_{V}(x_{T_i})\}_{1\le i\le n})=\Prb_{M^2}(\{X_{V}(x_{T_i})\}_{1\le i\le n})
    \] 
    for all $T_i\subseteq V$, all $x_{T_{i}}\in \Xc_{T_i}$ and $i=1,\ldots,n$.
\end{definition}

Now we present the notations of (conditional) interventional distributions that we use in the current manuscript.

\begin{notation}[(Conditional) Interventional distributions]\label{nota:prob}
We use $\Prb_M(X_V,X_W)$ to denote the unique probability distribution of $(X_V,X_W)$ induced by a simple SCM $M$. Let $S\subseteq V$, $O\coloneqq V\sm S$, and $T\subseteq V$ with $T\cap S=\emptyset$. For a measurable subseteq $\Sc\subseteq \Xc_S$, we use
\[
\begin{aligned}
    \Prb_M(X_{O\sm T}\mid \Do(X_T=x_T),X_S\in\Si)&\coloneqq \Prb_{M_{\Do(X_T=x_T)}}(X_{O\sm T}\mid X_S\in\Si)\\
    &\coloneqq
        \frac{\Prb_{M_{\Do(X_T=x_T)}}(X_{O\sm T},X_S\in\Si)}{\Prb_{M_{\Do(X_T=x_T)}}(X_S\in\Si)}
\end{aligned}
\]
to represent the probability distribution of $X_O$ when first intervening on $X_T=x_T$ and second conditioning on $X_S\in\Si$, assuming $\Prb_{M_{\Do(X_T=x_T)}}(X_S\in\Si)>0$.\footnote{``First'' and ``second'' here refer to the order of applying the operations on the SCM, which may not coincide with the chronological order of these operations in the data-generating process that the SCM is modeling.} Using the notation of potential outcomes, we have
\[
\Prb_M(X_{O\sm T}\mid \Do(X_T=x_T),X_S\in\Si)=\Prb_M(X_{O\sm T}(x_T)\mid X_S(x_T)\in\Si),
\]
which is not equal to $\Prb_M(X_{O\sm T}(x_T)\mid X_S\in\Si)$ in general if $T\cap \anc_{G(M)}(S)\ne \emptyset$. If $T=\emptyset$, then $\Prb_M(X_{O\sm T}\mid \Do(X_T=x_T),X_S\in\Si)=\Prb_M(X_{O}\mid X_S\in \Si)$ and $X_{O\sm T}(x_T)=X_O$.
\end{notation}

\subsection{Common cause and confounding}

``Confounder'', ``common cause'' and ``confounding'' have diverse and vague meanings in different literature \citep{vanderweele13defconfounder}. For conceptual clarity, we give formal definitions for these notions in the setting of acyclic SCMs.\footnote{Providing formal definitions of these concepts for cyclic models is an open research question and is not within the scope of the current manuscript.} 

\begin{definition}[Common cause and confounding]\label{def:causal_relation}
Let $M=\SCM$ be an acyclic SCM and $A,B,C\in V$ be distinct intervenable nodes.
\begin{enumerate}
    \item We say that $X_C$ is a \textbf{common cause} of $X_A$ and $X_B$ according to $M$ if there exist $x_A\in \Xc_A$, $x_B\in \Xc_B$, $x_C\in \Xc_C$, and $x_C^\prime\in \Xc_C$ such that
    \[
        \begin{aligned}
            &\Prb_M(X_A\mid \Do(X_C=x_C),\Do(X_B=x_B)) \ne \Prb_M(X_A\mid \Do(X_B=x_B)) \text{ and }\\
            & \Prb_M(X_B\mid \Do(X_C=x_C^\prime),\Do(X_A=x_A)) \ne \Prb_M(X_B\mid \Do(X_A=x_A)).
        \end{aligned}
    \]

    \item Assume that for all $x_B\in \Xc_B$, we have
    $\Prb_M(X_A)=\Prb_M(X_A\mid \Do(X_B=x_B)).$
    We say that there is \textbf{confounding bias} between $X_A$ and $X_B$, if there exists a measurable subset $\Ac\subseteq \Xc_A$ with $\Prb_M(X_A\in \Ac)>0$ such that for $x_A\in \Ac$
\[
    \Prb_M(X_B\mid \Do(X_A=x_A))\ne\Prb_M(X_B\mid X_A=x_A).
\]

% \item We say that $X_A$ and $X_B$ are \textbf{confounded due to selection bias} according to $M^{\Si}$ if there is a $d$- or $\sigma$-open walk containing $S$ as a collider between $A$ and $B$ in $G(M)$.

% \item We say that $X_A$ and $X_B$ are \textbf{confounded due to selection bias} according to $M^{\Si}$ if $X_A$ and $X_B$ are unconfounded according to $M$ but
% there exists measurable subsets $\Ac\subseteq \Xc_A$ and $\Bc\subseteq \Xc_B$ with $\Prb_M(X_A\in \Ac\mid X_S\in \Si)>0$ and $\Prb_M(X_B\in \Bc\mid X_S\in \Si)>0$ such that for $x_A\in \Ac$ and $x_B\in \Bc$
% \[
% \begin{aligned}
%     &\Prb_M(X_B\mid \Do(X_A=x_A),X_S\in\Si)\ne\Prb_M(X_B\mid X_A=x_A,X_S\in\Si),\\
%     &\text{and } \Prb_M(X_A\mid \Do(X_B=x_B),X_S\in\Si)\ne\Prb_M(X_A\mid X_B=x_B,X_S\in\Si).
% \end{aligned}
% \]
\end{enumerate}
\end{definition}

\begin{remark}\label{rem:conf}
    \begin{enumerate}
        \item Note that a common cause has to be an endogenous variable instead of an exogenous variable, because we treat exogenous variables as non-causal (non-intervenable). For example, $X_C$ is a common cause of $X_A$ and $X_B$ according to $M^{1}$ and $M^{2}$, but not according to $M^3$ and $M^3$, whose graphs are shown in Figure~\ref{fig:confounder}.

        \begin{figure}[ht]
\centering
\begin{tikzpicture}[scale=0.8, transform shape]
  \begin{scope}[xshift=0]
    \node[ndout] (C) at (1,3) {$C$};
    \node[ndout] (A) at (2,2) {$A$};
    \node[ndout] (B) at (0,2) {$B$};
    \draw[arout] (C) to (A);
    \draw[arout] (C) to (B);
    \node at (1,1) {$G(M^1)$};
  \end{scope}

  \begin{scope}[xshift=4cm]
    \node[ndlat] (C) at (1,3) {$C$};
    \node[ndout] (A) at (2,2) {$A$};
    \node[ndout] (B) at (0,2) {$B$};
    \draw[arout] (C) to (A);
    \draw[arout] (C) to (B);
    \node at (1,1) {$G(M^2)$};
  \end{scope}

   \begin{scope}[xshift=8cm]
    \node[ndexo] (C) at (1,3) {$C$};
    \node[ndout] (A) at (2,2) {$A$};
    \node[ndout] (B) at (0,2) {$B$};
    \draw[arout] (C) to (A);
    \draw[arout] (C) to (B);
    \node at (1,1) {$G(M^3)$};
  \end{scope}

  \begin{scope}[xshift=12cm]
    \node[ndash] (C) at (1,3) {$C$};
    \node[ndout] (A) at (2,2) {$A$};
    \node[ndout] (B) at (0,2) {$B$};
    \draw[arout] (C) to (A);
    \draw[arout] (C) to (B);
    \node at (1,1) {$G(M^4)$};
  \end{scope}

\end{tikzpicture}
\caption{The causal graphs of the SCMs $M^i$, $i=1,2,3,4$, where $X_C$ is a common cause of $X_A$ and $X_B$ according to $M^1$ and $M^2$ but not to $M^3$ and $M^4$. }
\label{fig:confounder}
\end{figure}

\item If $X_A$ and $X_B$ have confounding bias, then there must be a bidirected edge between $A$ and $B$ in $G(M_{\sm (V\sm \{A,B\})})$. The bidirected edge can come from either
\begin{enumerate}
    \item[(i)] a common cause $X_C$, or

    \item[(ii)] a non-intervenable variable $X_E$ such that there are directed paths from $E$ to $A$ and $B$ in $G(M)^a$, which do not intersect $B$ and $A$ respectively,\footnote{We shall show in the next section that this can represent selection bias. Therefore, a bidirected edge represents the possible existence of ``non-causal dependency'' between $X_A$ and $X_B$, which can arise from either common cause, selection bias, or any combination of the two, or in other ways.} or

    \item[(iii)] any combination of the above two items.
\end{enumerate}
If there is no bidirected edge between $A$ and $B$ in $G(M_{\sm (V\sm \{A,B\})})$, then by the do-calculus there is no confounding bias between $X_A$ and $X_B$ according to $M$, i.e., we have for a.a.\ $x_A\in \Xc_A$ (under the assumption that $B\notin \anc_{G(M_{\sm (V\sm \{A,B\})})}(A)$)
\[
    \Prb_M(X_B\mid \Do(X_A=x_A))=\Prb_M(X_B\mid X_A=x_A).
\]

\item In the potential outcome literature, the unconfoundedness assumption is usually stated as (for the special case of finite discrete outcome variables and binary ``treatments'') that $X_B$ does not cause $X_A$ (i.e.\ $X_A(x_B)=X_A$ for all $x_B$) and
\begin{equation}\label{eqn:po_unconfound}
    \forall x_A\in\{0,1\}: X_A\ind X_B(x_A).
\end{equation}
If we assume that there is an underlying acyclic SCM $M=\SCM$ inducing the potential outcomes $X_A$ and $X_B(x_A)$, then equation~\eqref{eqn:po_unconfound} (under a positivity assumption) implies
\begin{equation}\label{eqn:po_unconfound_1}
    \Prb_M(X_B\mid \Do(X_A=x_A))=\Prb_M(X_B\mid X_A=x_A) \text{ for all }x_A\in \{0,1\}.
\end{equation}
This means that $X_A$ and $X_B$ have no confounding bias according to $M$ in the sense of Definition~\ref{def:causal_relation}. See \cref{rem:pf} for a proof.
\end{enumerate}
\end{remark}

\section{Theory}\label{sec:cond}

In this section, we develop conditioning operations for both SCMs and DMGs. We first introduce s-SCMs, which explicitly encode selection mechanisms, and then define the corresponding conditioning operations and analyze their main properties: the induced causal semantics, closure of relevant model classes, commutation with marginalization, intervention, and further conditioning, the associated loss of information, graphical separation criteria and Markov properties, and the compatibility between the two conditioning operations. After the mathematical development, we conclude with several important caveats on the modeling side.

In the whole section, we make the following assumption (see \cref{rem:assum}):
\begin{assumption}\label{ass}
$M=\SCM$ is a simple SCM such that $\Prb_M(X_S\in\Si)>0$ for some $S\subseteq V$ and measurable subset $\Si\subseteq \Xc_S$.
\end{assumption}

\subsection{SCM with selection mechanism}\label{sec:sSCM}

First, we give a definition for SCMs with selection mechanisms.

\begin{definition}[SCMs with selection mechanism]\label{def:SCM_selection}
We call $M^\Si\coloneqq(M,X_S\in\Si)$ an \textbf{s-SCM} or \textbf{SCM with a selection mechanism}, where $M=\SCM$ is an SCM, $S\subseteq V$ is a subset of endogenous nodes, and $\Si\subseteq \Xc_S$ is a measurable subset. We call $M^\Si$ a simple s-SCM if $M$ is simple. The causal semantics of $M^\Si$ are defined as:\footnote{We do not specify the counterfactual semantics of s-SCMs.}
\begin{enumerate}
    \item Observational distribution:
    \[\Prb_{M^\Si}(X_{V})\coloneqq\Prb_M(X_{V}\mid X_S\in\Si);\]

    \item Interventional distributions: for $T\subseteq V\sm S$ and $x_T\in\Xc_{T}$ with $\Prb_M(X_S\in \Si\mid \Do(X_T=x_T))>0$, we define \[\Prb_{M^\Si}(X_{V\sm T}\mid \Do(X_T=x_T))\coloneqq\Prb_M(X_{V\sm T }\mid \Do(X_T=x_T) ,X_S\in\Si).\]
\end{enumerate}
\end{definition}

If $\Prb_M(X_S\in\Si)=1$, then $M^{\Si}=(M,X_S\in \Si)$ is observationally equivalent to $M$. We can draw a causal graph of an SCM with a selection mechanism by Definition \ref{def:gra}, and in addition use triangles to represent the nodes in $S$ like in Figure \ref{fig:marg_abstract}.

To gain some intuition, one can imagine a simple SCM as representing a data-generating process where the $i$-th sample is generated as follows: first sampling $X_w^{(i)}\sim \Prb(X_w)$ for each $w\in W$, and then using the solution function $g:\Xc_W\rightarrow \Xc_V$ to generate $X_v^{(i)}$ for each $v\in V$. An SCM with a selection mechanism is adding a rejection step to the above sampling procedure. More precisely, we have the following rejection sampler Algorithm~\ref{alg:sampler}. This sampler generates the observational distribution of $M^\Si$. To generate interventional distributions of $M^\Si$, one just needs to replace $M$ with the corresponding intervened submodel $M_{\Do(X_T=x_T)}$, which changes the solution function but leaves the other parts of the algorithm invariant.
\begin{algorithm}
\begin{algorithmic}
\caption{Sampler for an SCM $M$ with a selection mechanism $X_S\in\Si$}
\label{alg:sampler}
\Require $n \geq 1$
\State $i \gets 1$
\While{$i \leq n$}
        \For{each $w \in W$}
            \State sample $X_w^{(i)} \sim \Prb_M\left(X_w\right)$
        \EndFor
        \For{each $v \in V$}
            \State calculate $X_v^{(i)} \leftarrow g_v\left(X_W^{(i)}\right)$
        \EndFor
        \If{$X_S^{(i)} \in \Si$}
            \State output $X_{V}^{(i)}$
            \State $i \gets i + 1$
        \EndIf
\EndWhile
\end{algorithmic}
\end{algorithm}

\subsection{Conditioning operation for SCMs}\label{sec:condSCM}

\subsubsection{Definition}\label{sec:3.1}

Suppose that we condition on $X_S\in \Si$. Then, the conditioning operation can be divided into three steps:
\begin{enumerate}
  \item[(i)] merging exogenous variables that become dependent given the observation $X_{S}\in\Si$;
  \item[(ii)] updating the exogenous probability distribution $\Prb(X_W)$ to the posterior $\Prb_M(X_W\mid X_S\in \Si)$ given the observation $X_{S}\in\Si$;
  \item[(iii)] marginalizing out the selection variables $X_S$.\footnote{Selection variables are marginalized out because we consider latent selection.}
\end{enumerate}

Before giving a formal definition of the conditioned SCM, we discuss item (i). For the reason why we need to consider merging exogenous random variables, see \cref{app:variants_cond}. There is a ``finest'' partition of $W$ given $X_S\in \Si$:
\begin{restatable}[Finest partition]{lemma}{partition}\label{lem:partition}
    Let $\Pf_{\Si}$ denote the set of partitions $\Ic=\{I_1,\ldots, I_p\}$ of $W$ such that $\{X_{I_i}\}_{i=1}^p$ are mutually independent under $\tilde{\Prb}(X_W)=\Prb(X_W\mid X_S\in \Si)$. Then there exists $\Hc\in \Pf_{\Si}$ such that $\Hc$ is a finer partition than any other partition $\Ic \in \Pf_{\Si}$.
\end{restatable}

We now present the formal definition of the conditioned SCM.

% Define
% \begin{equation}\label{eqn:H_0}
%     \begin{aligned}
%     H_0\coloneqq \{ w\in W:\ \exists \Dc_w\subseteq \Xc_{W\sm w}\text{ s.t.\ } \Prb\left(X_W\in g_S^{-1}(\Si)\sd  (\Xc_{w}\times \Dc_w)\right)=0\}.
% \end{aligned}
% \end{equation}
% Write $H\coloneqq W\sm H_0$.

\begin{definition}[Conditioned SCM] \label{def:cdSCM}
Assume Assumption \ref{ass}.
Let $g:\Xc_W\rightarrow \Xc_V$ and $g^S:\Xc_{V\sm S}\times\Xc_W\rightarrow \Xc_{S}$ be the (essentially unique) solution functions of $M$ \wrt $V$ and $S$ respectively. We define the \textbf{conditioned SCM}
$M_{|X_S\in\Si}\coloneqq \left(\hat{V}, \hat{W}, \hat{\Xc}, \hat{\Prb}, \hat{f}\right)$ by:
\begin{enumerate}[label=(\roman*)]
  \item  $\hat{V}\coloneqq V\setminus S$;

  \item  $\hat{W}\coloneqq \{\hat{w}_1,\ldots, \hat{w}_n\}$ is the finest partition of $W$ such that $\Prb_M(X_W\mid X_S\in \Si)=\bigotimes_{i=1}^n\Prb_M(X_{\hat{w}_i}\mid X_S\in \Si)$;

  \item  $\hat{\Xc}\coloneqq \Xc_{\hat{V}}\times \Xc_{\hat{W}}\coloneqq \Xc_{\hat{V}} \times \bigtimes_{i=1}^n \Xc_{\hat{w}_i}$, where $\Xc_{\hat{w}_i}\coloneqq \bigtimes_{w\in \hat{w}_i}\Xc_w$;

  \item  $\hat{\Prb}\coloneqq \bigotimes_{i=1}^{n}\hat{\Prb}(X_{\hat{w}_i})$, where $\hat{\Prb}(X_{\hat{w}_i})\coloneqq \Prb_M(X_{\hat{w}_i}\mid X_S\in \Si)$;

  \item  $\hat{f}(x_{\hat{V}},x_{\hat{W}})
  \coloneqq f_{\hat{V}}(x_{\hat{V}}, g^S(x_{\hat{V}},
  x_{\hat{w}_1},\ldots, x_{\hat{w}_n}),x_{\hat{w}_1},\ldots, x_{\hat{w}_n}).$
\end{enumerate}
\end{definition}

It is easy to check that $M_{|X_S\in\Si}$ is indeed an SCM. We mark nodes in $\anc_{G(M)}(S)$ as non-intervenable in $M_{|X_S\in \Si}$.

\begin{remark}\label{rem:defcond}
\begin{enumerate}
  \item In Definition \ref{def:cdSCM}, $M_{|X_S\in\Si}$ actually depends on the choice of $g^S$, but different versions are equivalent (in the sense of Definition \ref{def:equiv}). Here we abuse terminology and call
  $M_{|X_S\in\Si}$ ``the conditioned SCM'' of $M$ given $X_S\in \Si$ rather
  than ``a conditioned SCM'', and implicitly work with equivalence classes of SCMs. Note that if $M$ and $\tilde{M}$ are equivalent, then $M_{|X_S\in\Si}$ is equivalent to $\tilde{M}_{|X_S\in\Si}$.

  \item The definition of $\hat{W}$ does not depend on the choice of $g$ but depends on $S$ and $\Si$. Note that if $\Prb_M(X_S\in \Si)=1$, then $\hat{W}\cong W$ and conditioning on the selection mechanism $X_S\in\Si$ reduces to marginalizing out $S$. 
  
  \item If $w\in W\sm \anc_{G^a(M)}(S)$ or $w\in W\sm \anc_{G^a(M_{\sm(V\sm S)})}(S)$, then there exists $\hat{w}_i$ such that $\hat{w}_i=\{w\}$. In other words, if node $w$ is not an ancestor of $S$ in $G^a(M)$ or is not an ancestor of $S$ in $G^a(M_{\sm(V\sm S)})$, then node $w$ is not merged with any other nodes in $W$. In these cases, one has $\hat{\Prb}(X_{\hat{w}_i})=\Prb(X_{w})$.

  \item Since marginalization preserves simplicity and $\Prb$-null sets are also $\Prb_M(X_W\mid X_S\in \Si)$-null sets, $M_{|X_S\in\Si}$ is simple (\cf Proposition \ref{prop:modelcalss}).
\end{enumerate}
\end{remark}

\begin{notation}
  We often denote $M_{|X_S\in\Si}$ by $M_{|\Si}$ if it is clear from the context that $\Si$ is a measurable subset in which the variable $X_S$ takes values.
\end{notation}

\begin{example}[Example~\ref{ex:car_mechanic} continued]\label{ex:car_mechanic_continue}
    We consider Example~\ref{ex:car_mechanic} in the Introduction. Let $M$ be the SCM in Example~\ref{ex:car_mechanic}. Then $\tilde{M}=(M_{|S_0=0})_{\sm \{B_0, E_0\}}=(M_{\sm \{B_0, E_0\}})_{|S_0=0}$.
\end{example}

\begin{example}[conditioning operation for SCMs]
Consider the following SCMs with nonzero real coefficients $a_i$ for $i=1,\ldots,6$ such that $a_5+a_6\ne 0$:

\begin{minipage}{0.5\linewidth}
\begin{equation*}
    M^1: \left\{
    \begin{aligned}
        X & = E_1 \sim \mathrm{Uni}([0,1]) \\
        Y & = E_2 \sim \mathrm{Uni}([0,1]) \\
        S & = X + Y \\
        Z_1 & = a_1 X + E_3 \\
        Z_2 & = a_2 Z_1 + E_4 \\
        Z_3 & = a_3 Z_1 + a_4 Z_2 + a_5 S + a_6 Y + E_5
    \end{aligned}
    \right.
\end{equation*}
\vspace{2 pt}
\end{minipage}
\begin{minipage}{0.5\linewidth}
\begin{equation*}
    M^2:  \left\{
    \begin{aligned}
        X & = E_1 \sim \mathrm{Uni}([0,1]) \\
        Y & = E_2 \sim \mathrm{Uni}([0,1]) \\
        S & = \Ibbm (X+Y\geq 0.8) \\
        Z_1 & = a_1 X + E_3 \\
        Z_2 & = a_2 Z_1 + E_4 \\
        Z_3 & = a_3 Z_1 + a_4 Z_2 + a_5 S + a_6 Y + E_5.
    \end{aligned}
    \right.
\end{equation*}
\vspace{2 pt}
\end{minipage}
With $\pr_i$ denoting the projection to the $i$-th coordinate and $D\coloneqq \{(x,y)\in[0,1]^2: x+y\geq 0.8\}$, we then have

\begin{minipage}{0.5\linewidth}
    \begin{equation*}
    M_{|S\geq 0.8}^1:\left\{\begin{aligned}
              E_{1,2}&\sim\Uni(D)  \\
              X&=\pr_1(E_{1,2}) \\
              Y&=\pr_2(E_{1,2})\\
              Z_1 & = a_1 X + E_3 \\
              Z_2 & = a_2 Z_1 + E_4 \\
              Z_3 & = a_3 Z_1 + a_4 Z_2 + a_5 X\\
              &\quad + (a_5+a_6) Y + E_5,
    \end{aligned}\right.
 \end{equation*}
 \vspace{2 pt}
\end{minipage}
\begin{minipage}{0.5\linewidth}
\begin{equation*}
    M_{|S=1}^2:\left\{\begin{aligned}
              E_{1,2}&\sim\Uni(D)  \\
              X&=\pr_1(E_{1,2}) \\
              Y&=\pr_2(E_{1,2})\\
              Z_1 & = a_1 X + E_3 \\
              Z_2 & = a_2 Z_1 + E_4 \\
              Z_3 & = a_3 Z_1 + a_4 Z_2 \\
              & \quad + a_5  + a_6 Y + E_5.
    \end{aligned}\right.
\end{equation*}
 \vspace{2 pt}
\end{minipage}
We draw the (augmented) causal graphs of $M^1$, $M^2$, $M^1_{|S\geq 0.8}$ and $M^2_{|S=1}$ as shown in Figure \ref{fig:cdscm}. Note that
\begin{enumerate}[label=(\roman*)]
    \item bidirected edges can not only represent latent common causes but also latent selection bias;
    \item given two SCMs with the same causal graphs, the conditioned SCMs can have different graphs.
\end{enumerate}

If one applies Definition~\ref{def:scm_hard_intervention} to perform interventions on ancestors of $S$ in $M^1_{|S\geq 0.8}$ and $M^2_{|S=1}$ (e.g., $X$ and $Y$ in Figure~\ref{fig:cdscm}), the interventional distribution will correspond to a \emph{counterfactual} distribution of $M$. For instance, $\Prb_{M^1_{|S\geq 0.8}}(Z_3\mid \Do(X=x))=\Prb_{M^1}(Z_3(x)\mid S\geq 0.8)\ne \Prb_{M^1}(Z_3\mid \Do(X=x), S\geq 0.8)$. See Theorem~\ref{thm:causal_semantics} and Section \ref{sec:caveat} for details. This is the reason why we mark the ancestors of $S$ as dashed (non-intervenable) in $G^a(M^1_{|S\geq 0.8})$, $G^a(M^2_{|S=1})$, $G(M^1_{|S\geq 0.8})$ and $G(M^2_{|S=1})$.

\begin{figure}[ht]
\centering
\begin{tikzpicture}[scale=0.8, transform shape]
  \begin{scope}[xshift=0]
    \node[ndexo] (E1) at (1,4.1) {$E_1$};
    \node[ndexo] (E2) at (3,4.1) {$E_2$};
    \node[ndexo] (E3) at (0.3,2.2) {$E_3$};
    \node[ndexo] (E4) at (3.7,2.2) {$E_5$};
    \node[ndexo] (E5) at (0.5,0) {$E_4$};
    \node[ndout] (X) at (1,2.8) {$X$};
    \node[ndout] (Y) at (3,2.8) {$Y$};
    \node[ndout] (S) at (2,1.8) {$S$};
    \node[ndout] (Z1) at (1,1) {$Z_1$};
    \node[ndout] (Z2) at (2,0) {$Z_2$};
    \node[ndout] (Z3) at (3,1) {$Z_3$};
    \draw[arout] (E1) to (X);
    \draw[arout] (E2) to (Y);
    \draw[arout] (E3) to (Z1);
    \draw[arout] (E4) to (Z3);
    \draw[arout] (E5) to (Z2);
    \draw[arout] (X) to (S);
    \draw[arout] (Y) to (S);
    \draw[arout] (S) to (Z3);
    \draw[arout] (X) to (Z1);
    \draw[arout] (Z1) to (Z3);
    \draw[arout] (Z2) to (Z3);
    \draw[arout] (Z1) to (Z2);
    \draw[arout] (Y) to (Z3);
    \node at (2,-1) {$G^a(M^1)=G^a(M^2)$};
  \end{scope}
  \begin{scope}[xshift=1cm]
    \node[ndexo] (E12) at (7,4.1) {$E_{\{1,2\}}$};
    \node[ndexo] (E3') at (5.3,2.2) {$E_3$};
    \node[ndexo] (E4') at (8.7,2.2) {$E_5$};
    \node[ndexo] (E5') at (5.5,0) {$E_4$};
    \node[ndash] (X') at (6,2.8) {$X$};
    \node[ndash] (Y') at (8,2.8) {$Y$};
    \node[ndout] (Z1') at (6,1) {$Z_1$};
    \node[ndout] (Z2') at (7,0) {$Z_2$};
    \node[ndout] (Z3') at (8,1) {$Z_3$};
    \draw[arout] (E12) to (X');
    \draw[arout] (E12) to (Y');
    \draw[arout] (E3') to (Z1');
    \draw[arout] (E4') to (Z3');
    \draw[arout] (E5') to (Z2');
    \draw[arout] (X') to (Z1');
    \draw[arout] (X') to (Z3');
    \draw[arout] (Z1') to (Z3');
    \draw[arout] (Z2') to (Z3');
    \draw[arout] (Z1') to (Z2');
    \draw[arout] (Y') to (Z3');
  \node at (7,-1) {$G^a(M_{|S\geq 0.8}^1)$};
  \end{scope}
  \begin{scope}[xshift=7cm]
    \node[ndexo] (E12) at (7,4.1) {$E_{\{1,2\}}$};
    \node[ndexo] (E3') at (5.3,2.2) {$E_3$};
    \node[ndexo] (E4') at (8.7,2.2) {$E_5$};
    \node[ndexo] (E5') at (5.5,0) {$E_4$};
    \node[ndash] (X') at (6,2.8) {$X$};
    \node[ndash] (Y') at (8,2.8) {$Y$};
    \node[ndout] (Z1') at (6,1) {$Z_1$};
    \node[ndout] (Z2') at (7,0) {$Z_2$};
    \node[ndout] (Z3') at (8,1) {$Z_3$};
    \draw[arout] (E12) to (X');
    \draw[arout] (E12) to (Y');
    \draw[arout] (E3') to (Z1');
    \draw[arout] (E4') to (Z3');
    \draw[arout] (X') to (Z1');
    \draw[arout] (Z1') to (Z3');
    \draw[arout] (Z2') to (Z3');
    \draw[arout] (Z1') to (Z2');
    \draw[arout] (Y') to (Z3');
    \draw[arout] (E5') to (Z2');
  \node at (7,-1) {$G^a(M_{|S=1}^2)$};
  \end{scope}

  \begin{scope}[xshift=0,yshift=-5cm]
    \node[ndout] (X) at (1,2.8) {$X$};
    \node[ndout] (Y) at (3,2.8) {$Y$};
    \node[ndout] (S) at (2,1.8) {$S$};
    \node[ndout] (Z1) at (1,1) {$Z_1$};
    \node[ndout] (Z2) at (2,0) {$Z_2$};
    \node[ndout] (Z3) at (3,1) {$Z_3$};
    \draw[arout] (X) to (S);
    \draw[arout] (Y) to (S);
    \draw[arout] (S) to (Z3);
    \draw[arout] (X) to (Z1);
    \draw[arout] (Z1) to (Z3);
    \draw[arout] (Z2) to (Z3);
    \draw[arout] (Z1) to (Z2);
    \draw[arout] (Y) to (Z3);
    \node at (2,-1) {$G(M^1)=G(M^2)$};
  \end{scope}

  \begin{scope}[xshift=1cm, yshift=-5cm]
    \node[ndash] (X') at (6,2.8) {$X$};
    \node[ndash] (Y') at (8,2.8) {$Y$};
    \node[ndout] (Z1') at (6,1) {$Z_1$};
    \node[ndout] (Z2') at (7,0) {$Z_2$};
    \node[ndout] (Z3') at (8,1) {$Z_3$};
    \draw[arout] (X') to (Z1');
    \draw[arout] (X') to (Z3');
    \draw[arout] (Z1') to (Z3');
    \draw[arout] (Z2') to (Z3');
    \draw[arout] (Z1') to (Z2');
    \draw[arout] (Y') to (Z3');
    \draw[arlat, bend left] (X') to (Y');
  \node at (7,-1) {$G(M_{|S\geq 0.8}^1)$};
  \end{scope}
  \begin{scope}[xshift=7cm, yshift=-5cm]
    \node[ndash] (X') at (6,2.8) {$X$};
    \node[ndash] (Y') at (8,2.8) {$Y$};
    \node[ndout] (Z1') at (6,1) {$Z_1$};
    \node[ndout] (Z2') at (7,0) {$Z_2$};
    \node[ndout] (Z3') at (8,1) {$Z_3$};
    \draw[arout] (X') to (Z1');
    \draw[arout] (Z1') to (Z3');
    \draw[arout] (Z2') to (Z3');
    \draw[arout] (Z1') to (Z2');
    \draw[arout] (Y') to (Z3');
    \draw[arlat, bend left] (X') to (Y');
  \node at (7,-1) {$G(M_{|S=1}^2)$};
  \end{scope}
\end{tikzpicture}
\caption{Graphical representation of conditioning on $S\geq 0.8$ and $S=1$ respectively in $M^1$ and $M^2$. First, merge the exogenous ancestors of $S$, i.e., $E_1$ and $E_2$, to obtain a merged node $E_{\{1,2\}}$. Then update the exogenous probability distribution $\Prb(E_1,E_2)$ to the posterior $\Prb_{M^1}(E_1,E_2\mid S\geq 0.8)$ and $\Prb_{M^2}(E_1,E_2\mid S=1)$. Finally, marginalize out the node $S$. After conditioning,  $X$ and $Y$ are dashed, since we mark them as non-intervenable.}
\label{fig:cdscm}
\end{figure}
\end{example}

\subsubsection{Properties}\label{sec:3.2}

We derive some mathematical properties of the conditioning operation. First, we note that the conditioning operation preserves the simplicity, linearity, and acyclicity of SCMs.

\begin{restatable}[Simple, acyclic, linear SCMs and conditioning]{proposition}{modelclass}\label{prop:modelcalss}
    If $M$ is a simple (\resp acyclic) SCM with conditioned SCM $M_{|X_S\in\Si}$, then the conditioned SCM $M_{|X_S\in\Si}$ is simple (\resp acyclic). If $M$ is also linear, then so is $M_{|X_S\in\Si}$.
\end{restatable}

This implies that opting for simple/acyclic/linear SCMs as a model class and performing model abstraction through the conditioning operation will consistently maintain one within the chosen model class. This convenience proves valuable in practical applications, where adherence to the specific model class is often desired.

\begin{remark}
     \citet[Proposition 8.2]{bongers2021foundations} show that the class of simple SCMs is closed under marginalization, perfect intervention, and the twinning operation. Here we show that the class of simple SCMs is also closed under the conditioning operation of Definition~\ref{def:cdSCM} (ignoring the subtlety that the ancestors of the conditioning variables became non-intervenable).
\end{remark}

In the causal inference community, it is well known that conditioning and intervention do not commute. For pedagogical purposes, we give an example where $T\cap \anc_{G(M)}(S)\ne\emptyset$ and the intervention and conditioning do not commute. See also, e.g., \citet{mathur24simple_graphical} for an example.

\begin{example}[Conditioning and intervention do not commute]\label{ex}
   Consider a linear SCM with $\Prb(E_T), \Prb(E_X), \Prb(E_Y)$ such that $\Prb_M(X=x)>0$ and $\Prb_M(X=x\mid \Do(T=t))>0$ for $t=0,1$ and structural equations
  \[
    M:\left\{\begin{aligned}
              T&=E_T, X=\alpha T+E_X, \\
              Y&=X+\beta T+E_Y.
    \end{aligned}\right.
  \]

  In $M$, if we first condition on $X=x$ and second intervene on $T$ (despite $T$ being considered non-intervenable as $T\in \anc_{G(M)}(X)$), then we have
  \[
  \begin{aligned}
    \mathrm{E}_{(M_{|X=x})_{\Do(T=1)}}[Y]-\mathrm{E}_{(M_{|X=x})_{\Do(T=0)}}[Y]=\alpha+\beta.
  \end{aligned}
  \]
 On the contrary, if we first intervene on $T$ and second condition on $X=x$, then we have
  \[
    \begin{aligned}
      \mathrm{E}_{(M_{\Do(T=1)})_{|X=x}}[Y]-\mathrm{E}_{(M_{\Do(T=0)})_{|X=x}}[Y]=\beta.
   \end{aligned}
  \]
   In general, conditioning and intervention do not commute.
\end{example}

Conditioning does commute with interventions on the non-ancestors of the conditioned variables.

\begin{restatable}[Conditioning and intervention]{lemma}{intcond}\label{lem:int}
Assume Assumption \ref{ass}. Let $T \subseteq V\setminus \anc_{G(M)}(S)$ and $x_T\in\Xc_T$. Then we have
\[
\left(M_{\Do(X_T=x_T)}\right)_{|X_S\in\Si}\equiv\left(M_{|X_S\in\Si}\right)_{\Do(X_T=x_T)}.
\]
\end{restatable}

Note that in the above lemma, the probability
\[
\Prb_{M_{\Do(X_T=x_T)}}(X_S\in\Si)=\Prb_M(X_S\in\Si)
\]
is well defined and strictly larger than zero.

The next presented theorem characterizes the causal semantics of conditioned SCMs in terms of the original SCM with selection mechanisms.

\begin{restatable}[Main result I: Causal semantics of conditioned SCMs]{theorem}{causalsemantics}\label{thm:causal_semantics}
  Assume Assumption \ref{ass} and write $O\coloneqq V\sm S$. Let $T_i\subseteq O$ and $x_{T_i}\in \Xc_{T_i}$ for $i=1,\ldots, n$. Then we have
   \[
    \Prb_{M_{|X_S\in \Si}}\big(\{X_{O\sm T_i}(x_{T_i})\}_{1\leq i\leq n}\big)=
    \Prb_M\big(\{X_{O\sm T_i}(x_{T_i})\}_{1\leq i\leq n}\mid X_S\in \Si\big).
   \]
\end{restatable}

By noticing that $\Prb_M(X_{O\sm T}(x_T))=\Prb_M(X_{O\sm T}\mid \Do(X_T=x_T))$ and $\Prb_M(X_{O\sm T}(x_T)\mid X_S\in \Si)=\Prb_M(X_{O\sm T}\mid \Do(X_T=x_T), X_S\in \Si)$ provided $T\cap \anc_{G(M)}(S)=\emptyset$, we have the following result.

\begin{corollary}\label{cor:causal_semantics}
Assume Assumption \ref{ass} and write $O\coloneqq V\sm S$. Then we have:
\begin{enumerate}
    \item \textbf{Observational}: $\Prb_{M_{|X_S\in\Si}}(X_O)=\Prb_M(X_O\mid X_S\in \Si)$.

    \item \textbf{Interventional}: Let $T\subseteq O$ be $T=T_1\dcup T_2$ such that $T_1\subseteq V\sm \anc_{G(M)}(S)$ and $T_2\subseteq  \anc_{G(M)}(S)\sm S$. For $x_T\in\Xc_T$,
    \[
    \Prb_{M_{|X_S\in\Si}}\left(X_{O\setminus T}\mid \Do(X_T=x_T)\right)=
     \Prb_M\left(X_{O\setminus T}(x_{T_2})\mid \Do(X_{T_1}=x_{T_1}),X_S\in \Si\right).
    \]

    \item \textbf{Counterfactual via twinning}: Let $T\subseteq O$ be $T=T_1\dcup T_2$ such that $T_1\subseteq V\setminus \anc_{G(M)}(S)$ and $T_2\subseteq \anc_{G(M)}(S)\sm S$. Let $\tilde{T}\subseteq V^\prime$ be $\tilde{T}=T_3\dcup T_4$ such that $T_{3}\subseteq (V \setminus \anc_{G(M)}(S))^\prime$ and $T_4\subseteq (\anc_{G(M)}(S)\sm S)^\prime$.\footnote{$V^\prime$ means a copy of $V$. See Definition \ref{def:twin}.} For any $x_{T}\in\Xc_{T}$ and $x_{\tilde{T}}\in\Xc_{\tilde{T}}$,
    \[
      \begin{aligned}
        &\Prb_{\left(M_{|X_S\in\Si}\right)^{\twin}}(X_{(O\cup O^\prime)\setminus(T\cup \tilde{T})}\mid\Do(X_{T}=x_{T},X_{\tilde{T}}=x_{\tilde{T}}))\\
        &=\Prb_{M^{\twin}}(X_{O\setminus T}(x_{T_2}), X_{O^\prime\setminus \tilde{T}}(x_{T_4})\mid\Do(X_{T_1}=x_{T_1},X_{T_3}=x_{T_3}),X_S\in\Si).
    \end{aligned}
   \]
\end{enumerate}
\end{corollary}

We see that the conditioned SCM faithfully encapsulates the observational distribution of every observed endogenous variable under selection $X_S\in \Si$, and (taking $T_2=T_4=\emptyset$) the causal semantics of the non-ancestors of $S$ under selection $X_S\in \Si$, in accordance with the original SCM.  Therefore, the simplified abstracted model yields identical results as the original more intricate model with the selection mechanism $X_S\in \Si$ as long as one does not consider $\Prb_M(X_{O\sm T}\mid \Do(X_T=x_T),X_S\in \Si)$ where $T\cap \anc_{G(M)}(S)\ne \emptyset$.

One may wonder if it is possible to modify the definition of $M_{|X_S\in \Si}$ in such a way that we have, e.g., $\Prb_{M_{|X_S\in\Si}}\left(X_{O\setminus T}\mid \Do(X_T=x_T)\right)= \Prb_M\left(X_{O\setminus T}\mid \Do(X_{T}=x_{T}),X_S\in \Si\right)$ also for some $T\subseteq O$ such that $T\cap \anc_{G(M)}(S)\ne \emptyset$. We show in \cref{app:impossible} that it is impossible to find an SCM that preserves the causal semantics of the ancestors of latent selection variables in general. 

\begin{remark}[Not all data-generating processes with causal interpretation can be modeled by SCMs]
    \cref{prop:impossible} shows that there exist simple s-SCMs that cannot be modeled by simple SCMs. It suggests that not all data-generating processes with certain causal interpretations can be modeled by SCMs. Such examples also exist in the equilibrium behavior of dynamical systems and functional laws in physics and chemistry \citep{Blom19beyond}. 
\end{remark}

In addition, by \cref{thm:causal_semantics}, we have the following simple result. The conditional operation is an operation with a constructive definition that preserves as much causal information as possible.

\begin{restatable}[Conditioning operation preserves as much causal information as possible]{theorem}{optimal}\label{thm:as_much_as_possible} 
    There are no mappings $(M,X_S\in \Si)\mapsto \tilde{M}$ that preserve more causal information than $(M,X_S\in \Si)\mapsto M_{|X_S\in \Si}$ in the following sense. Assume that the mapping $(M,X_S\in \Si)\mapsto \tilde{M}$ is such that for all $T_i\subseteq V$ and $x_{T_i}\in\Xc_{T_i}$
    \[
        \Prb_{\tilde{M}}\big(\{X_{O\sm T_i}(x_{T_i})\}_{1\leq i\leq n}\big)=     \Prb_M\big(\{X_{O\sm T_i}(x_{T_i})\}_{1\leq i\leq n}\mid X_S\in \Si\big),
    \]
    and furthermore for some $T\subseteq \anc_{G(M)}(S)$ and $x_T\in \Xc_T$
    \[
        \Prb_{\tilde{M}}\left(X_{O\setminus T}\mid \Do(X_T=x_T)\right)=      \Prb_M\left(X_{O\setminus T}\mid \Do(X_{T}=x_{T}),X_S\in \Si\right). 
    \]
    Then it holds
    \[         
    \Prb_{M_{|X_S\in \Si}}\left(X_{O\setminus T}\mid \Do(X_T=x_T)\right)=      \Prb_M\left(X_{O\setminus T}\mid \Do(X_{T}=x_{T}),X_S\in \Si\right).     
    \]   
\end{restatable}

The subsequent result establishes the commutativity of conditioning and marginalization.

\begin{restatable}[Conditioning and marginalization commute]{proposition}{condmarg}\label{prop:cond_marg} 
    Assume Assumption~\ref{ass} and let $L\subseteq V\sm S$. Then we have $(M_{\setminus L})_{|X_S\in \Si}\equiv (M_{|X_S\in\Si})_{\setminus L}$.
\end{restatable}

Suppose that we have a selection mechanism $X_{S_1\cup S_2}\in\Si_1\times \Si_2$, then we can generate three ``versions'' of the conditioned SCMs:
\begin{enumerate}
    \item[(i)] applying the single conditioning operation \wrt $X_{S_1\cup S_2}\in\Si_1\times \Si_2$ to get $M_{|\Si_1\times \Si_2}$;

    \item[(ii)] first conditioning on $X_{S_1}\in \Si_1$ and second conditioning on $X_{S_2}\in \Si_2$ to get $(M_{|\Si_1})_{|\Si_2}$;

    \item[(iii)] first conditioning on $X_{S_2}\in \Si_2$ and second conditioning on $X_{S_1}\in \Si_1$ to get $(M_{|\Si_2})_{|\Si_1}$.
\end{enumerate}
The following proposition demonstrates that these three versions are counterfactually equivalent (and hence empirically indistinguishable).

\begin{restatable}[Conditioning and conditioning commute]{proposition}{iterativecond} \label{prop:iterative_cond}
    Assume Assumption \ref{ass} with $S=S_1\dcup S_2$ and $\Si=\Si_1\times \Si_2$ where $\Si_1\subseteq \Xc_{S_1}$ and $\Si_2\subseteq \Xc_{S_2}$ are both measurable. Then $(M_{|\Si_1})_{|\Si_2}, (M_{|\Si_2})_{|\Si_1}$, and $M_{|\Si_1\times\Si_2}$ are counterfactually equivalent and induce the same laws of potential outcomes. Also, $G(M_{|\Si_1\times \Si_2})$ is a subgraph of $G((M_{|\Si_1})_{|\Si_2})$ and $G((M_{|\Si_2})_{|\Si_1})$. Furthermore, if
    \begin{enumerate}
        \item [(i)] $\anc_{G^a(M_{\sm(V\sm S_{1})})}(S_1)\cap \anc_{G^a(M_{\sm(V\sm S_{2})})}(S_2)=\emptyset$, or

        \item [(ii)] we have
        \[
        \Prb\left(X_{W}\in \left(g_{S_1}^{-1}(\Si_1)\sd g_{S_2}^{-1}(\Si_2)\right)\right)=0,
        \]
    \end{enumerate}
    then $(M_{|\Si_1})_{|\Si_2}\equiv (M_{|\Si_2})_{|\Si_1}\equiv M_{|\Si_1\times \Si_2}$.
\end{restatable}

This result implies that the ordering of applying the conditioning operation does not matter up to counterfactual equivalence of SCMs. Therefore, there is essentially no ambiguity in referring to $M_{|X_S\in \Si}$ when the non-empty set $S$ is not a singleton. For marginalization \citep[Proposition 5.4]{bongers2021foundations}, a stronger property holds: marginalizing out variables in different orderings yields equivalent marginal SCMs. Overall, marginalization and conditioning commute both with each other and with themselves up to counterfactual equivalence. So, given a set of latent variables and latent selection mechanisms, irrespective of the intermediate steps taken, one consistently arrives at counterfactually indistinguishable models via marginalization and the conditioning operation. This underscores the robustness and reliability of the overall procedure for model abstraction.

\begin{example}[Iterative conditioning and joint conditioning]\label{ex:count_cond}
We give an example showing that applying the conditioning operation iteratively with different orders gives non-equivalent conditioned SCMs and joint conditioning gives a finer model than iterative conditioning does.

Consider an SCM
\[
    M:\left\{\begin{aligned}
              X_{w_1}&\sim \Uni\{1,2,3\}, X_{w_2}\sim \Uni\{1,2\}, \\
              X_{S_1}&=(X_{w_1},X_{w_2}), X_{S_2}=(X_{w_1},X_{w_2}).
    \end{aligned}\right.
\]
Take $\Si_1\coloneqq \{(2,1),(2,2),(3,1),(3,2)\}$ and $\Si_2\coloneqq \{(1,2),(2,1),(2,2)\}$. Write 
\[
    \begin{aligned}
    (M_{|\Si_1})_{|\Si_2}&=(V^{12},W^{12},\Xc^{12},f^{12},\Prb^{12})\\
    (M_{|\Si_2})_{|\Si_1}&=(V^{21},W^{21},\Xc^{21},f^{21},\Prb^{21})\\
    M_{|\Si_1\times\Si_2}&=(V^{1\times 2},W^{1\times 2},\Xc^{1\times 2},f^{1\times 2},\Prb^{1\times 2}).
    \end{aligned}
\]
Then we have $W^{12}=\{w_1,w_2\}=W$, $W^{21}=\{\{w_1,w_2\}\}$ (nodes $w_1$ and $w_2$ merge) and $\hat{W}^{1\times 2}=\{w_1,w_2\}=W$. Therefore, $(M_{|\Si_1})_{|\Si_2}\not\equiv (M_{|\Si_2})_{|\Si_1}$ and $(M_{|\Si_2})_{|\Si_1}\not\equiv M_{|\Si_1\times \Si_2}$, but $(M_{|\Si_1})_{|\Si_2}\equiv M_{|\Si_1\times \Si_2}$. In addition, $M_{|\Si_1\times \Si_2}$ is finer than $(M_{|\Si_2})_{|\Si_1}$.
\end{example}

\begin{remark}
The phenomenon in \cref{ex:count_cond} occurs since: two exogenous random variables ($X_{w_1}$ and $X_{w_2}$) merged due to conditioning on $X_{S_2}\in \Si_2$ may become independent after further conditioning on $X_{S_1}\in \Si_1$, while the conditioned SCM $M_{|X_{S_2}\in \Si_2}$ only records the information of the label set $\hat{W}$ where $w_1$ and $w_2$ become merged but forgets the original label set $W$ where $w_1$ and $w_2$ are distinct. Therefore, $w_1$ and $w_2$ are distinct in $M_{|\Si_1\times \Si_2}$ but merge in $(M_{|\Si_2})_{|\Si_1}$.

This problem can be fixed by modifying the definition of SCMs. For example, one can adapt the definition by equipping it with the information of a specific partition of $W$. We do not pursue this approach in the current manuscript and opt in \cref{def:cdSCM} for the ``joint'' version of the conditioning operation.
\end{remark}

Often, the selection nodes do not have children, as depicted in Figure \ref{sg1}. Even when the selection nodes have children, one can always create a copy of it, reducing it to the situation where the selection nodes lack children. For example, consider $G$ depicted in Figure \ref{sg2}, where selections occur on $A$ and $C$ by conditioning on a common effect $S_1$ and on $D$ by conditioning on $S_2$. Introducing a copy $\tilde{S}$ of $(S_1,S_2)$, or setting $\tilde{S}=\mathbbm{1}_{\{(S_1,S_2)\in \Si\}}$, yields the causal graph $\tilde{G}$, where conditioning is performed on $\tilde{S}$ instead of $(S_1,S_2)$. Consequently, in the causal graph $\tilde{G}$, the selection node $\tilde{S}$ does not have children. The following lemma validates this construction.

\begin{restatable}[Conditioning on binary variable without children]{lemma}{indicator}\label{lem:indicator}
Let $(M,X_S\in\Si)$ be a simple s-SCM and $(\tilde{M},X_{\tilde{S}}=1)$ be another simple s-SCM where $\tilde{M}\coloneqq (\tilde{V},W,\tilde{X},\Prb,\tilde{f})$ is such that $\tilde{V}=V\dcup\{ \tilde{S}$\}, $\tilde{\Xc}=\Xc\times \Xc_{\tilde{S}}\coloneqq\Xc\times \{0,1\}$ and $\tilde{f}_{\tilde{V}\sm \tilde{S}}(x_{\tilde{V}},x_W)=f_{V}(x_V,x_W)$ and $\tilde{f}_{\tilde{S}}(x_{\tilde{V}},x_W)=\Ibbm_{\Si}(x_S)$.  Then $M_{|X_S\in\Si}\equiv (\tilde{M}_{|X_{\tilde{S}}=1})_{\sm S}$.
\end{restatable}

\begin{figure}[ht]
\centering
\begin{tikzpicture}[scale=0.8, transform shape]
  \begin{scope}[xshift=0cm]
    \node[ndout] (C) at (2,4) {$C$};
    \node[ndout] (A) at (1,3) {$A$};
    \node[ndout] (B) at (3,3) {$B$};
    \node[ndsel] (S) at (2,2) {$S$};
    \draw[arout] (A) to (C);
    \draw[arout] (B) to (C);
    \draw[arout] (A) to (S);
    \draw[arout] (B) to (S);
    \draw[arout] (A) to (B);
    \node at (2,1) {$G_1$};
  \end{scope}
  \begin{scope}[xshift=5cm]
    \node[ndout] (C) at (2,4) {$C$};
    \node[ndout] (A) at (1,3) {$A$};
    \node[ndout] (B) at (3,3) {$B$};
    \node[ndsel] (S) at (1,2) {$S$};
    \draw[arout] (A) to (C);
    \draw[arout] (C) to (B);
    \draw[arout] (A) to (S);
    \draw[arout] (A) to (B);
    \node at (2,1) {$G_2$};
  \end{scope}
\end{tikzpicture}
\caption{Causal graphs representing selection on node $S$.}
\label{sg1}
\end{figure}

\begin{figure}[ht]
\centering
\begin{tikzpicture}[scale=0.8, transform shape]
  \begin{scope}[xshift=-1cm]
    \node[ndout] (C) at (2,4) {$C$};
    \node[ndout] (A) at (1,3) {$A$};
    \node[ndsel] (S1) at (3,3) {$S_1$};
    \node[ndout] (B) at (5,3) {$B$};
    \node[ndout] (D) at (0,4) {$D$};
    \node[ndsel] (S2) at (0,1.9) {$S_2$};
    \draw[arout] (C) to (B);
    \draw[arout] (C) to (S1);
    \draw[arout] (A) to (S1);
    \draw[arout] (S1) to (B);
    \draw[arout] (D) to (S2);
    \node at (3,1) {$G$};
  \end{scope}
  \begin{scope}[xshift=7cm]
    \node[ndout] (C) at (2,4) {$C$};
    \node[ndout] (A) at (1,3) {$A$};
    \node[ndout] (S1) at (3,3) {$S_1$};
    \node[ndout] (B) at (5,3) {$B$};
    \node[ndsel] (S') at (3,1.9) {$\tilde{S}$};
    \node[ndout] (D) at (0,4) {$D$};
    \node[ndout] (S2) at (0,1.9) {$S_2$};
    \draw[arout] (C) to (B);
    \draw[arout] (C) to (S1);
    \draw[arout] (A) to (S1);
    \draw[arout] (S1) to (B);
    \draw[arout] (D) to (S2);
    \draw[arout] (S1) to (S');
    \draw[arout] (S2) to (S');
    \node at (3,1) {$\tilde{G}$};
  \end{scope}
\end{tikzpicture}
\caption{Causal graphs representing selection on nodes $S_1$ and $S_2$, and on node $\tilde{S}$, respectively.}
\label{sg2}
\end{figure}

Based on the above observation, we give a conditioning operation for if we want to observe the selection variable. Let $M=\SCM$ be a simple SCM. Suppose that we want to condition on $X_S\in \Si$ for $S\subseteq V$ and $\Si\subseteq \Xc_S$ so that $\Prb_M(X_S\in \Si)>0$ but we still want to observe the values of $X_S$. Then a solution is that we introduce a selection variable $X_{\tilde{S}}$ such that $X_{\tilde{S}}=\Ibbm_{\Si}(X_S)$ as we did in Lemma~\ref{lem:indicator} and condition on $X_{\tilde{S}}=1$.

\subsection{Conditioning operation for DMGs}\label{sec:condDMG}

The conditioning operation (\cref{def:cdSCM}) is defined on simple SCMs. For causal modeling purposes, people often use causal graphs to communicate causal knowledge without referring to the precise underlying SCMs. To support this, we give a purely graphical conditioning operation defined on directed mixed graphs (DMGs). The idea is: (i) to add bidirected edges to node pairs that are ancestors of the conditioned nodes or siblings of ancestors of the conditioned nodes, and then (ii) to graphically marginalize the conditioned nodes out. 

\begin{definition}[Conditioned DMG]\label{def:cond_dmg}Let $G=(V, E^d, E^b)$ be a DMG consisting of nodes $V$, directed edges $E^d$ and bidirected edges $E^b$. For $S\subseteq V$, we define the conditioned DMG $G_{|S}$ by
  \begin{enumerate}
    \item adding bidirected edges to $G$: $\{a\huh b:a,b\in \anc_{G}(S)\cup \sib_{G}(\anc_{G}(S))\}$.\footnote{$\sib_G(v)\coloneqq \{w\in G\mid v\huh w \text{ is in } G\}$.}

    \item marginalizing out $S$ and marking ancestors of $S$ as dashed.
 \end{enumerate}
\end{definition}
The definition is inspired by the conditioning operation for SCMs. As we shall show, the purely graphical conditioning operation is compatible with the SCM conditioning operation.

\begin{example}[DMGs conditioning]\label{ex:cond_dmg}
    We show an example of the purely graphical conditioning operation. Assume that we are given a graph $G$ as shown in Figure \ref{fig:cond_dmg}. Then conditioning on node $V_5$ gives the graph $G_{|V_5}$ shown in Figure \ref{fig:cond_dmg}.
\begin{figure}[ht]
\centering
\begin{tikzpicture}[scale=0.8, transform shape]
  \begin{scope}[xshift=0cm]
    \node[ndout] (v1) at (4,4.5) {$V_1$};
    \node[ndout] (v2) at (6.5,1) {$V_6$};
    \node[ndout] (v3) at (2,2) {$V_8$};
    \node[ndout] (v4) at (3,3) {$V_2$};
    \node[ndout] (v5) at (5,3) {$V_3$};
    \node[ndout] (v6) at (7.5,3) {$V_4$};
    \node[ndout] (v7) at (4,2) {$V_5$};
    \node[ndout] (v8) at (5,1) {$V_7$};
    \draw[arout] (v1) to (v4);
    \draw[arout, bend left] (v2) to (v8);
    \draw[arout, bend left] (v8) to (v2);
    \draw[arout] (v4) to (v3);
    \draw[arout] (v4) to (v7);
    \draw[arout] (v5) to (v7);
    \draw[arout] (v7) to (v8);
    \draw[arout] (v1) to (v5);
    \draw[arout] (v6) to (v8);
    \draw[arout] (v5) to (v8);
    \draw[arlat, bend left] (v1) to (v6);
    \node at (5,0) {$G$};
  \end{scope}
  \begin{scope}[xshift=8cm]
    \node[ndash] (v1) at (4,4.5) {$V_1$};
    \node[ndout] (v2) at (6.5,1) {$V_6$};
    \node[ndash] (v4) at (3,3) {$V_2$};
    \node[ndash] (v5) at (5,3) {$V_3$};
    \node[ndout] (v6) at (7.5,3) {$V_4$};
    \node[ndout] (v8) at (5,1) {$V_7$};
    \node[ndout] (v3) at (2,2) {$V_8$};
    \draw[arout] (v4) to (v3);
    \draw[arout] (v4) to (v8);
    \draw[arout] (v1) to (v4);
    \draw[arout] (v1) to (v5);
    \draw[arout] (v6) to (v8);
    \draw[arout] (v5) to (v8);
    \draw[arout, bend left] (v2) to (v8);
    \draw[arout, bend left] (v8) to (v2);
    \draw[arlat, bend left] (v1) to (v6);
    \draw[arlat, bend left] (v5) to (v6);
    \draw[arlat, bend left] (v4) to (v6);
    \draw[arlat, bend left] (v4) to (v5);
    \draw[arlat, bend left] (v4) to (v1);
    \draw[arlat, bend left] (v1) to (v5);
    \node at (5,0) {$G_{|V_5}$};
  \end{scope}
\end{tikzpicture}
\caption{DMG $G$ and its conditioned DMG $G_{|V_5}$ in Example~\ref{ex:cond_dmg}.}
\label{fig:cond_dmg}
\end{figure}
\end{example}

The following result shows that the conditioning graph $G_{|S}$ represents $\sigma$-separations (also $d$-separations) encoded in the original graph $G$ soundly (but not completely in general). The notion of $\sigma$-separation (\cref{def:sigma_sep}) is defined for DMGs and reduces to normal $d$-separation (or $m$-separation) of acyclic DMGs  \citep{richardson03markov_admg} if there are no cycles \citep{Forre2017markov}. Note that when there are cycles, $\sigma$-separation implies $d$-separation but not the other direction. 

\begin{restatable}[Main result II: Graphical separation in conditioning graph]{theorem}{dsep}\label{thm:dsep}
  Let $G=(V, E^d, E^b)$ be a DMG and $S \subseteq V$ a set of nodes. Then for any subsets of nodes $A, B, C \subseteq V$ such that
  \[
    S \cap(A \cup B \cup C)=\emptyset,
  \]
  it holds that
  \[
    A \underset{G_{|S}}{\stackrel{\sigma/d}{\perp}} B\mid C \quad \Longrightarrow \quad A \underset{G}{\stackrel{\sigma/d}{\perp}} B\mid C\cup S.
  \]
  If furthermore $S$ is a singleton set with $\ch_{G}(S)=\emptyset$ and $C\cap\anc_G(S)=\emptyset$, then we have
  \[
    A \underset{G}{\stackrel{\sigma/d}{\perp}} B\mid C\cup S \quad \Longrightarrow \quad A \underset{G_{|S}}{\stackrel{\sigma/d}{\perp}} B\mid C .
  \]
\end{restatable}

\begin{remark}\label{rem:cond_dmg}
    If we only consider a singleton conditioning event $\Si=\{x_S\}$ (see Example \ref{ex:counterex_ci} and Lemma \ref{lem:ci_event} for the subtle difference), then we can define another version of the graphical conditioning operation transforming $G$ to $G_{\mathrm{cd}(S)}$ (``cd'' represents condition on). It is defined by $G_{\mathrm{cd}(S)}\coloneqq(G_{\ul{S}})_{|S}$, i.e., we first delete all the arrows emerging from $S$ and then apply the original conditioning operation to $G_{\ul{S}}$. It is easy to see that this new construction can strengthen the above result by removing the assumption that $\ch_G(S)=\emptyset$. This is similar in spirit to marking constant variables and deterministic dependencies distinctly in a causal graph to capture more conditional independence information escaping from the usual $d$-separation Markov property (see, e.g., \citet{geiger90identifying,spirtes2001causation}).
\end{remark}

The proof of the second claim in \cref{thm:dsep} relies on the three assumptions. See \cref{ex:counterex_sep} in \cref{app:ex}.

% \begin{corollary}\label{cor:separation}
% Let $G=(V, E^d, E^b)$ be a DMG and $S \subseteq V$ be a subset of nodes. 
% \begin{enumerate}
%     \item Let $\tilde{G}$ be a DMG constructed from $G$ by adding a child $\tilde{s}$ to each $s\in S$ and denote by $\tilde{S}$ the set of these nodes. Then we have for any $A,B,C\subseteq V$ such that $C\cap\anc_{G}(S)=\emptyset$:
%     \[
%         A \underset{G_{|S}}{\stackrel{\sigma/d}{\perp}} B\mid C\quad \Longleftrightarrow \quad  A \underset{\tilde{G}}{\stackrel{\sigma/d}{\perp}} B\mid C\cup \tilde{S}.
%     \]

%     \item Assume $\ch_{G}(S)=\emptyset$. Then for any subsets of nodes $A, B  \subseteq V$ such that
%   $S \cap(A \cup B )=\emptyset$, it holds that
%   \[
%     A \underset{G}{\stackrel{\sigma/d}{\perp}} B\mid  S \quad \iff \quad A \underset{G_{|S}}{\stackrel{\sigma/d}{\perp}} B.
%   \]
% \end{enumerate}
% \end{corollary}

\begin{restatable}[Graph conditioning commutes with marginalization, conditioning and intervention]{proposition}{propgraphcond}\label{prop:prop_graph_cond}
  Let $G=(V, E^d, E^b)$ be a DMG.
    \begin{enumerate}
        \item Let $L \subseteq  V$ and $S \subseteq V$ be two disjoint subsets of nodes from $G$. Then we have
        $$
        \left(G_{\sm L}\right)_{|S}=\left(G_{|S}\right)_{\sm L}.
        $$
        
        \item Let $S_1, S_2 \subseteq V$ be two disjoint subsets. Then we have
        $$
        \left(G_{| S_1}\right)_{| S_2}=\left(G_{|S_2}\right)_{| S_1}\subseteq G_{|S_1\cup S_2}.
        $$

        \item Let $T \subseteq  V$ and $S \subseteq V$ be two disjoint subsets of nodes from $G$ such that $T\cap \anc_G(S)=\emptyset$. Then we have
        $$
        \left(G_{\mathrm{do}\left(T\right)}\right)_{|S}=\left(G_{|S}\right)_{\operatorname{do}\left(T\right)}.
        $$
    \end{enumerate}
\end{restatable}

\begin{remark}\label{rem:graph_cond_cartesian}
      One should note that $\left(G_{| S_1}\right)_{| S_2}$ could be a strict subgraph of $G_{|\left(S_1 \cup S_2\right)}$. An example is shown in Figure~\ref{ce4}. When knowing that the conditioned set can be decomposed as a Cartesian product, one should use the iterative conditioned graph, which is finer than the jointly conditioned graph. See \cref{sec:graph_cond_cartesian}.
  \begin{figure}[ht]
  \centering
  \begin{tikzpicture}[scale=0.8, transform shape]
    \begin{scope}[xshift=0cm]
      \node[ndout] (C) at (2,4) {$C$};
      \node[ndout] (A) at (0,4) {$A$};
      \node[ndout] (B) at (3,4) {$B$};
      \node[ndout] (S1) at (1,2.5) {$S_1$};
      \node[ndout] (S2) at (3,2.5) {$S_2$};
      \draw[arout] (C) to (S1);
      \draw[arout] (B) to (S2);
      \draw[arout] (A) to (S1);
      \node at (2,1.5) {$G$};
    \end{scope}
    \begin{scope}[xshift=5cm]
      \node[ndash] (C) at (2,4) {$C$};
      \node[ndash] (A) at (0,4) {$A$};
      \node[ndash] (B) at (3,4) {$B$};
      \draw[arlat, bend left] (A) to (C);
      \node at (2,1.5) {$(G_{|S_1})_{|S_2}=(G_{|S_2})_{|S_1}$};
    \end{scope}
    \begin{scope}[xshift=10cm]
      \node[ndash] (C) at (1.5,4) {$C$};
      \node[ndash] (A) at (0,3) {$A$};
      \node[ndash] (B) at (3,3) {$B$};
      \draw[arlat, bend left] (A) to (C);
      \draw[arlat, bend left] (A) to (B);
      \draw[arlat, bend left] (C) to (B);
      \node at (1.5,1.5) {$G_{|(S_1\cup S_2)}$};
    \end{scope}
  \end{tikzpicture}
  \caption{$\left(G_{| S_1}\right)_{| S_2}=\left(G_{|S_2}\right)_{| S_1}\subsetneq G_{|\left(S_1 \cup S_2\right)}$}
  \label{ce4}
  \end{figure}
\end{remark}

The following proposition states that the purely graphical conditioning operation is compatible with the SCM conditioning operation. See \cref{rem:cond_scm_dmg} in \cref{app:ex} for some remarks on \cref{prop:cond_scm_dmg}.

\begin{restatable}[Main result III: DMG conditioning is compatible with SCM conditioning]{proposition}{condscmdmg}\label{prop:cond_scm_dmg}
  Let $M$ be a simple SCM with conditioned SCM $M_{|X_S\in\Si}$. Then $G(M_{|X_S\in\Si})$ is a subgraph of $G(M)_{|S}$. If furthermore $S=\{s_1,\ldots, s_n\}$ and $\Si=\bigtimes_{i=1}^n\Si_i$ with $\Si_i\subseteq \Xc_{s_i}$ measurable for $i=1,\ldots,n$, then $G(M_{|X_S\in\Si})$ is a subgraph of $((G(M)_{|s_1})_{\ldots})_{|s_n}$.
\end{restatable}

(Generalized) Directed global Markov properties connect the causal graph $G(M)$ and the induced distribution $\Prb_M(X_V)$ of the SCM $M$ in the sense that they enable one to read off conditional independence relations from the  graph via the $d$-separation (resp.\ $\sigma$-separation) criterion (\cref{def:dmarkov,def:gmarkov}). Obviously, $\Prb_{M_{|X_S\in \Si}}(X_O)$ satisfies the (generalized) directed Markov  property relative to $G(M_{|X_S\in \Si})$ (\cref{thm:dmarkov,thm:gmarkov}).  Therefore, the above proposition immediately implies that $\Prb_{M_{|X_S\in \Si}}(X_O)$ satisfies the Markov property relative to $G(M)_{|S}$ (as Corollary \ref{cor:markov} states), illustrating the role of the conditioned graph $G(M)_{|S}$ as an effective graphical abstraction. However, one should be aware of the subtlety that one \emph{cannot} directly conclude
\[
A \underset{G(M)}{\stackrel{\sigma/ d}{\perp}} B\mid C,S
\quad  \Longrightarrow \quad X_A\underset{\Prb_M(X_V)}{\ind} X_B\mid X_C,X_S\in \Si
\]
even if $\Prb_M(X_V)$ satisfies the Markov property relative to $G(M)$. See Example \ref{ex:counterex_ci} for details.

The following lemma establishes a connection between conditional independence given a variable and a family of conditional independencies given certain events.
\begin{restatable}{lemma}{cievent}\label{lem:ci_event}
    Let $X_A,X_B,X_C$ and $X_S$ be random variables defined on a probability space $(\Omega,\Fc,\Prb)$ and $X_S$ take values in a standard measurable space $(\Xc_S,\Bc_{\Xc_S})$. Then the first statement implies the second statement:
    \begin{enumerate}
        \item $X_A\ind X_B \mid X_C,X_S\in \Hc$ for all $\Hc\in \Bc_{\Xc_S}$ with positive probability.
        \item $X_A\ind X_B \mid X_C,X_S$.
    \end{enumerate}
\end{restatable}

This lemma incorporates some classical cases as special instances. For example, if $X_S$ has a countable support $\Xc_S$, then $X_A\ind X_B\mid X_S=s$ for all $s\in \Xc_S$ implies $X_A\ind X_B\mid X_S$ (the converse also holds).
The converse of the above lemma does not hold as Example \ref{ex:counterex_ci} shows. Also given one single $\Si\subseteq \Xc_S$ such that $\Prb(X_S\in \Si)>0$ and $X_A\ind X_B\mid X_C,X_S\in \Si$, we cannot infer $X_A\ind X_B\mid X_C,X_S$ in general.

Understanding these subtleties in conditional independence allows us to better appreciate the following Markov property, which is an easy corollary from \cref{prop:modelcalss,thm:gmarkov,thm:dmarkov,prop:cond_scm_dmg}.

\begin{corollary}\label{cor:markov}
    Let $M$ be a simple SCM with conditioned SCM $M_{|X_S\in \Si}$. Then the uniquely induced distribution $\Prb_{M_{|X_S\in \Si}}(X_O)$ satisfies the generalized directed global Markov property relative to $G(M)_{|S}$, i.e., for $A,B,C\subseteq O$ we have
    \[
        A\sep{\sigma}{G(M)_{|S}} B\mid C\quad \Longrightarrow \quad X_A\underset{\Prb_{M_{|X_S\in \Si}}(X_O)}{\ind} X_B\mid X_C.\footnote{\text{Recall that $X_A\underset{\Prb_{M_{|X_S\in \Si}}(X_O)}{\ind} X_B\mid X_C \quad \Longleftrightarrow \quad  X_A\underset{\Prb_{M}(X_V)}{\ind} X_B\mid X_C,X_S\in \Si$.}}
    \]
    Furthermore, assume one of the following conditions: (i) $M$ is acyclic; (ii) all endogenous state spaces $\mathcal{X}_v$ are discrete; (iii) $M_{|\Si}$ satisfies the third assumption in \cref{thm:dmarkov}. Then $\Prb_{M_{|X_S\in \Si}}(X_O)$ satisfies the directed global Markov property relative to $G(M)_{|S}$, i.e., 
    \[
        A\sep{d}{G(M)_{|S}} B\mid C\quad \Longrightarrow \quad X_A\underset{\Prb_{M_{|X_S\in \Si}}(X_O)}{\ind} X_B\mid X_C.
    \]
\end{corollary}

We can apply the intervention operation to $M_{|X_S\in \Si}$ and $G(M)_{|S}$, respectively. Recall that a simple SCM $M$ always satisfies the generalized Markov property w.r.t.\ $G(M)$. Therefore, we obtain that $(M_{|X_S\in \Si})_{\Do(X_T=x_T)}$ satisfies the generalized Markov property w.r.t.\ $(G(M)_{|S})_{\Do(T)}$ for any $T\subseteq O$.

% This allows us to compare the `do' distribution with `do-free' distribution and then derive do-calculus under latent selection.

The converse of (generalized) directed global Markov properties is $d$-faithfulness (resp.\ $\sigma$-faithfulness) \citep{spirtes2001causation,pearl2009causality,forre2018constraint}, which plays an important role in constraint-based causal discovery algorithms \citep{spirtes95causal,spirtes99alogorithm,MooijClaassen20constraint}. A natural question arises: how does faithfulness interact with the conditioning operation? Recall that marginalization preserves faithfulness. If $\Prb_M(X_V)$ is faithful to $G(M)$, then $\Prb_{M_{\sm L}}(X_{V\sm L})$ is faithful to $G(M)_{\sm L}$ for any $L\subseteq V$. However, this does not generally hold for conditioning. Even if $\Prb_M(X_V)$ is faithful to $G(M)$, the distribution $\Prb_{M_{|X_S\in \Si}}(X_{V\sm S})$ may not be faithful to $G(M)_{|S}$. See \cref{ex:unfaith} for a simple example.

\begin{example}[Conditioning does not preserve faithfulness]\label{ex:unfaith}
    Consider an SCM $M$ and its conditioned SCM $M_{|X_s=1}$
    \[
        M:\begin{cases}
            X_{w_1},X_{w_2}\sim \mathrm{Ber}(0.5),\\
            X_a=X_{w_1}, X_s=X_{w_2}\\
            X_b=\Ibbm(X_{w_2}=0)X_a,
        \end{cases}
        \qquad
        M_{|X_s=1}:\begin{cases}
            X_{w_1}\sim \mathrm{Ber}(0.5),X_{w_2}=1,\\
            X_a=X_{w_1},\\
            X_b=\Ibbm (X_{w_2}=0)X_a.
        \end{cases}
    \]
    It is easy to see that $X_a\underset{\Prb_{M_{|X_s=1}}(X_{V\sm \{s\}})}{\ind} X_b$. It holds that $\Prb_{M_{|X_s=1}}(X_{V\sm \{s\}})$ is faithful to $G(M_{|X_s=1})$ but not to $G(M)_{|\{s\}}$. 
\end{example}

% \begin{lemma}[Conditional independence given events or variables] \label{lem:ci} Let $(\Omega,\Fc,\Prb)$ be a probability space and $X,B,C$, and $U$ be random variables defined on it with values in standard measurable spaces $(\Xc,\Bc_{\Xc})$, $(\Yc,\Bc_{\Yc})$,   $(\Zc,\Bc_{\Zc})$, and $(\Uc,\Bc_{\Uc})$ respectively. Let $\Gc=\{H_n:n\in \Nb\}$ be a countable generator of the $\sigma$-algebra $\Bc_{\Uc}$ (which always exists in a standard measurable space). Then,
% \[
% (\forall H\in \Gc \text{ s.t.\ } \Prb(U\in H)>0, \ X\ind B\mid C, U\in H) \Longrightarrow X\ind B \mid C,U.
% \]
% \end{lemma}

% \begin{proof}
% % If $X\ind B\mid C, U$, then we have for every measurable subsets $A\subseteq \Xc$, $E^b\subseteq \Yc$ and for $\Prb(C)$-a.e.\ $z\in \Zc$ and a.e.\ $\omega\in C^{-1}(z)\cap U^{-1}(D)$
% % \[
% % \begin{aligned}
% %     \Prb\left((X, B)\in A\times E^b\mid C=z,U\in D\right)&=\Prb(X\in A,B\in E^b\mid C=z,U)(\omega)\\
% %     &=\Prb(X\in A\mid C=z,U)(\omega)\otimes\Prb(B\in E^b\mid C=z,U)(\omega)\\
% %     &=\Prb(X\in A\mid C=z,U\in D)\otimes\Prb(B\in E^b\mid C=z,U\in D).
% % \end{aligned}
% % \]
% % This shows the first statement.
% \end{proof}

Faithfulness is usually stated in terms of conditional independence given variables while the conditioned SCMs encode conditional independence information given selection variables taking values in some sets. \cref{ex:counterex_ci} in \cref{app:ex} shows the subtle difference between the two. However, \cref{lem:ci_event} builds the connection between them and allows us to state faithfulness in terms of conditional independence given events.

% \begin{proof}
%  % First note that
%  %    \[
%  %        \begin{aligned}
%  %            X_{A}\underset{\Prb_{M_{|X_S\in \Si}}(X_{O})}{\ind} X_B\mid X_C
%  %            &\Longleftrightarrow
%  %            X_{A}\underset{\Prb_{M}(X_{O}\mid X_S\in \Si)}{\ind} X_B\mid X_C &\text{(Theorem \ref{thm:causal_semantics})}\\
%  %            & \Longleftrightarrow
%  %             X_{A}\underset{\Prb_{M}(X_{O})}{\ind} X_B\mid X_C,X_S\in \Si.
%  %        \end{aligned}
%  %    \]
% By Lemma \ref{lem:ci_event} and the definition of conditional independence, the above conditions (i) and (ii) both imply that $X_{A}\underset{\Prb_{M}(X_{O})}{\ind} X_B\mid X_C,X_S$. We therefore obtain that
% \[
% \begin{aligned}
%     &  X_{A}\underset{\Prb_{M}(X_{O})}{\ind} X_B\mid X_C,X_S \\
%             \ & \Longrightarrow  \ A \underset{G(M)}{\stackrel{\sigma/d}{\perp}}  B\mid C, S &\text{(Faithfulness assumption)}\\
%             \ & \Longrightarrow  \ A\underset{G(M)_{|S}}{\stackrel{\sigma/d}{\perp}}  B\mid C. &\text{(Theorem \ref{thm:dsep})}
% \end{aligned}
% \]
% \end{proof}

Putting all the above results in Section~\ref{sec:cond} together gives us an answer to Questions~\ref{Q2} and \ref{Q3} in the Introduction. For the interaction between SCMs and causal graphs, assume that we have a simple SCM $M$ and only a subset $O$ of $V$ is observable. Denote $V \setminus O$ by $L\dcup S$ where $L$ denotes the latent part that is marginalized out and $S$ denotes the latent selection nodes. Fix a measurable set $\Si\subseteq \Xc_S$ with $\Prb_M(X_S\in \Si)>0$. Define the observable marginalized conditioned SCM
    \[
        M_{O|\Si}\coloneqq (M_{\sm L})_{|\Si}\equiv (M_{|\Si})_{\sm L},
    \]
    and observable marginalized conditioned graph
    \[
    G^{[O]}\coloneqq ((G(M))_{\setminus L})_{|S}=((G(M))_{|S})_{\setminus L},
    \]
    Then we have \cref{fig:cond‐diagram} (dashed arrows mean that the implications are only true under some extra conditions, and the numbers near the arrows correspond to theorems, lemmas, corollaries, and examples) for $A,B,C\subseteq O$. 
    \begin{figure}[ht]
  \centering
  \[
\begin{tikzcd}[ampersand replacement=\&, column sep=3.0em, row sep=2.2ex]
  % Row 1 (top)
  A\underset{G^{[O]}}{\overset{d/\sigma}{\perp}}B\mid C
    \&
  A\underset{G(M)}{\overset{d/\sigma}{\perp}}B\mid C,S
    \&
  {} \\
  % Row 2 (middle) -- centered mid node in the same (right) column
  {}
    \&
  |[name=mid]|
  X_A \underset{\Prb_M(X_{V})}{\ind} X_B\mid X_C,\, X_{S}
    \&
  {} \\
  % Row 3 (bottom) -- left & right items at the same height
  X_A \underset{\Prb_{M_{O|\Si}}(X_{O})}{\ind} X_B\mid X_C
    \&
   X_A \underset{\Prb_M(X_{V})}{\ind} X_B\mid X_C,\, X_{S}\in \Si
    \&
  {}
  % top-left <-> top-right
  \arrow[from=1-1, to=1-2, "\ref{thm:dsep}", Rightarrow, shift left]
  \arrow[from=1-1, to=1-2, dashed, "\ref{thm:dsep}"', Leftarrow, shift right]
  % top-left <-> bottom-left
  \arrow[from=1-1, to=3-1, Rightarrow,
       shift left=1.2ex,
       "{\scriptscriptstyle d/\sigma\text{-Markov}}+\ref{cor:markov}" ]
\arrow[from=1-2, to=mid, Rightarrow,
       shift left=1.2ex,
       "{\scriptscriptstyle d/\sigma\text{-Markov}}"{xshift=0.2em}]
\arrow[from=1-2, to=mid, Leftarrow,
       shift right=1.2ex, swap,
       "{\scriptscriptstyle d/\sigma\text{-Faith}}"{xshift=-0.2em}]
\arrow[from=mid, to=3-2, dashed, Rightarrow,
       shift left=1.2ex,
       "\text{(*)}"{xshift=0.2em}]
\arrow[from=mid, to=3-2, dashed, Leftarrow,
       shift right=1.2ex, swap,
       "{\scriptscriptstyle \ref{lem:ci_event}}"{xshift=-0.2em}]
  \arrow[from=3-1, to=3-2, " \ref{thm:causal_semantics}", Leftrightarrow]
\end{tikzcd}
\]
  \caption{Diagram relating graphical separation and stochastic independence under marginalization and conditioning for a simple SCM $M$.}
  \label{fig:cond‐diagram}
\end{figure}

(*): \cref{ex:counterex_ci} tells us that this implication does not hold in general, but if $\Si=\{x_s\}$ is a singleton set (recall that we assume $\Prb_M(X_S\in \Si)>0$) then this implication holds. 

The diagram shown in \cref{fig:cond‐diagram} still holds if we replace $G^{[O]}$ and $M_{O|\Si}$ with $(G^{[O]})_{\Do(T)}$ and $(M_{O|S})_{\Do(X_T=x_T)}$ respectively.\footnote{One should be careful with the causal interpretation when $T\nsubseteq O\sm\anc_{G(M)}(S)$. See \cref{sec:caveat} for more details.}

\begin{remark}
    To compare observational distributions and interventional distributions in one single graph and therefore derive the general measure-theoretic causal calculus rigorously, one may need the so-called exogenous input variables (or non-stochastic regime indicator variables according to \citet{Dawid21decision}) and transitional probability theory \citep{forre2021transitional,forre2020causal}. We discuss a conditioning operation for this more general class of models in Appendix~\ref{sec:cond_iSCM}. With this, we can generalize measure-theoretic causal calculus and other identification results to the case with latent selection bias via the conditioning operation (\cf Definition~\ref{def:cdiSCM}), and develop a commutative diagram similar to the one shown above but with exogenous non-stochastic input variables. Since causal information can be alternatively characterized by conditional independence involving regime indicators \citep{Dawid02influence,Dawid21decision}, the corresponding diagram for SCMs with input nodes gives us a clearer picture of what causal information can be preserved during the process of model abstraction via the conditioning operation.
\end{remark}

\subsubsection{Conditioning operation for DMGs: explicit modularity and locality}\label{sec:graph_cond_cartesian}

As mentioned in \cref{rem:graph_cond_cartesian}, when knowing that the conditioning set $\Si\subseteq \Xc_S$ can be decomposed as a Cartesian product $\Si=\bigtimes_{i=1}^n \Si_i$ with $\Si_i\subseteq \Xc_{s_i}$ and $S=\{s_1,\ldots,s_n\}$, we can obtain a finer conditioned graph by iterative conditioning than by joint conditioning. We give a formal definition.

\begin{definition}[Conditioned DMG: special case]\label{def:cond_dmg_special} Let $G$ be a DMG. For $S=\{s_1,\ldots,s_n\}\subseteq V$, we define the conditioned DMG $G_{|^\boxtimes S}$ by $((G_{|s_1})_{\ldots})_{|s_n}$.
\end{definition}

This definition is more in line with the principle that the SCM expresses the modular structure of causal mechanisms and selection mechanisms. This is particularly relevant when modeling physical systems where the locality principle of special relativity should be respected (both for causal mechanisms and selection mechanisms). Note that the definition of $G_{|^\boxtimes S}$ does not depend on the ordering of the iterative conditioning by \cref{prop:prop_graph_cond}. \cref{thm:dsep,prop:prop_graph_cond} all hold if we replace $(\cdot)_{|S}$ with $(\cdot)_{|^\boxtimes S}$. By \cref{prop:cond_scm_dmg}, we know that $G(M_{|X_S\in \Si})$ is a subgraph of $G(M)_{|^\boxtimes S}$ where $S=\{s_1,\ldots,s_n\}$ and $\Si=\bigtimes_{i=1}^n\Si_i$. Furthermore, we have $(G_{|^\boxtimes S_1})_{|^\boxtimes S_2}=(G_{|^\boxtimes S_2})_{|^\boxtimes S_1}=G_{|^\boxtimes (S_1\cup S_2)}$. If we introduce a common child $S^*$ of $s_1,\ldots,s_n$ and call the extended graph $G^*$, then $G_{|S}=(G^*)_{|S^*}=(G^*)_{|^\boxtimes S^*}\supseteq G_{|^\boxtimes S}$.

\subsection{Caveats on modeling interpretation}\label{sec:caveat}

In the previous subsections, we presented the SCM conditioning operation and DMG conditioning as purely mathematical operations and derived some mathematical properties of them. In this subsection, we make some remarks on how to interpret the conditioned SCMs appropriately to avoid confusion in modeling applications.

The subtleties are about intervening on ancestors of selection nodes. In this case, conditioning and interventions are not commutative, as we showed before. Therefore, one should be careful about the order of these two operations. On the one hand, if we first intervene and second condition on descendants of intervened variables, then the selected subpopulation will also change according to the intervention. On the other hand, first conditioning and second intervening on ancestors of selection nodes has a ``counterfactual flavor''. Suppose that an SCM $M$ with three variables $T$ (``treatment''), $Y$ (``outcome'') and $S$ (``selection") has a causal graph $T\tuh Y\tuh S$.  Intuitively, ``first-conditioning-second-intervening'' indicates that we first observe the results of the treatment and select units with specific values (say $S=s$) and fix this subpopulation. After that, we ``go back'' to perform an intervention (say $\Do(T=t)$) on this \emph{fixed selected subpopulation} instead of on the total population.\footnote{ This is often impossible to do in the real world where time travel is not an option, except if we can ``redo'' interventions while the exogenous variables remain invariant. Therefore, we prefer not to degrade from rung-2 causal queries to rung-3 ones.} Mathematically, we have
\[
\begin{aligned}
  \Prb_{\left((M_{|S=s})_{\Do(T=t)}\right)}(Y)&=\Prb_{M_{|S=s}}(Y\mid \Do(T=t))\\
  &=\Prb_{M_{|S=x}}(Y(t)) \\
  &=\Prb_M(Y(t)\mid S=s)\\
  &=\Prb_{M^{\mathrm{twin}}}(Y^{\prime}\mid \Do(T^\prime=t), S=s)\\
  &\ne \Prb_M(Y\mid \Do(T=t), S=s ), \quad \text{(in general)}\\
  &=\Prb_{\left((M_{\Do(T=t))_{|S=s}}\right)}(Y)
\end{aligned}
\]
where we used the language of potential outcomes and the twinning operation. In Pearl's terminology, this mixes different rungs: a rung-two query in the conditioned SCM is equivalent to a rung-three query in the original SCM. See also \citet{pearl2015conditioning} for an illustration. 

% \footnote{See \citet{nabi2022causal} for some interesting connections between missing data problems and counterfactuals.}

One can think of at least three possible ways to use the conditioning operation for modeling:
\begin{enumerate}
  \item[(i)] marginalize out all the ancestors of the selection nodes, or only consider cases where selection happens on exogenous random variables, so that there is no chance of being tempted to intervene on the ancestors of the selection nodes;
  \item[(ii)] specify in the conditioned SCM or in its graph which variables are ancestors of the selection nodes in the original SCM and do not apply interventions on them (which is what we opt for in this work);
  \item[(iii)] one can ignore the issue if one does not mind mixing up the rung-two quantities and rung-three quantities for her tasks, at the risk of introducing confusion about the causal interpretation of the conditioned SCM (which is not recommended). 
\end{enumerate}

See \cref{rem:model} in \cref{app:ex} for some further remarks on modeling interpretation.

\section{Applications}\label{sec:ex}

In this section, we illustrate several applications of the conditioning operation. The conditioning operation has a wide range of uses: all classical results for SCMs, such as identification results (back-door adjustment, do-calculus), apply directly to the conditioned SCMs $M_{|X_S\in\Si}$. Using the properties of the conditioning operation, these conclusions for $M_{|X_S\in\Si}$ can then be translated back to $(M,X_S\in\Si)$. In combination with marginalization, conditioning operation also provides a way to interpret a DMG as a causal graph that compactly encodes causal assumptions, with latent details of both latent common causes and latent selection abstracted away.

The examples in this section form a cohesive sequence. It navigates us from the philosophical implications of conditioning (a “generalized Reichenbach’s principle”), to the versatility of applying classical results to conditioned SCMs (back-door criterion, ID-algorithm, causal discovery, instrumental variables, mediation analysis), and finally to a concrete practical application of conditioned SCMs to modeling real-world problems (the COVID example).

\subsection{Reichenbach's principle under latent selection}

Reichenbach's Principle of Common Cause \citep{Reichenbach1956-REITDO-2} is often stated in this way: if two variables are dependent, then one must cause the other, or the variables must have a common cause (or any combination of these three possibilities). Note that this conclusion holds only when latent selection bias is ruled out, an assumption that is often left implicit.

\begin{example}[Reichenbach's principle]\label{ex:rei}
Using the conditioning operation, we can generalize and prove the principle under the framework of SCMs in the following way. Assume that $M$ is a simple SCM that has two observed endogenous variables $X$ and $Y$. By the Markov property (Theorem \ref{thm:gmarkov}), if $X$ and $Y$ are dependent, then $X \tuh Y$, $X \hut Y$, or $X \huh Y$ (or any combination of these three possibilities) are in the graph $G(M)$. There exist infinitely many SCMs $M^i$, $i\in I$ with an infinite index set $I$,
such that $(M^i_{\sm L_i})_{|\Si_i}=M$ where $L_i$ is a set of latent variables of $M^i$ and $X_{S_i}\in\Si_i$ is the latent selection in $M^i$. Hence, it implies that if two variables are dependent, then one causes the other, or the variables have a common cause, or are subject to latent selection (or any combination of these four possibilities).
\end{example}

\begin{remark}
    This provides one possible explanation for some real-world scenario in which one can exclude the possibilities of causal effects and common causes between two variables but can still observe the stochastic dependency between them.
\end{remark}

\subsection{Causal identification under latent selection}

\begin{example}[Back-door theorem]\label{ex:identif}

  Let $M^1$ and $M^2$ be two SCMs with three variables $T$ (``treatment''), $X$ (``covariates''), and $Y$ (``outcome'') whose causal graphs are shown in Figure \ref{fig_ex_identif}.  Under some assumptions, Pearl's Back-Door Theorem \citep{pearl2009causality} gives, for $i=1,2$, the identification result:\footnote{For simplicity, here we ignore the measure-theoretic subtlety. Indeed, we need to assume $\Prb_{M^i}(X)\otimes \Prb_{M^i}(T)\ll \Prb_{M^i}(X,T)$ and then the identity holds $\Prb_{M^i}(T)$-a.s. See \citet{forre25causality} for more details.}
  \begin{equation}\label{eq:bd}
      \Prb_{M^i}(Y\mid \Do(T=t))
      =\int\Prb_{M^i}(Y\mid X=x, T=t) \Prb_{M^i}(X\in\mathrm{d}x).
  \end{equation}
  \begin{figure}[ht]
  \centering
  \begin{tikzpicture}[scale=0.8, transform shape]
    \begin{scope}[xshift=0]
      \node[ndout] (T) at (1,1) {$T$};
      \node[ndout] (X) at (2,2) {$X$};
      \node[ndout] (Y) at (3,1) {$Y$};
      \draw[arout] (X) to (T);
      \draw[arout] (T) to (Y);
      \draw[arout] (X) to (Y);
      \node at (2,0) {$G(M^1)$};
    \end{scope}
    \begin{scope}[xshift=4cm]
      \node[ndout] (T) at (1,1) {$T$};
      \node[ndout] (X) at (2,2) {$X$};
      \node[ndout] (Y) at (3,1) {$Y$};
      \draw[arout] (X) to (T);
      \draw[arout] (T) to (Y);
      \draw[arout] (X) to (Y);
      \draw[arlat, bend left] (X) to (Y);
      \node at (2,0) {$G(M^2)$};
    \end{scope}

    \begin{scope}[xshift=8cm,yshift=1cm]
      \node[ndout] (T) at (1,1) {$T$};
      \node[ndout] (X) at (2,2) {$X$};
      \node[ndout] (Y) at (3,1) {$Y$};
      \node[ndlat] (L) at (2,0) {$L$};
      \draw[arout] (X) to (T);
      \draw[arout] (T) to (Y);
      \draw[arout] (X) to (Y);
      \draw[arout] (L) to (T);
      \draw[arout] (L) to (Y);
      \node at (2,-1) {$G(M^3)$};
    \end{scope}
  \end{tikzpicture}
  \caption{Causal graphs of SCMs $M^1$ and $M^2$ in Example \ref{ex:identif} and of $M^3$ in Remark \ref{rem:back-door}.}
  \label{fig_ex_identif}
  \end{figure}

  Thanks to marginalization and the conditioning operation, we can see $M^1$ and $M^2$ as abstractions of other SCMs, i.e., $M^i=(\tilde{M}^i_{\sm L^i})_{|\Si^i}$,
  for SCMs $\tilde{M}^i$, latent variables $L^i=\{L^i_1,\ldots,L^i_n\}$, and latent selection variables $S^i=\{S^i_1,\ldots,S^i_m\}$ taking values in measurable sets $\Si^i$ with $i=1,2$. For both $M^1$ and $M^2$, we present two examples $\tilde{M}^i_{(j)}$ for $j=1,2$, respectively, out of the infinite possibilities in Figure \ref{fig_ex_identif2}.

  With the help of Theorem \ref{thm:causal_semantics}, we can write \eqref{eq:bd} as
  \begin{equation}\label{eq:back_door}
      \Prb_{\tilde{M}^i}(Y\mid \Do(T=t),S^i\in\Si^i)=\int\Prb_{\tilde{M}^i}(Y\mid X=x, T=t, S^i\in\Si^i) \Prb_{\tilde{M}^i}(X\in\mathrm{d}x\mid S^i\in\Si^i).
  \end{equation}
  Thus, the back-door theorem can be applied directly to the conditioned SCM, which is useful especially if the specific latent structure of the SCM is \emph{unknown}. 

  \begin{figure}[ht]
  \centering
  \begin{tikzpicture}[scale=0.8, transform shape]
    \begin{scope}[xshift=0]
      \node[ndout] (T) at (1,1) {$T$};
      \node[ndout] (X) at (2,2) {$X$};
      \node[ndout] (Y) at (3,1) {$Y$};
      \node[ndlat] (L) at (2,3.5) {$L^1$};
      \node[ndsel] (S) at (3.5,3.5) {$S^1$};
      \draw[arout] (X) to (T);
      \draw[arout] (T) to (Y);
      \draw[arout] (X) to (Y);
      \draw[arout] (L) to (X);
      \draw[arout] (L) to (S);
      \node at (2,0) {$G(\tilde{M}^1_{(1)})$};
    \end{scope}
    \begin{scope}[xshift=4.5cm]
      \node[ndout] (T) at (1,1) {$T$};
      \node[ndout] (X) at (2.5,3) {$X$};
      \node[ndout] (Y) at (4,1) {$Y$};
      \node[ndlat] (l1) at (2.8,1.6) {$L_1^2$};
      \node[ndlat] (l2) at (5,2) {$L_2^2$};
      \node[ndlat] (l3) at (3.5,3.5) {$L_3^2$};
      \node[ndsel] (S) at (3.9,2.3) {$S^2$};
      \draw[arout] (X) to (T);
      \draw[arout] (T) to (Y);
      \draw[arout] (X) to (Y);
      \draw[arout] (l1) to (X);
      \draw[arout] (l1) to (Y);
      \draw[arout] (l3) to (X);
      \draw[arout] (l2) to (Y);
      \draw[arout] (l3) to (S);
      \draw[arout] (l2) to (S);
      \node at (3,0) {$G(\tilde{M}^2_{(1)})$};
    \end{scope}
    \begin{scope}[xshift=0,yshift=-5cm]
      \node[ndout] (T) at (1,1) {$T$};
      \node[ndout] (X) at (2,2) {$X$};
      \node[ndout] (Y) at (3,1) {$Y$};
      \node[ndlat] (l1) at (2,4) {$L_1^1$};
      \node[ndlat] (l2) at (0.5,2) {$L_2^1$};
      \node[ndsel] (s1) at (0.5,3.5) {$S_1^1$};
      \node[ndsel] (s2) at (3.5,3.5) {$S_2^1$};
      \node[ndsel] (s3) at (2,3) {$S_3^1$};
      \node[ndsel] (s4) at (3.5,2) {$S_4^1$};
      \draw[arout] (X) to (T);
      \draw[arout] (T) to (Y);
      \draw[arout] (X) to (Y);
      \draw[arout] (s1) to (X);
      \draw[arout] (s2) to (X);
      \draw[arout] (s4) to (X);
      \draw[arout] (l1) to (s1);
      \draw[arout] (l1) to (s2);
      \draw[arout] (s1) to (s3);
      \draw[arout] (s2) to (s3);
      \draw[arout] (s2) to (s4);
      \draw[arout] (l2) to (s1);
      \draw[arout] (l2) to (X);
      \node at (2,0) {$G(\tilde{M}^1_{(2)})$};
    \end{scope}
    \begin{scope}[xshift=4.5cm,yshift=-5cm]
        \node[ndout] (T) at (1,1) {$T$};
        \node[ndout] (X) at (2.5,2.5) {$X$};
        \node[ndout] (Y) at (4,1) {$Y$};
        \node[ndlat] (l1) at (5,3.5) {$L_1^2$};
        \node[ndlat] (l2) at (5,2) {$L_2^2$};
        \node[ndlat] (l3) at (3.4,3.4) {$L_3^2$};
        \node[ndsel] (s1) at (4,2.5) {$S_1^2$};
        \node[ndsel] (s2) at (2.5,4) {$S_2^2$};
        \draw[arout] (X) to (T);
        \draw[arout] (T) to (Y);
        \draw[arout] (X) to (Y);
        \draw[arout] (s1) to (X);
        \draw[arout] (s1) to (Y);
        \draw[arout] (l3) to (X);
        \draw[arout] (l1) to (l3);
        \draw[arout] (l1) to (s1);
        \draw[arout] (l1) to (l2);
        \draw[arout] (l1) to (Y);
        \draw[arout] (l2) to (Y);
        \draw[arout] (s2) to (X);
        \node at (3,0) {$G(\tilde{M}^2_{(2)})$};
    \end{scope}
  \end{tikzpicture}
  \caption{Some possible causal graphs of SCMs $\tilde{M}^i$ in Example \ref{ex:identif}.}
  \label{fig_ex_identif2}
  \end{figure}
\end{example}

\begin{remark}
    One can generalize other identification results similarly.
\end{remark}

\begin{remark}\label{rem:back-door}
    The result in Example \ref{ex:identif} differs from that in \citet{Correa17causal}, where it is assumed that the explicit causal structure of the selection mechanism is known, allowing identification of the causal effect in the whole population from the selected data. It is a slight generalization of the conditional back-door adjustment in \citet{pearl2009causality}, which is expressed as $\Prb(Y=y\mid \Do(T=t),S=s)=\sum_{x} \Prb(Y=y\mid X=x,T=t,S=s)\Prb(X=x\mid S=s)$ when certain graphical criteria are met. One difference is that in \eqref{eq:back_door}, $\Si^i$ may not be a singleton but a general set. The generalized back-door criterion for MAGs cannot be applied here, since it rules out selection bias explicitly \citep{Maathuis15generalized}. Note that even if we rule out selection bias, interpreting the graph $G(M^1)$ as a MAG will have different consequences than interpreting it as a causal graph of an SCM. Indeed, $\Prb(Y\mid \Do(T))$ is not identifiable in $G(M^3)$ in Figure \ref{fig_ex_identif} while the MAG representation of $G(M^3)$ is syntactically equal to the graph $G(M^1)$.
\end{remark}

Pearl's do-calculus is proved to be sound and complete (under some conditions) for identifying interventional distributions in terms of the observational distribution given a causal graph \citep{pearl95causal,huang06pearl}. Using a causal graph and observational distribution as inputs, the ID-algorithm, as a sound and complete algorithm, systematically expresses the target interventional distribution in terms of a functional of the observational distribution, if the target is identifiable and outputs FAIL if not \citep{tian02general,shpister06joint,Huang2008OnTC}. Various variants of the ID-algorithm exist, each with different targets and inputs (see e.g., \citet{yaroslav23identifiability} and the references therein).

\begin{example}[ID-algorithm]\label{ex:ID-algorithm}
One such variant, the s-ID-algorithm, is a sound and complete algorithm for the \emph{s-identification problem}, whose goal is to identify interventional distributions on a subpopulation ($\Prb(X_A\mid \Do(X_T=x_T),X_S=1)$) given a causal graph with selection mechanism ($G$) and selected observational distribution ($\Prb(X_V\mid X_S=1)$) \citep[Theorem 1, Corollary 2]{Abouei24sID}.\footnote{Note that in the usual c-ID-algorithm for conditional interventional distribution, the input is $\Prb(X_V)$ but not $\Prb(X_V\mid X_S=1)$.}  As we shall see, the conditioning operation can help simplify a part of the original proof \citep[Lemma~5]{abouei24sIDlatent}. 

Consider the single-variable case $T=\{t\}$. In the setting of \citet{Abouei24sID}, there are no latent variables. Therefore, if $ T\cap \anc_{G}(S)=\emptyset$, then there are no bidirected edges connecting to $t$ in $G_{|S}$, which implies that $\Prb(X_A\mid \Do(X_T=x_T),X_S=1)$ is identifiable by \citet[Theorem 1]{tian02general}. Now, assuming that $A \underset{G_{\ul{T}} }{\overset{d}{\perp}} T \mid S$,  the second rule of Pearl's do-calculus provides the identification result. Combining these two gives a sound and complete algorithm for the s-identification problem \citep[Theorem 1]{Abouei24sID}. The soundness of this algorithm immediately generalizes to settings with latent variables.

Besides, if $T \cap \anc_{G}(S)=\emptyset$, then one can also consider identifying the conditional causal effect on the subpopulation $\Prb(X_A\mid \Do(X_T=x_T),X_B,X_S=1)$ from a graph with latent variables and selected observational distribution $\Prb(X_V\mid X_S=1)$, by first applying the conditioning operation for $G$ to get $G_{|S}$ and then applying the classical ID-algorithm for conditional causal effect with latent variables on $G_{|S}$ \citep{shpitser06identification}.\footnote{If $T \cap \anc_{G}(S)\ne \emptyset$, one can still apply the corresponding ID-algorithm to $G_{|S}$, but the algorithm would output an expression for $\Prb(X_A(x_T)\mid X_S=1)$ instead of $\Prb(X_A\mid \Do(X_T=x_T),X_S=1)$. See Theorem \ref{thm:causal_semantics} and Section \ref{sec:caveat}.}  This result seems to be new in the literature to our knowledge.\footnote{When we were writing this manuscript, we found that an s-ID-algorithm under latent variables was proposed in \citet{abouei24sIDlatent}. However, they only consider identification for the unconditional interventional distribution $\Prb(X_A\mid \Do(X_T=x_T),X_S=1)$, not for the conditional interventional distribution $\Prb(X_A\mid \Do(X_T=x_T),X_B,X_S=1)$.} Similar generalizations can be made for other variants of the ID-algorithm, by first applying the conditioning operation for the graph and then applying the corresponding version of the ID-algorithm to the conditioned graph (e.g., in the one-line formulation of the ID-algorithm \citet[Theorem~48]{richardson2023nested}, replace $G$ with $G_{|S}$).

However, one should note that applying an ID-algorithm to the conditioned graph alone can hardly give a complete algorithm in general, due to the abstraction nature of the conditioning operation. For example, in the case of the s-ID-algorithm, we can use the conditioning operation to handle cases where $T\cap \anc_{G}(S)=\emptyset$, but a complete algorithm should also be able to address cases where $T\cap \anc_{G}(S)\ne \emptyset$ and $T \underset{G_{\ul{T}} }{\overset{d}{\perp}} A \mid S$ (see \citet[Theorem 1]{Abouei24sID}).
\end{example}

% \begin{remark}\label{rem:mag_idalgorithm}
%     It is important to note that if we use the conditioning of MAGs, then it is not clear how to apply the classical ID-algorithm to the conditioned graph.
% \end{remark}

\subsection{Causal discovery under latent selection}

Many causal discovery algorithms address unobserved common causes, but exclude selection bias. For simplicity, we consider consistent algorithms that output a single ADMG (instead of equivalence class). We can interpret the output of such algorithms as $G((M_{\sm L})_{|\Si})$ where $M$ is an acyclic SCM with latent nodes $L$, selection mechanism $X_S\in \Si$, and $L\cap S=\emptyset$. This can give a certain causal interpretation to the output of these algorithms under selection bias even if selection bias is excluded in the original formulations. The high-level idea is: given one such algorithm $\Ac$, ideal infinite i.i.d.\ data $\Dc$ and a model class $\Mb$ of SCMs, the algorithm outputs a causal graph $\Ac(\Dc)$ such that there exists $M$ in $\Mb$ such that $G(M)=\Ac(\Dc)$. Now suppose that the data $\widetilde{\Dc}$ are generated by an s-SCM $M^S$ in some model class $\Mb^S$ such that the conditioning operation projects $\Mb^S$ into a subclass $\Mb_{|S}$ of $\Mb$. Then we can apply the \emph{same algorithm} to get $\Ac(\widetilde{\Dc})$ and by the fact $\Prb_{M^S}(X_O)=\Prb_{M}(X_O\mid X_S\in\Si)=\Prb_{M_{|X_S\in \Si}}(X_O)$, we have $G(M_{|X_S\in \Si})=\Ac(\widetilde{\Dc})$. \cref{thm:causal_semantics,thm:as_much_as_possible,prop:impossible} tell us that $\Ac(\widetilde{\Dc})$ can be seen as the closest approximation of $M^S$ in $\Mb$.

\begin{example}[Causal discovery]\label{ex:discovery}
    For one instance, \citet{wang23discovery} explored recovering causal graphs uniquely from data generated by an acyclic linear non-Gaussian SCM with a bow-free graph (i.e., no simultaneous bidirected and directed edges between two variables) and rule out selection bias. Assume that the data are generated from an acyclic linear s-SCM $(M,X_S\in \Si)$ and the conditioned marginalized SCM of it has a bow-free graph. If the exogenous distribution of $(M_{\sm L})_{|\Si}$ is non-Gaussian and $(M_{\sm L})_{|\Si}$ satisfies the assumptions in \citet[Section~3]{wang23discovery}, then we can use the algorithm $\mathrm{BANG}$ in \citet{wang23discovery} to recover the graph of $(M_{\sm L})_{|\Si}$.
\end{example}

If we know from data or prior knowledge that a node $t$ is not an ancestor of $S$, then we can give a causal interpretation of $X_t$ in the discovered graph and apply causal identification results to identify $\Prb_M(X_O\mid \Do(X_t=x_t), X_S\in \Si)$ with $O\subseteq V\sm (L\cup S)$. For example, if the data are selected by $X_S=x_S$, we can sometimes read off whether $t\notin \anc_{G(M)}(S)$ from a PAG (Partial Ancestral Graphs) or a MAG \citep{spirtes95causal,richardson2002ancestral}.\footnote{Note that if $t\in \anc_{G(M)}(S)$, we can still apply the identification result to the interventional distribution given $\Do(X_t=x_t)$ in $M_{|X_S\in \Si}$, but the causal identification results will output a formula for $\Prb_M(X_O(x_t)\mid X_S\in \Si)$ instead of $\Prb_M(X_O\mid \Do(X_t=x_t), X_S\in \Si)$ (\cf Theorem \ref{thm:causal_semantics}, Section \ref{sec:caveat}).} 

 In addition to the causal discovery algorithms mentioned above, some causal model selection methods, such as the inflation technique \citep{WolfeSpekkensFritz+2019}, can also be generalized to deal with selection bias via the conditioning operation.

 \subsection{Instrumental variable and mediation analysis under latent selection}

 In some situations, we cannot achieve point identification results, but we can derive informative bounds for target causal effects. A well-known example is the instrumental inequality \citep{pearl2009causality,balke94counterfactual,balke97bounds,pearl95instrumental}. More recent advances include, e.g., showing that the instrumental inequality is sharp for finite discrete variables under certain constraints on the cardinality of the variables \citep{Bhadane25Revisiting,VanHimbeeck19quantum}, and extending the bounds to continuous outcomes \citep{Zhang21bounding}. Not only can the original instrumental inequality for binary variables be extended to the case with certain selection bias immediately via the conditioning operation, but also the results we mentioned above.

\begin{example}[Instrumental variables]\label{ex:inst}

The instrumental inequality was originally derived for the SCMs with the graph $G(M)$ shown in Figure \ref{fig_ex_instrum}. Similarly to Example \ref{ex:identif}, if we know that for an SCM $\tilde{M}$ with latent variables $L$ and latent selection $S\in\Si$, the causal graph $G\big(\big(\tilde{M}_{\sm L}\big)_{|\Si}\big)$ takes the form shown in Figure~\ref{fig_ex_instrum}, then we can conclude that the same form of inequality also holds for $\tilde{M}$ under the subpopulation.

  \begin{figure}[ht]
  \centering
  \begin{tikzpicture}[scale=0.8, transform shape]
    \begin{scope}[xshift=0]
      \node[ndout] (T) at (-0.5,1) {$T$};
      % \node[ndexo] (E) at (-0.5,2.25) {$E$};
      \node[ndout] (X) at (1,1) {$X$};
      \node[ndlat] (U) at (2,2) {$U$};
      \node[ndout] (Y) at (3,1) {$Y$};
      % \draw[arout] (E) to (T);
      \draw[arout] (T) to (X);
      \draw[arout] (X) to (Y);
      \draw[arout] (U) to (Y);
      \draw[arout] (U) to (X);
      \node at (1,0) {$G(M)$};
    \end{scope}
    \begin{scope}[xshift=5cm]
      \node[ndout] (T) at (-0.5,1) {$T$};
      \node[ndout] (X) at (1,1) {$X$};
      \node[ndout] (Y) at (3,1) {$Y$};
      \draw[arout] (T) to (X);
      \draw[arout] (X) to (Y);
      \draw[arlat, bend left] (X) to (Y);
      \node at (1.5,0) {$G\big(\big(\tilde{M}_{\sm L}\big)_{|\Si}\big)$};
    \end{scope}

    \begin{scope}[xshift=10cm]
      \node[ndout] (T) at (-0.5,1) {$T$};
      % \node[ndexo] (E) at (-0.5,2.25) {$E$};
      \node[ndout] (X) at (1,1) {$X$};
      \node[ndlat] (L1) at (1,2.5) {$L_1$};
      \node[ndlat] (L2) at (3,2.5) {$L_2$};
      \node[ndsel] (S) at (2,1.75) {$S$};
      \node[ndout] (Y) at (3,1) {$Y$};
      \draw[arout] (L1) to (X);
      \draw[arout] (L2) to (Y);
      \draw[arout] (L1) to (S);
      \draw[arout] (L2) to (S);
      % \draw[arout] (E) to (T);
      \draw[arout] (T) to (X);
      \draw[arout] (X) to (Y);
      \node at (1,0) {$G$};
    \end{scope}

    \begin{scope}[xshift=15cm]
      \node[ndout] (T) at (-0.5,1) {$T$};
      % \node[ndexo] (E) at (-0.5,2.25) {$E$};
      \node[ndout] (X) at (1,1) {$X$};
      \node[ndout] (Y) at (3,1) {$Y$};
      % \draw[arout] (E) to (T);
      \draw[arout] (T) to (X);
      \draw[arout] (X) to (Y);
      \draw[arout, bend left] (T) to (Y);
      \node at (1,0) {$\tilde{G}$};
    \end{scope}
  \end{tikzpicture}
  \caption{$G(M)$ and $G\big(\big(\tilde{M}_{\sm L}\big)_{|\Si}\big)$ are graphs for the instrumental variables model. $G$ is the graph of a model with selection bias whose marginalized and conditioned graph is $G\big(\big(\tilde{M}_{\sm L}\big)_{|\Si}\big)$ while $\tilde{G}$ is its MAG representation.}
  \label{fig_ex_instrum}
  \end{figure}

  If we further assume a continuous linear model $Y=\beta X+f(U)$ in $M$, then the parameter $\beta$ is identifiable when $\mathrm{Cov}(X,Y)\ne 0$ and is estimated as $\frac{\mathrm{Cov}_{M}(T,Y)}{\mathrm{Cov}_{M}(X,Y)}$, where selection bias is implicitly ruled out \citep{imbens00instrumental}. With the conditioning operation, we can see that the parameter remains identifiable from the selected conditional distribution $\Prb_{\tilde{M}}(T,X,Y\mid S\in\Si)$ as $\frac{\mathrm{Cov}_{\tilde{M}}(T,Y\mid
  S\in\Si)}{\mathrm{Cov}_{\tilde{M}}(X,Y\mid S\in\Si)}$ even under certain forms of selection bias. Therefore, we have extended the identification result to include certain forms of selection bias.
\end{example}

\begin{remark}\label{rem:mag_instrumental}
    Note that it is unclear how MAGs can handle this example. If graph $G$ in Figure \ref{fig_ex_instrum} is interpreted as a MAG then conditioning on $S$ and marginalizing out $L_1$ and $L_2$ would yield the MAG $\tilde{G}$ in Figure \ref{fig_ex_instrum}, where the assumptions of instrumental variables are violated, thus being too coarse to establish the instrumental inequality.
\end{remark}

Mediation analysis is crucial in many fields such as epidemiology, natural science, and policy making, where understanding ``path-specific'' causal effects is often necessary \citep{pearl2001direct,pearl14interpretation,pearl2009causality,Robins1992IdentifiabilityAE}. Traditional methods rely on linear regression, but linear SCMs have been proven problematic due to potential nonlinear interactions among variables, latent common causes, and selection bias in real-world problems \citep{shpitser2013counterfactual}.

\begin{example}[Mediation analysis]\label{ex:mediation}
 With the help of potential outcomes and causal graphs of SCMs, \citet{pearl14interpretation} and \citet{shpitser2013counterfactual} study methods to perform mediation analysis when there are nonlinear functional dependencies and unobserved common causes. By extending the interpretation of bidirected edges to also represent \emph{selection bias}, we can extend these results to account for selection bias immediately, similarly to the approach in previous examples.

For another example on how the conditioning operation is helpful, suppose that one is interested in the effect of, e.g., $A$ (obesity) on $Y$ (mortality) while conditioning a mediating variable on the path between them to a specific value (e.g., $S=1$: having heart disease) \citep{smith2020selection}. The graph $G$ is shown in Figure~\ref{fig:mediation}. Applying the graphical conditioning operation gives $G_{|S}$. This shows that we can obtain a causal identification result for $\mathrm{E}[Y(a)\mid S=1]-\mathrm{E}[Y(a^\prime)\mid S=1]$ via back-door adjustment on $L$.

\begin{figure}[ht]
  \centering
  \begin{tikzpicture}[scale=0.8, transform shape]
    \begin{scope}[xshift=0]
      \node[ndout] (A) at (0,0) {$A$};
      \node[ndout] (S) at (1.5,0) {$S$};
      \node[ndout] (Y) at (3,0) {$Y$};
      \node[ndout] (L) at (2,-1) {$L$};
      \draw[arout] (A) to (S);
      \draw[arout] (S) to (Y);
      \draw[arout] (L) to (S);
      \draw[arout] (L) to (Y);
      \draw[arout, bend left] (A) to (Y);
      \node at (1.5,-2) {$G$};
    \end{scope}
    \begin{scope}[xshift=5cm]
      \node[ndash] (A) at (0,0) {$A$};
      \node[ndout] (Y) at (3,0) {$Y$};
      \node[ndash] (L) at (2,-1) {$L$};
      \draw[arout] (A) to (Y);
      \draw[arout] (L) to (Y);
      \draw[arlat, bend right] (A) to (L);
      \node at (1.5,-2) {$G_{|S}$};
    \end{scope}
  \end{tikzpicture}
  \caption{Graph $G$ for mediation analysis conditioning on one mediator and its conditioned graph $G_{|S}$.}
  \label{fig:mediation}
  \end{figure}
\end{example}

\begin{remark}
    The previous example shows that the conditioning operation is helpful in obtaining causal identification results, so it could also play a role in fairness analysis \citep{Nabi18fair,Chiappa19path-specific,kusner2017counterfactual,zhang2018fairness,Bhadane25Revisiting}.
\end{remark}

\subsection{Causal modeling under latent selection}

The question of how to perform causal modeling under selection bias is one of the original main motivations for this work. In the following example, we show how the conditioning operation can help with causal modeling under (latent) selection bias. The high-level idea is from Example~\ref{ex:rei} that even if there are no causal effects and no common causes between two variables there could still be dependency between them caused by selection bias. To state the example, recall that one possible workflow of causal inference is:
\begin{enumerate}
    \item[(i)] asking causal queries;
    \item[(ii)] \textbf{building a causal model} from prior knowledge and data;
    \item[(iii)] determining the target causal quantity and identifying the estimand in terms of available observational and interventional distributions;
    \item[(iv)] using data to estimate the estimand.
\end{enumerate}
As~concise encodings of causal assumptions, causal graphs can be used to decide the estimand for addressing causal queries, and therefore incorrect graphs might generate wrong results. 

\begin{example}[Causal modeling]\label{ex:causalassum}
To understand the causal effect of treatment strategies from different countries on the fatality rate of COVID-19, \citet{von2021simpson} analyzed data
from the initial virus outbreaks in 2020 in China and Italy, and assumed the causal graph $G$ shown in Figure \ref{fig_ex_causalassum}. For COVID-19 infected people, age ($A$), country of residence ($C$) at the time of infection and fatality rate ($F$) are recorded.

The data suggest that $C$ and $A$ are dependent. In the traditional understanding of bidirected edges, assuming that $C$ and $A$ do not share a latent common cause, one has to draw a directed edge between $C$ and $A$ so that the hypothesized graph is compatible with the observation. However, drawing a directed edge from $C$ to $A$ is not a reasonable causal assumption. It assumes that if we conduct a randomized trial to assign people to different countries, then immediately ($A$ and $C$ are measured almost the same time) the resulting age distribution will differ depending on the assigned country. Similarly, $A\tuh C$ would also be an unreasonable assumption.

However, the conditioning operation tells us that bidirected edges do not have to represent latent common causes only, but can also represent latent selection bias. Therefore, we can draw a bidirected edge $C\huh A$ as shown in $\tilde{G}$ to explain the statistical association between $C$ and $A$, which could represent different
latent selection mechanisms or latent common causes or combinations of the two between $C$ and $A$.\footnote{Note that the difference of common cause and selection bias does not matter for the current task, which shows the power of model abstraction.} 
First, the age distribution may differ between two countries
already before the outbreak of the virus (latent selection on `person was alive ($S^\prime=1$) in
early 2020', as in $G^1$). Second, since only \emph{infected} patients were registered and both the country and the age may influence the risk of getting infected, selection
of the infection status ($S=1$) can also lead to $C\huh A$ (as in $G^2$). The combinations of both
selection mechanisms (such as in $G^3$ or $G^4$) also lead to $C\huh A$. With the conditioning operation, we do not need to list (potentially infinitely many) all the possible causal graphs in detail, including \emph{all} relevant latent
variables that model the selection mechanism. We only need to consider DMGs on these three observed variables, which is a much smaller (finite) model space.

Thanks to properties of the conditioning operation, we can answer causal queries like ``what would be the effect on fatality of changing from China to Italy". It allows us to compute the total causal effect $\mathrm{TCE}(F;c^\prime\rightarrow c)\coloneqq \mathrm{E}[F\mid\Do(C=c)]-\mathrm{E}[F\mid \Do(C=c^\prime)]$ via the abstracted
(conditioned) model $\tilde{G}$ (e.g., by adjusting on age) without fully knowing all the latent details. Note that the results based on $G$ and $\tilde{G}$ are clearly different. In fact, for an SCM with graph $G$, one has:
\[
\mathrm{TCE}(F;c^\prime\rightarrow c)\coloneqq \mathrm{E}[F\mid\Do(C=c)]-\mathrm{E}[F\mid \Do(C=c^\prime)]=\mathrm{E}[F\mid C=c]-\mathrm{E}[F\mid C=c^\prime].
\]
On the other hand, for an SCM with graph $\tilde{G}$, one has:
\[
\begin{aligned}
  \mathrm{TCE}(F;c^\prime\rightarrow c)\coloneqq &~\mathrm{E}[F\mid\Do(C=c)]-\mathrm{E}[F\mid \Do(C=c^\prime)]\\
  =&\sum_a\left(\mathrm{E}[F\mid C=c,A=a]-\mathrm{E}[F\mid C=c^\prime,A=a]\right)\Prb(A=a)\\
  \ne&\sum_a(\mathrm{E}[F\mid C=c,A=a]\Prb(A=a\mid C=c) &\text{(in general)}\\
  &~-\mathrm{E}[F\mid C=c^\prime,A=a]\Prb(A=a\mid C=c^\prime)) \quad \\
  =&~\mathrm{E}[F\mid C=c]-\mathrm{E}[F\mid C=c^\prime],
\end{aligned}
\]
where in the second equality we use the back-door theorem allowed by the graph $\tilde{G}$.

\begin{figure}[ht]
\centering
\begin{tikzpicture}[scale=0.8, transform shape]
  \begin{scope}[xshift=0]
    \node[ndout] (T) at (1,1) {$C$};
    \node[ndout] (X) at (2,2) {$A$};
    \node[ndout] (Y) at (3,1) {$F$};
    \draw[arout] (T) to (X);
    \draw[arout] (T) to (Y);
    \draw[arout] (X) to (Y);
    \node at (2,0) {$G$};
  \end{scope}
  \begin{scope}[xshift=4cm]
    \node[ndout] (T) at (1,1) {$C$};
    \node[ndout] (X) at (2,2) {$A$};
    \node[ndout] (Y) at (3,1) {$F$};
    \draw[arout] (T) to (Y);
    \draw[arout] (X) to (Y);
    \draw[arlat, bend left] (T) to (X);
    \node at (2,0) {$\tilde{G}$};
  \end{scope}

  \begin{scope}[xshift=9cm]
    \node[ndout] (T) at (1,1) {$C$};
    \node[ndout] (X) at (2,2) {$A$};
    \node[ndlat] (T') at (-0.25,1) {$C^\prime$};
    \node[ndlat] (X') at (0.25,3) {$A^\prime$};
    \node[ndsel] (S') at (-0.8,2.25) {$S^\prime$};
    \node[ndout] (Y) at (3,1) {$F$};
    \draw[arout] (T') to (S');
    \draw[arout] (X') to (S');
    \draw[arout] (T') to (T);
    \draw[arout] (X') to (X);
    \draw[arout] (T) to (Y);
    \draw[arout] (X) to (Y);
    \node at (1,0) {$G^{1}$};
  \end{scope}
  \begin{scope}[xshift=0,yshift=-4cm]
    \node[ndout] (T) at (1,1) {$C$};
    \node[ndout] (X) at (2,2) {$A$};
    \node[ndlat] (T') at (-0.25,1) {$C^\prime$};
    \node[ndlat] (X') at (0.25,3) {$A^\prime$};
    \node[ndsel] (S) at (0.8,1.8) {$S$};
    \node[ndout] (Y) at (3,1) {$F$};
    \draw[arout] (T') to (S);
    \draw[arout] (X') to (S);
    \draw[arout] (T') to (T);
    \draw[arout] (X') to (X);
    \draw[arout] (T) to (Y);
    \draw[arout] (X) to (Y);
    \node at (1,0) {$G^2$};
  \end{scope}
  \begin{scope}[xshift=5cm,yshift=-4cm]
    \node[ndout] (T) at (1,1) {$C$};
    \node[ndout] (X) at (2,2) {$A$};
    \node[ndlat] (T') at (-0.25,1) {$C^\prime$};
    \node[ndlat] (X') at (0.25,3) {$A^\prime$};
    \node[ndsel] (S) at (0.8,1.8) {$S$};
    \node[ndsel] (S') at (-0.8,2.25) {$S^\prime$};
    \node[ndout] (Y) at (3,1) {$F$};
    \draw[arout] (T') to (S);
    \draw[arout] (X') to (S);
    \draw[arout] (T') to (S');
    \draw[arout] (X') to (S');
    \draw[arout] (T') to (T);
    \draw[arout] (X') to (X);
    \draw[arout] (T) to (Y);
    \draw[arout] (X) to (Y);
    \node at (1,0) {$G^{3}$};
  \end{scope}
  \begin{scope}[xshift=10cm,yshift=-4cm]
    \node[ndout] (T) at (1,1) {$C$};
    \node[ndout] (X) at (2,2) {$A$};
    \node[ndlat] (T') at (-0.25,1) {$C^\prime$};
    \node[ndlat] (X') at (0.25,3) {$A^\prime$};
    \node[ndsel] (S) at (0.8,1.8) {$S$};
    \node[ndsel] (S') at (-0.8,2.25) {$S^\prime$};
    \node[ndout] (Y) at (3,1) {$F$};
    \draw[arout] (T') to (S);
    \draw[arout] (X') to (S);
    \draw[arout] (T') to (S');
    \draw[arout] (X') to (S');
    \draw[arout] (S') to (S);
    \draw[arout] (T') to (T);
    \draw[arout] (X') to (X);
    \draw[arout] (T) to (Y);
    \draw[arout] (X) to (Y);
    \node at (1,0) {$G^{4}$};
  \end{scope}
\end{tikzpicture}
\caption{Causal graphs for COVID-19 data. Note that after applying the conditioning operation to selection variables and marginalizing out remaining latent variables, we reduce $G^i$ to $\tilde{G}$ for $i=1,2,3,4$.}
\label{fig_ex_causalassum}
\end{figure}
\end{example}

% \begin{remark}\label{rem:mag_causalmodel}
%     Note that although applying the conditioning and marginalization of MAGs to DAGs $G^i$ for $i=1,\ldots,4$ can lead to the same graph as $\tilde{G}$,\footnote{We thank Robin Evans for pointing this out to us.} it is not clear yet how one can apply the identification result we have for SCMs to MAGs under selection bias since the interpretation of the two graphs are different.\footnote{This is similar to the connection between Bayesian net and causal Bayesian net. Although graphically they are both DAGs, the interpretation of the two are different and one cannot directly apply the causal identification result of causal Bayesian net to Bayesian net \citep{dawid10dag}.}
% \end{remark}

\section{Discussion} \label{s5}

Although SCMs have been widely used to study selection bias from a structural causal viewpoint, a formal theory was still absent. We gave a mathematical definition of s-SCMs (Definition~\ref{def:SCM_selection}), which formalizes the idea of SCMs with selection mechanisms, and a description of the data-generating processes that they are modeling. Motivated by the marginalization of causal models, which plays an important role in abstracting away unnecessary \emph{unconditioned} latent details of causal models, we defined a conditioning operation (Definition~\ref{def:cdSCM}) to transform an SCM with selection mechanisms into an SCM without selection mechanisms so that the new SCM preserves important information from the original SCM with selection mechanisms. The benefit of doing so is that, without the need to develop a separate theory for s-SCMs, we can \emph{reduce} the problems involving s-SCMs to ordinary SCMs, so that all the well-developed tools from SCMs can be applied directly. We also explored the theoretical limit of such a transformation by showing what can be preserved (Section~\ref{sec:3.2}) and what is impossible to preserve (Appendix~\ref{app:impossible}) during such a model abstraction process.

Most importantly, we generalized the interpretation of bidirected edges in directed mixed graphs (interpreted as causal graphs of SCMs) so that they can represent not only latent common causes but also latent selection bias. This makes the rough idea of ``using bidirected edges to represent selection bias'' formal, such as Pearl's claim in his causality book \citep[p.163]{pearl2009causality}. Using the same symbol (bidirected edge) to represent potential latent common causes and latent selection bias is also consistent with some observation in epidemiology \citep[Footnote 11]{richardson2013single}. Combined with marginalization and intervention, the conditioning operation provides a powerful tool for causal model abstraction and helps with many causal inference tasks such as prediction under interventions, identification, and model selection.

 One approach of causal modeling involves: (i) commencing with a complete graph, i.e., it has two directed edges in different directions and a bidirected edge between any two observed endogenous variables; (ii) iteratively deleting edges based on prior knowledge and available data. Our result contributes to this procedure by mathematically proving that, within the SCM framework, one should retain the bidirected edge between two variables when there is insufficient knowledge to rule out both  unmeasured common cause and latent selection bias.

The current work focuses mainly on the theoretical aspects of the conditioning operation. Some of the applications are briefly examined. We envision exploring further and more detailed research of applications enabled by conditioning operation in future work. In particular, the conditioning operation might be helpful in giving a causal interpretation to the output of certain causal discovery algorithms under selection bias.

Markov categories have recently emerged as a categorical framework for probability and statistics \citep{fritz2020synthetic}. In this “synthetic” approach, classical measure-theoretic foundation is replaced by a categorical one, and many familiar results can be proved  algebraically within the framework \citep{fritz2020synthetic,fritz2021_definetti_catprob_josa,chen2024aldous_hoover_catprob,fritz2025empirical_slln_catprob}.  Crucially, causal modeling can also be formulated at this abstract level \citep{fritz2023d,lorenz2023causal_string_diagrams}. In particular, it is possible to extend the theory of conditioning SCMs to the categorical setting, where recent work on partializations of Markov categories might be relevant \citep{mohammed2025partializations_markov_categories}.

\acks{The authors acknowledge Booking.com for support. The authors thank Philip Boeken, Stephan Bongers, Robin Evans, Tobias Fritz, Areeb Shah Mohammed, and Thomas Richardson for discussions. The authors are grateful to Luigi Gresele for discussions regarding \cref{ex:causalassum} and for affirming our approach to causal modeling of the Covid example.}

\newpage
\appendix
\section{More SCM preliminaries}\label{app:def}

  To be as self-contained as possible, we include the definitions of twin SCM and (augmented) causal graphs of SCMs. We follow the formal definitions of \citet{bongers2021foundations}.

  \begin{definition}[Twin SCM]\label{def:twin}
    Let $M=\SCM$ be an SCM. The twinning operation maps $M$ to the \textbf{twin structural causal model (twin SCM)}
    $$
    M^{\twin}\coloneqq \left(V \cup V^{\prime}, W, \mathcal{X}_V \times \mathcal{X}_{V^\prime}\times\Xc_W,\Prb,\tilde{f}\right),
    $$
    where $V^{\prime}=\left\{v^{\prime}: v \in V\right\}$ is a disjoint copy of $V$ and the causal mechanism $\tilde{f}: \mathcal{X}_V \times \mathcal{X}_{V^\prime}\times\Xc_W \rightarrow \mathcal{X}_V \times \mathcal{X}_{V^\prime}$
    is the measurable mapping given by $\tilde{f}\left(x_V, x_{V^\prime}, x_W\right)=\left(f(x_V, x_W), f(x_{V^\prime}, x_W)\right)$.
  \end{definition}

  \begin{definition}[Parent]\label{def:pa}
   Let $M=\SCM$ be an SCM. We call $k \in V \cup W$ a \textbf{parent} of $v \in V$ if and only if there does not exist a measurable mapping $ \tilde{f}_v: \mathcal{X}_{V\sm k} \times \mathcal{X}_{W\sm k} \rightarrow \mathcal{X}_v$ such that for $\Prb(X_W)$-almost every $x_W\in\Xc_W$, for all $x_V\in\Xc_V$,
  $$
  x_v=f_v(x_V, x_W) \quad \Longleftrightarrow \quad x_v=\tilde{f}_v\left(x_{V\sm k}, x_{W\sm k}\right).
  $$
  \end{definition}

  \begin{definition}[Graph and augmented graph]\label{def:gra}
  Let $M=\SCM$ be an SCM. We define:
  \begin{enumerate}
        \item[(1)] the augmented graph $G^a(M)$ as the directed graph with nodes $V\cup W$ and directed edges $u \rightarrow v$ if and only if $u \in V\cup W$ is a parent of $v \in V$;

    \item[(2)] the graph $G(M)$ as the directed mixed graph with nodes $V$, directed edges $u \tuh v$ if and only if $u \in V$ is a parent of $v \in V$ and bidirected edges $u \huh v$ if and only if there exists a $w \in W$ that is a parent of both $u \in V$ and $v \in V$.
  \end{enumerate}
  \end{definition}
  Note that $G(M)=(G^a(M))_{\sm W}$, where the graphical marginalization (also known as ``latent projection'') is defined in \citet[Definition~5.7]{bongers2021foundations}.

  \begin{example}\label{ex:appdend}
    Consider the SCM
    \[
    M:\left\{\begin{array}{l}
    U \sim \operatorname{Ber}(1-\xi),U_B \sim \operatorname{Ber}(1-\delta),
    U_E \sim \operatorname{Ber}(1-\varepsilon),\\
    B_0 = U, E_0 = U, S_0=B_0 \wedge E_0, \\
    B_1=B_0\wedge U_B, E_1=E_0\wedge U_E, S_1=B_1 \wedge E_1.
    \end{array}\right.
    \]
    Then we have the (augmented) causal graphs of $M$ shown in Figure~\ref{suppl}.
  \end{example}

  \begin{figure}[ht]
  \centering
  \begin{tikzpicture}[scale=0.85, transform shape]
    \begin{scope}[xshift=0]
      \node[ndexo] (U) at (2,5.5) {$U$};
      \node[ndout] (B0) at (1,4.5) {$B_0$};
      \node[ndout] (E0) at (3,4.5) {$E_0$};
      \node[ndexo] (Ub) at (0,4) {$U_B$};
      \node[ndexo] (Ue) at (4,4) {$U_E$};
      \node[ndout] (B1) at (1,3) {$B_1$};
      \node[ndout] (E1) at (3,3) {$E_1$};
      \node[ndout] (S0) at (2,3.5) {$S_0$};
      \node[ndout] (S1) at (2,2) {$S_1$};
      \draw[arout] (U) to (B0);
      \draw[arout] (U) to (E0);
      \draw[arout] (Ub) to (B1);
      \draw[arout] (Ue) to (E1);
      \draw[arout] (B0) to (S0);
      \draw[arout] (E0) to (S0);
      \draw[arout] (B0) to (B1);
      \draw[arout] (E0) to (E1);
      \draw[arout] (B1) to (S1);
      \draw[arout] (E1) to (S1);
      \node at (2,1) {$G^a(M)$};
    \end{scope}

    \begin{scope}[xshift=5cm]
      \node[ndout] (B0) at (1,4.5) {$B_0$};
      \node[ndout] (E0) at (3,4.5) {$E_0$};
      \node[ndout] (S0) at (2,3.5) {$S_0$};
      \node[ndout] (B1) at (1,3) {$B_1$};
      \node[ndout] (E1) at (3,3) {$E_1$};
      \node[ndout] (S1) at (2,2) {$S_1$};
      \draw[arout] (B0) to (S0);
      \draw[arout] (E0) to (S0);
      \draw[arout] (B0) to (B1);
      \draw[arout] (E0) to (E1);
      \draw[arout] (B1) to (S1);
      \draw[arout] (E1) to (S1);
      \draw[arlat, bend left] (B0) to (E0);
      \node at (2,1) {$G(M)$};
    \end{scope}

  \end{tikzpicture}
  \caption{The (augmented) causal graphs of the SCM $M$ in Example \ref{ex:appdend}.}
  \label{suppl}
  \end{figure}

\begin{definition}[(Counterfactual/interventional/observational) equivalence]\label{def:cio_equivalence}
\label{def:ce}\label{def:cou_equ}  A simple SCM $M=\SCM$ is \textbf{counterfactually equivalent} to a simple SCM $\tilde{M}=(\tilde{V},\tilde{W},\tilde{\Xc},\tilde{\Prb},\tilde{f})$ \wrt $O\subseteq V\cap\tilde{V}$ if for any
$T_1\subseteq O$ and $x_{T_1}\in\Xc_{T_1}$, and any $T_{2}\subseteq O^\prime$ and $x_{T_{2}}\in\Xc_{T_{2}}$,
\[
  \begin{aligned}
    &\Prb_{M^{\twin}}(X_{(O\cup O^\prime)\setminus(T_1\cup T_2)}\mid \Do(X_{T_1}=x_{T_1},X_{T_2}=x_{T_2}))\\
    &=\Prb_{\tilde{M}^{\twin}}(X_{(O\cup O^\prime)\setminus(T_1\cup T_2)}\mid \Do(X_{T_1}=x_{T_1},X_{T_2}=x_{T_2})).
\end{aligned}
\]
We say $M$ is \textbf{interventionally equivalent} to $\tilde{M}$ \wrt $O$ if
\[
    \Prb_{M}(X_{O\setminus T_1}\mid \Do(X_{T_1}=x_{T_1}))=\Prb_{\tilde{M}}(X_{O\setminus T_1}\mid \Do(X_{T_1}=x_{T_1})).
\]
We say $M$ is \textbf{observationally equivalent} to $\tilde{M}$ \wrt $O$ if
\[
    \Prb_{M}(X_{O})=\Prb_{\tilde{M}}(X_{O}).
\]
We say that $M$ and $\tilde{M}$ are observationally/interventionally/counterfatually equivalent if $V=\tilde{V}$ and $M$ is observationally/interventionally/counterfatually equivalent to $\tilde{M}$ \wrt $V$.
\end{definition}

\begin{remark}
We have
\[
\begin{aligned}
    \text{Equivalence of SCMs}\quad &\Longrightarrow \quad\text{Counterfactual equivalence}\\
    \quad &\Longrightarrow \quad \text{Interventional equivalence}\\
    \quad &\Longrightarrow \quad \text{Observational equivalence},
\end{aligned}
\]
but not conversely. See \citet[Proposition~4.6]{bongers2021foundations}.

\end{remark}

\begin{definition}[Directed global Markov property]\label{def:dmarkov}
    Let $G$ be a DMG with nodes $V$ and $\Prb(X_V)$ a probability distribution on $\Xc_V=\prod_{v\in V}\Xc_v$ for standard measurable spaces $X_v$. We say that the probability distribution $\Prb(X_V)$ satisfies the \textbf{directed global Markov property relative to $G$} if for subsets $A,B,C\subseteq V$ the set $A$ being $d$-separated from $B$ given $C$ implies that the random variable $X_A$ is conditionally independent of $X_B$ given $X_C$.
\end{definition}

\begin{theorem}[Directed Markov property for SCMs;  \citet{Forre2017markov}]\label{thm:dmarkov}
Let $M$ be a uniquely solvable SCM that satisfies at least one of the following three conditions:
\begin{enumerate}
    \item[(1)] $M$ is acyclic;
    \item[(2)]  all endogenous state spaces $\mathcal{X}_v$ are discrete and $M$ is ancestrally uniquely solvable \citep[Definition~3.9]{bongers2021foundations};
    \item[(3)] $M$ is linear \citep[Definition~C.1]{bongers2021foundations} and each of its causal mechanisms $\left\{f_v\right\}_{v \in V}$ has a nontrivial dependence on at least one exogenous variable, and $\Prb(X_W)$ has a density \wrt the Lebesgue measure on $\mathbb{R}^{W}$.
\end{enumerate}
Then its observational distribution $\Prb_M(X_V)$ exists, is unique, and satisfies the directed global Markov property relative to $G(M)$.
\end{theorem}

We first recall the definition of $\sigma$-separation. In the following, we write \[\operatorname{Sc}_G(C)\coloneqq \{\tilde{v}\in V: \exists \tilde{v} \tuh \cdots \tuh v \text{ and } \tilde{v} \hut \cdots \hut v \text{ for some } v\in C\}\] for the \textbf{strongly connected component} of $C\subseteq V$.

\begin{definition}[$\sigma$-sepation for DMGs, \citep{Forre2017markov,bongers2021foundations}]\label{def:sigma_sep}
Let $G$ be a DMG with nodes $V$ and $C \subseteq V$ a subset of nodes and $\pi$ a walk in $G$:
$$
\pi=\left(v_0 \sus \cdots \sus v_n\right) .
$$
\begin{enumerate}
    \item We say that the walk $\pi$ is \textbf{$C$-$\sigma$-blocked} or \textbf{$\sigma$-blocked by $C$} if:
    \begin{enumerate}
        \item [(i)] $v_0 \in C$ or $v_n \in C$ or;

        \item [(ii)] there are two adjacent edges in $\pi$ of one of the following forms:
            $$
            \begin{array}{rlll}
            \text { left chain: } & v_{k-1} \hut v_k \hus v_{k+1} & \text { with } & v_k \in C \wedge v_k \notin \mathrm{Sc}_G\left(v_{k-1}\right), \\
            \text { right chain: } & v_{k-1} \suh v_k \tuh v_{k+1} & \text { with } & v_k \in C \wedge v_k \notin \mathrm{Sc}_G\left(v_{k+1}\right), \\
            \text { fork: } & v_{k-1} \hut v_k \tuh v_{k+1} & \text { with } & v_k \in C \wedge v_k \notin \operatorname{Sc}_G\left(v_{k-1}\right) \cap \operatorname{Sc}_G\left(v_{k+1}\right), \\
            \text { collider: } & v_{k-1} \suh v_k \hus v_{k+1} & \text { with } & v_k \notin \operatorname{Anc}_G(C) .
            \end{array}
            $$
    \end{enumerate}
We say that the walk $\pi$ is  $C$-$\sigma$-open if it is not  $C$-$\sigma$-blocked.

\item Let $A, B, C \subseteq V$ (not necessarily disjoint) be  subsets of nodes. We then say that:
 \textbf{$A$ is $\sigma$-separated from $B$ given $C$ in $G$}, in symbols:
$$
A \underset{G}{\stackrel{\sigma}{\perp}} B \mid C,
$$
if every walk/path from a node in $A$ to a node in $B$ is $\sigma$-blocked by $C$. (In the definition, taking either walk or path gives an equivalent definition.)
\end{enumerate}
\end{definition}

\begin{definition}[Generalized directed global Markov property;  \citet{Forre2017markov}]\label{def:gmarkov}
    Let $G=(V,E,H)$ be a DMG and $\Prb(X_V)$ a probability distribution on $\Xc_V=\prod_{v\in V}\Xc_v$ for standard measurable spaces $X_v$. We say that the probability distribution $\Prb(X_V)$ satisfies the \textbf{generalized directed global Markov property relative to $G$} if for subsets $A,B,C\subseteq V$ the set $A$ being $\sigma$-separated (\cref{def:sigma_sep}) from $B$ given $C$ implies that the random variable $X_A$ is conditionally independent of $X_B$ given $X_C$.
\end{definition}

\begin{theorem}[Generalized directed Markov property for SCMs; \citet{Forre2017markov,bongers2021foundations}] \label{thm:gmarkov}
 Let $M$ be a simple SCM. Then its observational distribution $\Prb_M(X_V)$ exists, is unique, and satisfies the generalized
 directed global Markov property relative to $G(M)$.
\end{theorem}

\section{Some examples and remarks}\label{app:ex}

\begin{remark}[Proof of claim in \cref{rem:conf}]\label{rem:pf}
If we assume that there is an underlying acyclic SCM $M=\SCM$ inducing the potential outcomes $X_A$ and $X_B(x_A)$, then equation~\eqref{eqn:po_unconfound} is equivalent to
\begin{equation}\label{eqn:po_unconfound_scm}
    \forall x_A\in\{0,1\}: g_A(X_W)\ind g^{V\sm A}_B(x_A,X_W),
\end{equation}
where $g$ and $g^{V\sm A}$ are the (essentially unique) solution functions of $M$ \wrt $V$ and $V\sm A$, respectively. Then we can show that (under a positivity assumption) equation~\eqref{eqn:po_unconfound_1} holds. 

Indeed, for every $x_A\in\{0,1\}$ with $\Prb_M(X_A=x_A)>0$,
\[
\begin{aligned}
    \Prb_M(X_B\mid X_A=x_A)&=\Prb_M(g_B(X_W)\mid g_A(X_W)=x_A)\\
    &=\Prb_M(g_B^{V\sm A}(g_A(X_W),X_W)\mid g_A(X_W)=x_A)& \text{(Lemma \ref{lem:consistency})}\\
    &=\Prb_M(g_B^{V\sm A}(x_A,X_W)\mid g_A(X_W)=x_A)\\
    &=\Prb_M(g_B^{V\sm A}(x_A,X_W))  &\text{ (equation \eqref{eqn:po_unconfound_scm})}\\
    &=\Prb_M(X_B(x_A))\\
    &=\Prb_M(X_B\mid \Do(X_A=x_A)).
\end{aligned}
\]
\end{remark}

\begin{remark}[\cref{ass} is mild]\label{rem:assum}

\begin{enumerate}
    \item Note that the class of simple SCMs is a \emph{more general model class} than acyclic SCMs.  Besides, one can easily generalize all the results in this work to an even more general class of SCMs than simple SCMs by carefully postulating corresponding unique solvability assumptions for the SCM \wrt certain subsets of $V$. However, for non-simple SCMs, there are many counter-intuitive phenomena. For example, non-simple SCMs can induce none or multiple (observational and interventional) distributions, they lack a Markov property, and interventions may affect non-descendants of the intervened target (see e.g., \citet{bongers2021foundations}). This suggests that non-simple SCMs are not intuitive causal models. Therefore, we focus on simple SCMs in the current work.

    \item In real-world applications, we mostly observe data from events with positive probabilities. Also, mathematically, although measure theory provides a way to define conditional probabilities given a null event, it is still ambiguous in general when the Borel-Kolmogorov paradox arises \citep{kolmogorov18foundations, Jaynes03probability}. So, it is reasonable not to model selection events with zero probabilities.
\end{enumerate}
\end{remark}

\begin{remark}[Remark on \cref{prop:cond_scm_dmg}]\label{rem:cond_scm_dmg}
 \begin{enumerate}
     \item   Recall that $G(M_{\sm L})$ can be a strict subgraph of $G(M)_{\sm L}$. Also note that the ``merging step'' can be more coarse in $G(M)$ than in $M$. Therefore, $G(M_{|X_S\in\Si})$ can be a strict subgraph of $G(M)_{|S}$ due to the merging step and the marginalization. This means that $G(M)_{|S}$ is generally a (strictly) more conservative representation of the underlying conditioned SCM with less causal information due to the nature of abstraction (recall that a sparser causal graph encodes stronger assumptions).

     \item Any reasonable attempt to a purely graphical conditioning operation should satisfy this property. Otherwise, one would conclude from the graph $G_{|S}$ some results that do not hold for some conditioned SCMs $M_{|X_S\in \Si}$ where $M$ is compatible with the graph $G$. Also note that one is not able to further ``minimize'' the conditioned graph by eliminating some (bi)directed edges, since there always exists an SCM $M$ and a selection mechanism $X_S\in \Si$ such that $G(M_{|X_S\in \Si})=G(M)_{| S}$ (consider a linear SCM with positive coefficients in which every endogenous variable has at least one exogenous parent such that $\Prb(X_w)=\Nc(0,1)$ for all $w\in W$ and with selection mechanism $X_S\in [0,\infty)$).

     \item Recall that graphical marginalization preserves ancestral relationships, i.e., $\anc_{G_{\sm L}}(B)=\anc_{G}(B)\sm L$. This property also holds for graphical conditioning, i.e., $a\in\anc_{G}(b)$ iff $a\in \anc_{G_{| S}}(b)$ for any $S\subseteq V\sm \{a,b\}$. However, this is not the case for SCM marginalization and conditioning in general. At the level of SCMs, the best we can conclude is that if $a\in \anc_{G(M_{\sm L})}(b)$ or $a\in \anc_{G(M_{|X_S\in \Si})}(b)$, then $a\in \anc_{G(M)}(b)$ but not conversely.

     \item The conditioning operation for an SCM does not introduce new directed causal paths to the graph of the original SCM. This aligns with the principle that an ``individual causal effect” present in a subset of the population must also be present in the entire population, though the reverse is not necessarily true.\footnote{Note that in contrast to ``individual level'' causal effects, a ``population causal effect” in a subpopulation may not be present in the whole population due to the fact that “population causal effects” in different subpopulations may cancel out with each other.}
     
     % This aligns with the principle that an ``individual causal effect'' observed in a subset of the population must also exist in the entire population, though the reverse is not necessarily true.\footnote{Note that ``population causal effect'' in the subpopulation may not be present in the whole population due to the fact that ``population causal effects'' in different subpopulations may cancel out.}
 \end{enumerate}
\end{remark}

\begin{remark}[Remark on modeling interpretation]\label{rem:model}
  \begin{enumerate}

      \item The ``node-splitting'' trick has been applied in various forms and  situations \citep{pearl2009causality,richardson2013swig}. Usually, a copy $A^\prime$ is added as a child of $A$ to the original causal model and interventions are performed on $A^\prime$ instead of on $A$. Similarly, one can also use the node-splitting trick when using the conditioning operation where a copy $A^\prime$ is added as a parent of $A$ to the original causal model and selections are performed on $A^\prime$ instead of on $A$.

      \item There are some relations between conditioned SCMs and counterfactual reasoning. Therefore, the conditioning operation can provide an easy way to identify \emph{unnested} counterfactual quantities in some cases.\footnote{Notions defined via nested counterfactual quantity such as various notions of fairness \citep{kusner2017counterfactual,zhang2018fairness}  can always be rewritten as an unnested one using the Counterfactual Unnesting Theorem \citep[Theorem 1]{correa2021nest}.} For example, suppose that we want to identify the counterfactual quantity $\Prb_M(Y(t)\mid S=s)$ given the graph $G(M)$ in Figure \ref{fig:ex_counterfactual}. Then we can apply the graphical conditioning operation to get $G(M)_{|S}$ from $G(M)$. We know that $G(M_{|S=s})$ must be a subgraph of $G(M)_{|S}$ by \cref{prop:cond_scm_dmg}. For finite discrete variables under positivity assumptions, the back-door criterion applied to $M_{|S=s}$ gives \[
      \Prb_{M_{|S=s}}(Y(t))=\sum_{z}\Prb_{M_{|S=s}}(Y \mid T=t,Z=z)\Prb_{M_{|S=s}}(Z=z).
      \]
      Hence, we can conclude that
      \[
      \Prb_{M}(Y(t)\mid S=s)=\sum_{z}\Prb_{M}(Y \mid T=t,Z=z,S=s)\Prb_{M}(Z=z\mid S=s).
      \]
This is related to the problems of ``Type-I selection bias'' and ``internal validity'' according to the jargon of \citet{lu2022toward} and \citet{smith2020selection}, respectively.
\begin{figure}[ht]
  \centering
  \begin{tikzpicture}[scale=0.8, transform shape]
    \begin{scope}[xshift=0]
      \node[ndout] (S) at (-0.5,1) {$S$};
      \node[ndout] (T) at (1,1) {$T$};
      \node[ndout] (Z) at (2,2) {$Z$};
      \node[ndout] (Y) at (3,1) {$Y$};
      \draw[arout] (T) to (S);
      \draw[arout] (Z) to (T);
      \draw[arout] (T) to (Y);
      \draw[arout] (Z) to (Y);
      \node at (2,0) {$G(M)$};
    \end{scope}
    \begin{scope}[xshift=4cm]
      \node[ndash] (T) at (1,1) {$T$};
      \node[ndash] (Z) at (2,2) {$Z$};
      \node[ndout] (Y) at (3,1) {$Y$};
      \draw[arout] (Z) to (T);
      \draw[arout] (T) to (Y);
      \draw[arout] (Z) to (Y);
      \draw[arlat, bend left] (T) to (Z);
      \node at (2,0) {$G(M)_{|S}$};
    \end{scope}
  \end{tikzpicture}
  \caption{Causal graph $G(M)$ and its conditioned graph $G(M)_{|S}$ on $S$.}
  \label{fig:ex_counterfactual}
  \end{figure}

\end{enumerate}
\end{remark}

\begin{example}[The assumptions in Theorem \ref{thm:dsep} cannot be omitted]\label{ex:counterex_sep}

  \begin{enumerate}

    \item The assumption that $\ch_{G}(S)=\emptyset$ cannot be omitted.
    Consider the following case shown in Figure \ref{ce1}. We have $A \underset{G}{\stackrel{\sigma}{\perp}} B\mid S$ but do not have $A \underset{G_{|S}}{\stackrel{\sigma}{\perp}} B$.
    \begin{figure}[ht]
    \centering
    \begin{tikzpicture}[scale=0.8, transform shape]
      \begin{scope}[xshift=0cm]
        \node[ndout] (S) at (1.5,4) {$S$};
        \node[ndout] (A) at (0,4) {$A$};
        \node[ndout] (B) at (3,4) {$B$};
        \draw[arout] (S) to (B);
        \draw[arout] (A) to (S);
        \node at (1.5,3) {$G$};
      \end{scope}
      \begin{scope}[xshift=6cm]
        \node[ndash] (A) at (0,4) {$A$};
        \node[ndout] (B) at (1.5,4) {$B$};
        \draw[arout] (A) to (B);
        \node at (1,3) {$G_{|S}$};
      \end{scope}
    \end{tikzpicture}
    \caption{Causal graphs in the first item of Example~\ref{ex:counterex_sep}.}
    \label{ce1}
    \end{figure}

    \item  The assumption that $C\cap\anc_G(S)=\emptyset$ cannot be omitted. Consider the following case shown in Figure \ref{ce2}. We have $A \underset{G}{\stackrel{\sigma}{\perp}} B\mid C\cup S$ but we do not have $A \underset{G_{|S}}{\stackrel{\sigma}{\perp}} B\mid C$. However, note that we have $A \underset{\mathrm{MAG}(G)_{|S}}{\stackrel{d}{\perp}} B\mid C$ where $\mathrm{MAG}(G)_{|S}$ denotes the conditioned MAG of $G$ given $S$ (see \citet{richardson2002ancestral}).
    \begin{figure}[ht]
    \centering
    \begin{tikzpicture}[scale=0.8, transform shape]
      \begin{scope}[xshift=0cm]
        \node[ndout] (C) at (2,4) {$C$};
        \node[ndout] (A) at (1,3) {$A$};
        \node[ndout] (B) at (3,3) {$B$};
        \node[ndout] (S) at (2,2) {$S$};
        \draw[arout] (C) to (B);
        \draw[arout] (C) to (S);
        \draw[arout] (B) to (S);
        \draw[arout] (A) to (C);
        \node at (2,1) {$G$};
      \end{scope}
      \begin{scope}[xshift=5cm]
        \node[ndash] (C) at (2,4) {$C$};
        \node[ndash] (A) at (1,3) {$A$};
        \node[ndash] (B) at (3,3) {$B$};
        \draw[arlat, bend left] (A) to (C);
        \draw[arlat, bend left] (A) to (B);
        \draw[arlat, bend left] (C) to (B);
        \draw[arout] (C) to (B);
        \draw[arout] (A) to (C);
        \node at (2,1) {$G_{|S}$};
      \end{scope}
      \begin{scope}[xshift=10cm]
        \node[ndout] (C) at (2,4) {$C$};
        \node[ndout] (A) at (1,3) {$A$};
        \node[ndout] (B) at (3,3) {$B$};
        \draw[tut] (A) to (C);
        \draw[tut] (C) to (B);
        \node at (2,1) {$\mathrm{MAG}(G)_{|S}$};
      \end{scope}
    \end{tikzpicture}
    \caption{Causal graphs and MAGs in the second item of Example~\ref{ex:counterex_sep}.}
    \label{ce2}
    \end{figure}

    \item  The assumption that $S$ is a singleton set cannot be omitted. Consider the following case shown by Figure \ref{ce3}. For $S=\{S_1,S_2\}$, we have $A \underset{G}{\stackrel{\sigma}{\perp}} Y\mid Z\cup S$, but we do not have $A \underset{G_{|S}}{\stackrel{\sigma}{\perp}} Y\mid Z$.
    \begin{figure}[ht]
    \centering
    \begin{tikzpicture}[scale=0.8, transform shape]
      \begin{scope}[xshift=0cm]
        \node[ndout] (Z) at (2,4) {$Z$};
        \node[ndout] (A) at (0,4) {$A$};
        \node[ndout] (Y) at (3,4) {$Y$};
        \node[ndout] (S1) at (1,2.5) {$S_1$};
        \node[ndout] (S2) at (3,2.5) {$S_2$};
        \draw[arout] (Z) to (S1);
        \draw[arout] (Y) to (S2);
        \draw[arout] (A) to (S1);
        \node at (2,1.5) {$G$};
      \end{scope}
      \begin{scope}[xshift=6cm]
        \node[ndash] (Z) at (1.5,4) {$Z$};
        \node[ndash] (A) at (0,3) {$A$};
        \node[ndash] (Y) at (3,3) {$Y$};
        \draw[arlat, bend left] (A) to (Z);
        \draw[arlat, bend left] (A) to (Y);
        \draw[arlat, bend left] (Z) to (Y);
        \node at (1.5,1.5) {$G_{|S}$};
      \end{scope}
    \end{tikzpicture}
    \caption{When $S=\{S_1,S_2\}$ is not a singleton set, the $\sigma$-separation $A \underset{G}{\stackrel{\sigma}{\perp}} Y\mid Z\cup S$ does not imply $A \underset{G_{|S}}{\stackrel{\sigma}{\perp}} Y\mid Z$.}
    \label{ce3}
    \end{figure}

    % \item  The assumption that $S$ is a singleton set is necessary. Consider the following case shown by Figure \ref{ce3}. For $S=\{S_1,S_2\}$, we have $A \underset{G}{\stackrel{\sigma}{\perp}} B\mid C\cup S$, but do not have $A \underset{G_{|S}}{\stackrel{\sigma}{\perp}} B\mid C$.
    % \begin{figure}[ht]
    % \centering
    % \begin{tikzpicture}[scale=0.8, transform shape]
    %   \begin{scope}[xshift=0cm]
    %     \node[ndout] (C) at (2,4) {$C$};
    %     \node[ndout] (A) at (0,4) {$A$};
    %     \node[ndout] (B) at (3,4) {$B$};
    %     \node[ndout] (S1) at (1,2.5) {$S_1$};
    %     \node[ndout] (S2) at (3,2.5) {$S_2$};
    %     \draw[arout] (C) to (S1);
    %     \draw[arout] (B) to (S2);
    %     \draw[arout] (A) to (S1);
    %     \node at (2,1.5) {$G$};
    %   \end{scope}
    %   \begin{scope}[xshift=6cm]
    %     \node[ndout] (C) at (1.5,4) {$C$};
    %     \node[ndout] (A) at (0,3) {$A$};
    %     \node[ndout] (B) at (3,3) {$B$};
    %     \draw[arlat, bend left] (A) to (C);
    %     \draw[arlat, bend left] (A) to (B);
    %     \draw[arlat, bend left] (C) to (B);
    %     \node at (1.5,1.5) {$G_{|S}$};
    %   \end{scope}
    % \end{tikzpicture}
    % \caption{When $S=\{S_1,S_2\}$ is not a singleton set, the $\sigma$-separation $A \underset{G}{\stackrel{\sigma}{\perp}} B\mid C\cup S$ dose not imply $A \underset{G_{|S}}{\stackrel{\sigma}{\perp}} B\mid C$.}
    % \label{ce3}
    % \end{figure}
  \end{enumerate}
\end{example}

\begin{example}[Markov property does not imply conditional independence given an event]\label{ex:counterex_ci}
Consider a discrete acyclic causal model $M$ given by
\[
\begin{aligned}
    &\Prb_M(X_A\mid X_S=0)=\frac{1}{2}\delta_0+\frac{1}{2}\delta_1, \quad \Prb_M(X_B\mid X_S=0)=\frac{1}{2}\delta_0+\frac{1}{2}\delta_1,\\
    &\Prb_M(X_A\mid X_S=1)=\frac{1}{2}\delta_0+\frac{1}{2}\delta_1, \quad \Prb_M(X_B\mid X_S=1)=\frac{1}{2}\delta_0+\frac{1}{2}\delta_1,\\
    &\Prb_M(X_A\mid X_S=2)=\frac{1}{3}\delta_0+\frac{2}{3}\delta_1, \quad \Prb_M(X_B\mid X_S=2)=\frac{1}{3}\delta_0+\frac{2}{3}\delta_1, \\
    &\Prb_M(X_S)=\frac{1}{3}\delta_0+\frac{1}{3}\delta_1+\frac{1}{3}\delta_2,
\end{aligned}
\]
and $\Prb_M(X_A,X_B,X_S)=\Prb_M(X_A\mid X_S)\otimes \Prb_M(X_B\mid X_S)\otimes \Prb_M(X_S)$
whose graph is drawn in Figure \ref{fig:counterexample_ci}.

\begin{figure}[ht]
\centering
\begin{tikzpicture}[scale=0.8, transform shape]  
    \node[ndout] (C) at (2,4) {$S$};
    \node[ndout] (A) at (1,3) {$A$};
    \node[ndout] (B) at (3,3) {$B$};
    \draw[arout] (C) to (A);
    \draw[arout] (C) to (B);
    \node at (2,2) {$G(M)$};
\end{tikzpicture}
\caption{Causal graph of $M$ where $X_A\underset{\Prb_M(X_V)}{\notind} X_B\mid X_S\in \{1,2\}$ even if $A\underset{G(M)}{\stackrel{d}{\perp}}  B\mid S$.}
\label{fig:counterexample_ci}
\end{figure}

We can compute that
\[
\begin{aligned}
    \Prb_M(X_A\mid X_S\in \{1,2\})&=\frac{\Prb_M(X_A, X_S\in \{1,2\})}{\Prb_M(X_S\in \{1,2\})}=\frac{5}{12}\delta_0+\frac{7}{12}\delta_1,\\
    \Prb_M(X_B\mid X_S\in \{1,2\})&=\frac{\Prb_M(X_B, X_S\in \{1,2\})}{\Prb_M(X_S\in \{1,2\})}=\frac{5}{12}\delta_0+\frac{7}{12}\delta_1,\\
    \Prb_M(X_A,X_B\mid X_S\in \{1,2\})&=\frac{\Prb_M(X_A,X_B, X_S\in \{1,2\})}{\Prb_M(X_S\in \{1,2\})}=\frac{13}{72}\delta_{00}+\frac{17}{72}\delta_{01}+\frac{17}{72}\delta_{10}+\frac{25}{72}\delta_{11}.
\end{aligned}
\]
Then it is easy to see that
\[
\begin{aligned}
   &\Prb_M(X_A\mid X_S\in\{1,2\})\otimes \Prb_M(X_B\mid X_S\in \{1,2\})\\
   &=\frac{25}{144}\delta_{00}+\frac{35}{144}\delta_{01}+\frac{35}{144}\delta_{10}+\frac{49}{144}\delta_{11}\\
   &\ne \Prb_M(X_A,X_B\mid X_S\in \{1,2\}).
\end{aligned}
\]
Note that $A\underset{G(M)}{\stackrel{d}{\perp}}  B\mid S$. From the Markov property or direct calculation, one can see that $X_A\underset{\Prb_M(X_V)}{\ind} X_B\mid X_S$. However, the above calculation implies that $X_A\underset{\Prb_M(X_V)}{\notind} X_B\mid X_S\in \{1,2\}$.

\end{example}

\begin{remark}[Remark on \cref{ex:counterex_ci}]
    Note that one \emph{cannot} say that this falsifies the Markov property. In fact, the Markov property is about conditioning on a variable and the above example is about conditioning on an event, which are fundamentally different.
\end{remark}

\section{Proofs}\label{app:pf}

\partition*
\begin{proof}
     We first show that $(\Pf_{\Si}, \vee)$ is a finite join semi-lattice where $\Ic\vee \Jc\coloneqq \{I\cap J: I\in \Ic \text{ and } J\in \Jc\}\sm \{\emptyset\}$. To achieve that, it suffices to show that $\Pf_{\Si}$ is closed under the join operation. If $\{X_{I_i}\}_{i=1}^p$ are mutually independent and $\{X_{J_j}\}_{j=1}^q$ are mutually independent under some probability distribution $\tilde{\Prb}$ then we have that $\{X_{K_k}\}_{k=1}^m$ are mutually independent under $\tilde{\Prb}$ where $\{K_k\}_{k=1}^m=\Ic\vee \Jc$. That is, $\tilde{\Prb}(X_W)=\bigotimes_{i=1}^p\tilde{\Prb}(X_{I_i})=\bigotimes_{i=1}^p \bigotimes_{j=1}^q\tilde{\Prb}(X_{I_i\cap J_j})=\bigotimes_{k=1}^m \tilde{\Prb}(X_{K_k}) $. Since $(\Pf_{\Si}, \vee)$ is finite, there must exist a largest element, which we denote by $\Hc=\{H_i\}_{i=1}^n$. This means that every partition $\Jc$ in $\Pf_{\Si}$ must be coarser than $\Hc$ according to the order induced by the join $\vee$, so the partition $\Hc$ is the finest partition in $\Pf_{\Si}$. 
 \end{proof}

\modelclass*
 \begin{proof}
Since exogenous random variables do not have parents, merging exogenous random variables will not introduce cycles. Merging exogenous random variables will also preserve the simplicity and linearity of SCMs. Indeed, if $g^A:\Xc_{V\sm A}\times \Xc_W\rightarrow \Xc_A$ is the essentially unique solution function of $M$ \wrt $A$ for some $A\subseteq V\sm S$, then the function $\tilde{g}^A:\Xc_{V\sm A}\times \Xc_{\hat{W}}\rightarrow \Xc_A$ defined by $\tilde{g}^A(x_{V\sm A},x_{\hat{W}})=g^A(x_{V\sm A},x_{W})$ with $x_W=(x_{\hat{w}})_{\hat{w}\in \hat{W}}=x_{\hat{W}}\in \Xc_{\hat{W}}$ is the essentially unique solution function of $\tilde{M}$ \wrt $A$, where $\tilde{M}$ is the same as $M$ but with $W$ replaced by $\hat{W}$ and $\Xc$ by $\Xc_V\times \Xc_{\hat{W}}$. So merging exogenous variables preserves simplicity. For linearity, let $\Xc_{u}$ and $\Xc_v$ denote some linear vector spaces (they do not have to be the real line $\Rb$), and $\Lc(\Xc_u, \Xc_v)$ denote the set of linear mappings from $\Xc_u$ to $\Xc_v$. Let $f_v(x_V,x_W)=\sum_{u\in V}(T_{vu}(x_W))x_u+\Gamma_v(x_W)$ where $T_{vu}:\Xc_{W}\rightarrow \Lc(\Xc_u, \Xc_v)$ and $\Gamma_v:\Xc_W\rightarrow \Xc_v$ are (nonlinear) mappings such that $f_v$ is measurable. Then the mapping $\tilde{f}_v(x_V,x_{\hat{W}})=\sum_{u\in V}(\tilde{T}_{vu}(x_{\hat{W}}))x_u+\tilde{\Gamma}_v(x_{\hat{W}})$ is still linear and measurable where
\[
\tilde{T}_{vu}(x_{\hat{W}})\coloneqq
        T_{vu}(x_{W}) \text{ and } \tilde{\Gamma}_v(x_{\hat{W}})\coloneqq
        \Gamma_v(x_{W})  \text{ with } x_{\hat{W}}=x_W.
\]

Updating the probability distributions of the exogenous random variables to the posterior preserves simplicity, acyclicity, and linearity of SCMs. In particular, this preserves simplicity because if $\Prb_M(X_S\in \Si)>0$, then $\Prb_M(X_W\mid X_S\in \Si)\ll\Prb_M(X_W)$. By slightly generalizing \citet[Propositions 8.2, 5.11 and C.5]{bongers2021foundations},\footnote{Our definition of linear SCMs is more general than the one in \citet{bongers2021foundations}.} we have that marginalization preserves simplicity, acyclicity, and linearity of SCMs. Hence, we obtain that the conditioning operation preserves simplicity, acyclicity, and linearity of SCMs.
\end{proof}

\intcond*
\begin{proof}
     We check the definition one by one. Write $\left(M_{|\Si}\right)_{\Do(X_T=x_T)}\coloneqq (\hat{V},\hat{W},\hat{\Xc},\hat{\Prb},\hat{f})$ and $\left(M_{\Do(X_T=x_T)}\right)_{|\Si}\coloneqq (\ol{V},\ol{W},\ol{\Xc},\ol{\Prb},\ol{f})$. Set $O\coloneq V\sm S$.

     First, it is easy to see that $\hat{V}=V \sm S=\ol{V}$. Because $T\cap \anc_{G(M)}(S)=\emptyset$ and $M$ is simple, $g_S^{-1}(\Si)=\tilde{g}_S^{-1}(\Si)$ up to a $\Prb(X_W)$-null set where $g$ and $\tilde{g}$ are solution functions of $M$ and $M_{\Do(X_T=x_T)}$ respectively. Therefore, we can conclude that $\hat{W}=\ol{W}$. Then we have $\hat{\Xc}=\Xc_O\times \Xc_{\hat{W}}=\Xc_O\times \Xc_{\ol{W}}=\ol{\Xc}$. Since $M$ and $M_{\Do(X_T=x_T)}$ have the same exogenous distribution $\Prb$ and $g_S^{-1}(\Si)=\tilde{g}_S^{-1}(\Si)$ up to a $\Prb(X_W)$-null set, we have $\Prb_M(X_{\hat{w}_i},X_S\in\Si)=\Prb_{M_{\Do(X_T=x_T)}}(X_{\ol{W}_i},X_S\in\Si)$. Hence, we can conclude that
     \[
      \Prb_M(X_{\hat{w}_i}\mid X_S\in \Si)=\Prb_{M_{\Do(X_T=x_T)}}(X_{\ol{W}_i}\mid X_S\in \Si).
     \]
     Therefore, we have $\hat{\Prb}=\bigotimes_{i=1}^n\hat{\Prb}(X_{\hat{w}_i})=\bigotimes_{i=1}^n\ol{\Prb}(X_{\ol{W}_i})=\ol{\Prb}$. For the causal mechanisms, we have $\hat{f}\left(x_{\hat{V}}, x_{\hat{W}}\right)=\left(f_{O\sm T}\left(x_O, g^S\left(x_O, x_{\hat{W}}\right), x_{\hat{W}}\right),x_T\right)$ with $g^S$ the (essentially unique) solution function of $M$ \wrt $S$. Let $\tilde{f}$ be the causal mechanism of $M_{\Do(X_T=x_T)}$ and let $\tilde{g}^S$ be the (essentially unique) solution function of $M_{\Do(X_T=x_T)}$ \wrt $S$. Then we have
     \[
        \begin{aligned}
            \ol{f}\left(x_{\ol{V}}, x_{\ol{W}}\right)&=\tilde{f}_{O}\left( x_O, \tilde{g}^S\left( x_O, x_{\ol{W}}\right), x_{\ol{W}}\right)\\
            &=\left(f_{O\sm T}\left(x_O, \tilde{g}^S\left(x_O, x_{\ol{W}}\right), x_{\ol{W}}\right),x_T\right).
        \end{aligned}
     \]
     Since $T\cap S=\emptyset$, we have $f_S(x_V,x_W)=\tilde{f}_S(x_V,x_W)$ for all $x_W\in\Xc_W$ and all $x_{V}\in \Xc_{V}$. Recall that $g^S$ and $\tilde{g}^S$ are the (essentially unique) solution functions of $M$ and $M_{\Do(X_T=x_T)}$ \wrt $S$ respectively, i.e., for $\Prb(X_W)$-a.a.\ $x_W\in\Xc_W$ and all $x_{V}\in \Xc_{V}$
  \[
    x_S=g^S(x_O,x_{W})\Longleftrightarrow x_S=f_S(x_V,x_W) \text{ and } x_S=\tilde{g}^S(x_O,x_{W})\Longleftrightarrow x_S=\tilde{f}_S(x_V,x_W).
  \]
  Therefore, for $\Prb(X_{W}\mid X_S\in \Si)$-a.a.\ $x_{\hat{W}}=x_{\ol{W}}\in\Xc_{\hat{W}}=\Xc_{\ol{W}}$ and all $x_{O}\in \Xc_{O}$
  \[
  g^{S}(x_{O},x_{\hat{W}})=\tilde{g}^S(x_{O},x_{\ol{W}}).
  \]
  Hence, for $\Prb(X_{W}\mid X_S\in \Si)$-a.a.\ $x_{\hat{W}}=x_{\ol{W}}\in\Xc_{\hat{W}}=\Xc_{\ol{W}}$ and all $x_{O}\in \Xc_{O}$
  \[
  \hat{f}(x_O,x_{\hat{W}})=\ol{f}(x_O,x_{\ol{W}}).
  \]
  With the definition of the conditioning operation we conclude
\[
\left(M_{\Do(X_T=x_T)}\right)_{|X_S\in\Si}\equiv\left(M_{|X_S\in\Si}\right)_{\Do(X_T=x_T)}.
\]
\end{proof}

\causalsemantics*
\begin{proof}
    Let $T\subseteq O$.
    Let $\tilde{g}^{O\sm T}:\Xc_T \times \Xc_W \rightarrow \Xc_{O\sm T}$ be the (essentially unique) solution function of $M_{\sm S}$ \wrt $O\sm T$ and $\hat{g}^{O\sm T}:\Xc_T\times \Xc_{\hat{W}}\rightarrow \Xc_{O\sm T}$ be the (essentially unique) solution function of $M_{|\Si}$ \wrt $O\sm T$. For $\Prb_M(X_W\mid X_S\in \Si)$-a.a\ $x_{\hat{W}}=x_W\in\Xc_W$ and all $x_V \in \Xc_V$, we have
    \[
        \begin{aligned}
            x_{O\sm T}=\tilde{g}^{O\sm T} \left(x_T, x_W\right)
            & \Longleftrightarrow x_{O\sm T}=f_{O\sm T}\left(x_O, g^S\left(x_O, x_W\right), x_W\right) \\
            & \Longleftrightarrow x_{O\sm T}=\hat{f}_{O\sm T}\left(x_O, x_W\right) \\
            & \Longleftrightarrow x_{O\sm T}=\hat{g}^{O\sm T}\left(x_T,x_{\hat{W}}\right)
        \end{aligned}
    \]
    This implies that $\tilde{g}^{O\sm T}(x_W,x_T)=\hat{g}^{O\sm T}(x_{\hat{W}},x_T)$ for $\Prb_M(X_W\mid X_S\in\Si)$-a.a.\ $x_{\hat{W}}=x_W\in \Xc_W$ and all $x_{T}\in \Xc_{T}$. Hence, we have
    \[
    \begin{aligned}
        \Prb_{M_{|\Si}}\left(\{X_{O\sm T_i}(x_{T_i})\}_{1\leq i\leq n}\right)&= \left(\hat{g}^{O\sm T_1}(x_{T_1},\cdot),\ldots, \hat{g}^{O\sm T_n}(x_{T_n},\cdot)\right)_*\Prb_{M_{|\Si}}(X_{\hat{W}})\\
        &=\left(\tilde{g}^{O\sm T_1}(x_{T_1},\cdot),\ldots, \tilde{g}^{O\sm T_n}(x_{T_n},\cdot)\right)_*\Prb_{M}(X_{W}\mid X_S\in\Si)\\
        &=\left(g^{V\sm T_1}_{O\sm T_1}(x_{T_1},\cdot),\ldots, g^{V\sm T_n}_{O\sm T_n}(x_{T_n},\cdot)\right)_*\Prb_{M}(X_{W}\mid X_S\in\Si)\\
        &=\Prb_M\left(\{X_{O\sm T_i}(x_{T_i})\}_{1\leq i\leq n}\mid X_S\in\Si\right),
    \end{aligned}
    \]
    since $g^{V\sm T_i}_{O\sm T_i}(x_W,x_{T_i})=\tilde{g}^{O\sm T_i}(x_W,x_{T_i})$ for $\Prb_M(X_W)$-a.a.\ $x_W\in \Xc_W$ and all $x_{T_i}\in \Xc_{T_i}$ where $g^{V\sm T_i}:\Xc_{T_i\cup S} \times \Xc_W \rightarrow \Xc_{V\sm T_i}$ is the (essentially unique) solution function of $M$ \wrt $V\sm T_i$ by \citet[Lemma~6.8.4]{forre25causality}.
\end{proof}

\optimal*
\begin{proof}
    The proof is obvious by noting that 
\[
    \begin{aligned}
        \Prb_M\left(X_{O\setminus T}\mid \Do(X_{T}=x_{T}),X_S\in \Si\right) &=\Prb_{\tilde{M}}\left(X_{O\setminus T}\mid \Do(X_T=x_T)\right)\\
        &=\Prb_{\tilde{M}}\left(X_{O\setminus T}(x_T)\right)\\
        &=\Prb_M(X_{O\sm T}(x_T)\mid X_S\in \Si)\\
        &=\Prb_{M_{|X_S\in \Si}}\left(X_{O\setminus T}(x_T)\right)\\
        &=\Prb_{M_{|X_S\in \Si}}\left(X_{O\setminus T}\mid \Do(X_T=x_T)\right).
    \end{aligned}
\]
\end{proof}

\condmarg*
\begin{proof}
    Write $(M_{|X_S\in\Si})_{\setminus L}=(\hat{V},\hat{W},\hat{\Xc},\hat{\Prb},\hat{f})$ and $(M_{\setminus L})_{|X_S\in \Si}=(\ol{V},\ol{W},\ol{\Xc},\ol{\Prb},\ol{f})$.

    First, it is easy to see that $\hat{V}=\ol{V}=V\sm(L\cup S)$.
    Let $\tilde{g}:\Xc_W\rightarrow \Xc_{V\sm L}$ be the (essentially unique) solution function of $M_{\sm L}$. Then we have $g_{S}=\tilde{g}_S$ $\Prb(X_W)$-a.s. Therefore, by the definition of $\hat{W}$ and $\ol{W}$, we have $\hat{W}=\ol{W}$ and $\hat{\Prb}(X_{\hat{W}})=\ol{\Prb}(X_{\ol{W}})=\Prb_M(X_W\mid X_S\in \Si)$. Furthermore, we have $\hat{\Xc}=\ol{\Xc}$. By \citet[Proposition 5.4]{bongers2021foundations}, we have $\hat{f}(x_{V\sm (L\cup S)},x_{\hat{W}})=\ol{f}(x_{V\sm (L\cup S)},x_{\ol{W}})$  for $\Prb_M(X_W\mid X_S\in \Si)$-a.a.\ $x_{\hat{W}}=x_{\ol{W}}\in \Xc_{\hat{W}}$ and all $x_{V\sm (L\cup S)}\in \Xc_{V\sm (L\cup S)}$. Overall, $(M_{\setminus L})_{|X_S\in \Si}\equiv (M_{|X_S\in\Si})_{\setminus L}$.
\end{proof}

 \iterativecond*
\begin{proof}
    It is easy to see that $(M_{|\Si_1})_{|\Si_2}$, $(M_{|\Si_2})_{|\Si_1}$, and $M_{|\Si_1\times\Si_2}$ are well defined given the assumption above.
    Write $O:=V\sm(S_1\cup S_2)$. The counterfactual and potential-outcome equivalence among the three SCMs can be deduced by the following observation:
    \[
        \begin{aligned}
            \Prb_{M_{|\Si_1\times \Si_2}}\big(\{X_O(x_{T_i})\}_{i=1}^n\big)&=\Prb_{M}\big(\{X_O(x_{T_i})\}_{i=1}^n\mid X_S\in \Sc_1\times \Si_2\big)\\
            &=\Prb_{(M_{|\Si_1})_{|\Si_2}}\big(\{X_O(x_{T_i})\}_{i=1}^n\big)\\
            &=\Prb_{(M_{|\Si_2})_{|\Si_1}}\big(\{X_O(x_{T_i})\}_{i=1}^n\big).
        \end{aligned}
    \]

    Now we show that $G(M_{|\Si_1\times \Si_2})$ is a subgraph of $G((M_{|\Si_1})_{|\Si_2})$ and $G((M_{|\Si_2})_{|\Si_1})$. By \cref{lem:indicator}, we can find a simple SCM $\widetilde{M}$ such that $\sib_{G(\widetilde{M})}(\widetilde{S}_1\cup \widetilde{S}_2)\cup \ch_{G(\widetilde{M})}(\widetilde{S}_1\cup \widetilde{S}_2)=\emptyset$ and $M_{|\Si_1\times \Si_2}\equiv (\widetilde{M}_{|\widetilde{\Si}_1\times \widetilde{\Si}_2})_{\sm S_1\cup S_2}$ and $(M_{|\Si_1})_{|\Si_2}\equiv (((\widetilde{M}_{|\widetilde{\Si}_1})_{\sm S_1})_{|\widetilde{\Si}_2})_{\sm S_2}$. We have by \cref{prop:cond_marg}
    \[
        (((\widetilde{M}_{|\widetilde{\Si}_1})_{\sm S_1})_{|\widetilde{\Si}_2})_{\sm S_2}\equiv (((\widetilde{M}_{|\widetilde{\Si}_1})_{|\widetilde{\Si}_2})_{\sm S_1})_{\sm S_2}\equiv ((\widetilde{M}_{|\widetilde{\Si}_1})_{|\widetilde{\Si}_2})_{\sm S_1\cup S_2}.
    \]
    It suffices to show that $G((\widetilde{M}_{|\widetilde{\Si}_1\times \widetilde{\Si}_2})_{\sm S_1\cup S_2})$ is a subgraph of $G(((\widetilde{M}_{|\widetilde{\Si}_1})_{|\widetilde{\Si}_2})_{\sm S_1\cup S_2})$. Write $M^\prime\coloneqq \widetilde{M}_{|\widetilde{\Si}_1\times \widetilde{\Si}_2}$ and $M^{\prime\prime}\coloneqq (\widetilde{M}_{|\widetilde{\Si}_1})_{|\widetilde{\Si}_2}$. By definition, $M^{\prime}$ and $M^{\prime\prime}$ have the same causal mechanisms and the same exogenous distribution, which is $\Prb_M(X_W\mid X_{S_1}\in \Si_1,X_{S_2}\in \Si_2)$. Therefore, we can use the same solution functions \wrt $S_1\cup S_2$ for marginalizations. Hence, the directed edges in their graphs coincide. Note that the exogenous nodes of $\widetilde{M}_{|\widetilde{\Si}_1\times \widetilde{\Si}_2}$ have the finest partition of $W$ given $X_{S_1}\in \Si_1$ and $X_{S_2}\in \Si_2$. Therefore, the bidirected edges of $\widetilde{M}_{|\widetilde{\Si}_1\times \widetilde{\Si}_2}$ are a subset of those of $(\widetilde{M}_{|\widetilde{\Si}_1})_{|\widetilde{\Si}_2}$. This finishes the proof.

    Finally, we show that, under the conditions of the proposition, we have $(M_{|\Si_1})_{|\Si_2}\equiv (M_{|\Si_2})_{|\Si_1}\equiv M_{|\Si_1\times \Si_2}$. Write
    \[
    \begin{aligned}
    (M_{|\Si_1})_{|\Si_2}&=(V^{12},W^{12},\Xc^{12},f^{12},\Prb^{12})\\
    (M_{|\Si_2})_{|\Si_1}&=(V^{21},W^{21},\Xc^{21},f^{21},\Prb^{21})\\
    M_{|\Si_1\times\Si_2}&=(V^{1\times 2},W^{1\times 2},\Xc^{1\times 2},f^{1\times 2},\Prb^{1\times 2}).
    \end{aligned}
    \]

    First, note that to establish $(M_{|\Si_1})_{|\Si_2}\equiv (M_{|\Si_2})_{|\Si_1}\equiv M_{|\Si_1\times \Si_2}$, it suffices to show that $W^{12}=W^{21}=W^{1\times 2}$.\footnote{Strictly speaking they are not exactly equal to each other but are isomorphic. For simplicity, we see being isomorphic as equal.} Indeed, it is easy to see that $V^{12}=V^{21}=V^{1\times 2}$. Given that $W^{12}=W^{21}=W^{1\times 2}$, we have $\Xc^{12}=\Xc^{21}=\Xc^{1\times 2}$ and $\Prb^{12}=\Prb^{21}=\Prb^{1\times 2}$. By the properties of marginalization, we have $f^{12}=f^{21}=f^{1\times 2}$ $\Prb^{12}$-a.s..

    By the property of SCMs, there exist measurable functions $\tilde{g}_{S_1}:\Xc_{\anc_{G^a(M_{\sm (V\sm S_1)})}(S_1)\cap W}\to \Xc_{S_1}$ and $\tilde{g}_{S_2}:\Xc_{\anc_{G^a(M_{\sm (V\sm S_2)})}(S_2)\cap W}\to \Xc_{S_2}$ such that $g_{S_1}(x_W)=\tilde{g}_{S_1}(x_{\anc_{G^a(M_{\sm (V\sm S_1)})}(S_1)\cap W})$ and $g_{S_2}(x_W)=\tilde{g}_{S_2}(x_{\anc_{G^a(M_{\sm (V\sm S_2)})}(S_2)\cap W})$ for $\Prb(X_W)$-a.a.\ $x_W\in \Xc_W$. Since $\anc_{G^a(M_{\sm (V\sm S_1)})}(S_1)\cap \anc_{G^a(M_{\sm (V\sm S_2)})}(S_2)=\emptyset$, we have $W^{12}=W^{21}$. To see $W^{1\times 2}=W^{12}$, it suffices to note that $g^{-1}_{S_1\cup S_2}(\Si_1\times \Si_2)=g^{-1}_{S_1}(\Si_1)\cap g^{-1}_{S_2}(\Si_2)$. If $\Prb(X_W\in (g^{-1}_{S_1}(\Si_1)\sd g^{-1}_{S_2}(\Si_2)))=0$, then we have $g^{-1}_{S_1}(\Si_1)\pequ g^{-1}_{S_2}(\Si_2)\pequ g^{-1}_{S_1}(\Si_1)\cap g^{-1}_{S_2}(\Si_2)= g^{-1}_{S_1\cup S_2}(\Si_1\times \Si_2)$ where $\pequ$ denotes equality up to a null set. This implies $W^{12}=W^{21}=W^{1\times 2}$.

\end{proof}

\indicator*
\begin{proof}
First note that $\tilde{M}$ defined above is a simple SCM and $\Prb_{\tilde{M}}(X_{\tilde{S}}=1)=\Prb_{M}(X_{S}\in \Si)>0$. We declare some notation. We write $M_{|X_S\in\Si}=(\hat{V},\hat{W},\hat{\Xc},\hat{\Prb},\hat{f})$ and $(\tilde{M}_{|X_{\tilde{S}}=1})_{\sm S}=(\ol{V},\ol{W},\ol{\Xc},\ol{\Prb},\ol{f})$. It is easy to see that $\hat{V}=\ol{V}=V\sm S$ and $\hat{\Xc}=\ol{\Xc}=\Xc_{V\sm S}$. Also note that $g_S(\Si)^{-1}=\tilde{g}_{\tilde{S}}(1)^{-1}$ up to a $\Prb$-null set where $g$ and $\tilde{g}$ are the (essentially unique) solution functions of $M$ and $\tilde{M}$ respectively. Then we can conclude that $\hat{W}=\ol{W}$ and also $\hat{\Prb}=\ol{\Prb}=\Prb_M(X_W\mid X_S\in \Si)$. It is also obvious that $\hat{f}(x_{V\sm S},x_{\hat{W}})=\ol{f}(x_{V\sm S},x_{\ol{W}})=f_{V\sm S}(x_{V\sm S},g^{S}(X_{V\sm S},x_{W}),x_W)$ where $g^S$ is the (essentially unique) solution function of $M$ \wrt $S$, for $\Prb_M(X_W\mid X_S\in \Si)$-a.a.\ $x_{\ol{W}}=x_{\hat{W}}=x_W\in\Xc_W$ and all $x_{V\sm S}\in \Xc_{V\sm S}$. Overall, we have $M_{|X_S\in\Si}\equiv (\tilde{M}_{|X_{\tilde{S}}=1})_{\sm S}$.
\end{proof}

\dsep*

\begin{proof}
  We show the first statement. By \cref{prop:prop_graph_cond}, we can assume WLOG that $S=\{s\}$ is a singleton set. To show
  \[
    A \underset{G_{|S}}{\stackrel{\sigma}{\perp}} B\mid C \quad \Longrightarrow \quad A \underset{G}{\stackrel{\sigma}{\perp}} B\mid C\cup S,
  \]
  it suffices to show
  \[
    A \underset{G}{\stackrel{\sigma}{\notsep}} B\mid C\cup S \quad \Longrightarrow \quad A \underset{G_{|S}}{\stackrel{\sigma}{\notsep}} B\mid C.
  \]

    Let $\pi:v_0\sus \cdots \sus v_n$ be a $\sigma$-open walk between $A$ and $B$ given $C\cup S$ in $G$ such that all the colliders are in $C\cup S$ \citep[Proposition~3.3.6]{forre25causality}.  WLOG we can assume that $s$ appears on $\pi$ at most once. Indeed, assume that $\pi$ is of the form $v_0\sus \cdots \sus v_i\sus s \sus \cdots \sus s\sus v_j\sus \cdots v_n$ and the subwalks $\pi(v_0,v_i)$ and $\pi(v_j,v_n)$ do not contain $s$. If we have $v_i\suh s $ and $s\hus v_j$, then the walk $\pi(v_0,v_i)\oplus \pi(v_i,s)\oplus \pi(s,v_j)\oplus \pi(v_j,n)$ is $\sigma$-open given $C\cup S$. If we have $v_i\hut s$ and $s \hus v_j$, then $v_i$ must be in the same strongly connected component of $s$ since otherwise $\pi$ is blocked by $S$. Therefore,  the walk $\pi(v_0,v_i)\oplus \pi(v_i,s)\oplus \pi(s,v_j)\oplus \pi(v_j,n)$ is $\sigma$-open given $C\cup S$. The same conclusion holds if $v_i\suh s$ and $s\tuh v_j$, and likewise if $v_i\hut s$ and $s\tuh v_j$. If $s$ occurs on $\pi$ as a non-collider, i.e., $v_{i-1}\suh s \tuh v_{i+1}$ or $v_{i-1}\hut s \tuh v_{i+1}$ or $v_{i-1}\hut s \hus v_{i+1}$, then, by replacing the segments by $v_{i-1}\suh v_{i+1}$ or $v_{i-1}\huh v_{i+1}$ or $v_{i-1}\hus v_{i+1}$ respectively and keeping other parts of $\pi$ intact, we have a $\sigma$-open walk from $A$ to $B$ given $C$ in $G_{|S}$. So consider the case where $s$ does not occur on $\pi$ as a non-collider in the following. If all colliders on $\pi$ are in $C$ and $\pi$ does not contain $s$, then $\pi$ is $\sigma$-open between $A$ and $B$ given $C$ in $G_{|S}$. We consider the case where $s$ occurs on $\pi$ as collider. Define $v_l$ and $v_r$ to be the most left and most right nodes on $\pi$ respectively that are in $\anc_{G}(S)\cup \sib_G(\anc_G(S))$. So $\pi$ is of the form $v_0\sus \cdots \sus v_l \suh \cdots \tuh s\hut \cdots \hus v_r \sus \cdots \sus v_n$. We replace the subwalk $\pi(v_l,v_r)$ with $v_l\huh v_r$ on $\pi$ to construct a new walk $\tilde{\pi}$. Note that $v_l$ cannot become a collider on $\tilde{\pi}$ if it is a non-collider on $\pi$ and similarly for $v_r$. If $v_l$ is a collider on $\tilde{\pi}$, it must be in $C$. Similarly for $v_r$. If $v_l$ is a non-collider on $\pi$, it must be unblockable or not in $C$. If it is unblockable on $\pi$, it remains so on $\tilde{\pi}$. Similarly for $v_r$. Hence, $\tilde{\pi}$ is $\sigma$-open in $G_{|S}$ from $A$ to $B$ given $C$.

Overall, given any $\sigma$-open walk between $A$ and $B$ given $C\cup S$ in $G$, we can find a $\sigma$-open walk between $A$ and $B$ given $C$ in $G_{|S}$. This shows that
\[
  A \underset{G}{\stackrel{\sigma}{\notsep}} B\mid C\cup S \quad \Longrightarrow \quad A \underset{G_{|S}}{\stackrel{\sigma}{\notsep}} B\mid C.
\]

We now show
\[
  A \underset{G}{\stackrel{\sigma}{\perp}} B\mid C\cup S \quad \Longrightarrow \quad A \underset{G_{|S}}{\stackrel{\sigma}{\perp}} B\mid C
\]
under the conditions that $C\cap\anc_G(S)=\emptyset$, $\ch_G(S)=\emptyset$, and $S$ is a singleton set. Assume that $\pi$ is a $\sigma$-open walk between $A$ and $B$ given $C$ in $G_{|S}$ such that all the colliders are in $C$. We shall construct a $\sigma$-open walk between $A$ and $B$ given $C\cup S$ in $G$. 

For all edges in $\pi$ that are also in $G$, we keep them untouched. Note that since $\ch_G(S)=\emptyset$, there do not exist directed edges in $G_{|S}$ that are not in $G$. Therefore, if all bidirected edges on $\pi$ in $G_{|S}$ are also in $G$, then we are done. Thus, we are left with the case where there are some bidirected edges $v_i\huh v_{i+1}$ of $\pi$ not in $G$. Note that since $S$ is singleton and the possibilities of $v_i\huh s \tuh v_{i+1}$ and $v_i\hut s \huh v_{i+1}$ and $v_i\hut s \tuh v_{i+1}$ are excluded by the assumption $\ch_G(S)=\emptyset$, we can replace those bidirected edges $v_i\huh v_{i+1}$ with $v_{i}\suh w_{1} \tuh \cdots \tuh w_{k}\tuh s \hut w_{k+1}\hut\cdots\hut w_m \hus v_{i+1}$ ($k$ could be zero and $m$ could be $k$) and get a new walk $\tilde{\pi}$ in $G$. Next, we show that $\tilde{\pi}$ cannot be blocked by
$C\cup S$ at $v_i$ or $v_{i+1}$ in $G$.

If $v_i$ or $v_{i+1}$ is an endnode of $\pi$, then $\tilde{\pi}$ cannot be blocked at that node in $G$. Thus, we assume that $v_{i}$ and $v_{i+1}$ are not endnodes of $\pi$. Suppose that in $\pi$ we have $\cdots\suh v_i\huh v_{i+1}\ \cdots$. Since $v_i$ is a collider on $\pi$, by the assumption on $\pi$, we have $v_i\in C$.  Then it is impossible to have
$\cdots\suh v_{i}\tuh w_{1} \tuh \cdots \tuh w_{k}\tuh s $ in $G$, since in this case we would have $\anc_{G}(S)\cap C\neq \emptyset$.
Hence, we must have $\cdots\suh v_{i}\huh w_{1} \tuh \cdots w_{k}\tuh s $, and then we know that this walk is $\sigma$-open given $C$ at $v_i$ in $G$.

Now we suppose $\cdots \hut v_i\huh v_{i+1}$. We have to show that $\cdots \hut v_{i}\suh w_{1} \tuh \cdots \tuh w_{k}\tuh s $ is $\sigma$-open given $C\cup S$  at $v_i$. To check it, first assume $\cdots \hut v_{i}\huh w_{1} \tuh \cdots \tuh w_{k}\tuh s $.
In this case, if $v_i\notin C$, then this is obviously $\sigma$-open. If $v_i\in C$, then it is also $\sigma$-open, since $v_i$ must point to the same strongly connected component. Otherwise, $\pi$ cannot be $\sigma$-open given $C$ at $v_i$ in $G_{|S}$. Second, we assume $\cdots \hut v_{i}\tuh w_{1} \tuh \cdots \tuh w_{k}\tuh s $. In this case $v_i\notin C$, since if $v_i\in C$ then $\anc(C)\cap S\neq\emptyset$. Thus, $\tilde{\pi}$ is $(C\cup S)$-$\sigma$-open at $v_i$ in $G$.

A similar argument can be made for $v_{i+1}$. Then we have constructed a $\sigma$-open walk between $A$ and $B$ given $C\cup S$ in $G$. This contradicts the fact that $A \underset{G}{\stackrel{\sigma}{\perp}} B\mid C\cup S$. Therefore, there is no $\sigma$-open walk between $A$ and $B$ given $C$ in $G_{|S}$. So we have $A \underset{G_{|S}}{\stackrel{\sigma}{\perp}} B\mid C$.

Almost the same argument also applies to the case of $d$-separation.
\end{proof}

\propgraphcond*

\begin{proof}
We show the first two results. Denoting by $G_{\abe(S)}$ the graph obtained by the first step of \cref{def:cond_dmg} allows us to write $G_{|S}=(G_{\abe(S)})_{\sm S}$. By \cref{lem:abemarg} and the fact that marginalizations of two disjoint sets commute \citep[Proposition~5.8]{bongers2021foundations}, we have 
 \[
    (G_{\sm L})_{|S}=((G_{\sm L})_{\abe(S)})_{\sm S}=((G_{\abe(S)})_{\sm L})_{\sm S}=((G_{\abe(S)})_{\sm S})_{\sm L}=(G_{|S})_{\sm L},
 \]
 and by \cref{lem:abeabe}
 \[
 \begin{aligned}
     \left(G_{| S_1}\right)_{| S_2}&=\left(\left((G_{\abe(S_1)})_{\sm S_1}\right)_{\abe(S_2)}\right)_{\sm S_2}\\
     &=\left(\left((G_{\abe(S_1)})_{\abe(S_2)}\right)_{\sm S_2}\right)_{\sm S_1}\\
     &=\left(\left((G_{\abe(S_2)})_{\sm S_2}\right)_{\abe(S_1)}\right)_{\sm S_1}
     =(G_{|S_2})_{|S_1}.
 \end{aligned}
 \]
 Since $(G_{\abe(S_1)})_{\abe(S_2)}\subseteq G_{\abe(S_1\cup S_2)}$, we have $\left(G_{| S_1}\right)_{| S_2}\subseteq G_{|S_1\cup S_2}$.

We now show the third result. It suffices to consider $T=\{t\}$ and show that $(G_{\Do(T)})_{\abe(S)}=(G_{\abe(S)})_{\Do(T)}$. Note that there are two cases. One is $t\notin \sib_G(\anc_G(S))$ and the other one is $t\in\sib_G(\anc_G(S))\sm\anc_G(S)$. Local configuration of a node means the edges adjacent to that node in the graph. In the first case, the local configuration of $t$ is independent of performing $\abe(S)$. In the second case, first intervening on $t$ breaks the bidirected edge between $t$ and $\anc_{G}(S)$ and second performing $\abe(S)$ keeps the local configuration of $t$ untouched, while first performing $\abe(S)$ adds some bidirected edges between $t$ and nodes in $\anc_G(S)\cup \sib_G(\anc_G(S))$, but second intervening on $t$ then throws these bidirected edges away. Hence, no matter what order of the conditioning and the intervening are applied, one still ends up with the same graph.
\end{proof}

\condscmdmg*

\begin{proof}
    We call a subset $A$ of $V$ ancestral if $\anc_{G(M)}(A)=A$ and call an SCM $M$ ancestrally uniquely solvable if for every ancestral subset $A$ of $V$ the SCM $M$ is essentially uniquely solvable \wrt $A$. Since simple SCMs are ancestrally uniquely solvable, we have that $G(M_{\sm S})$ is a subgraph of $G(M)_{\sm S}$ by \citet[Proposition 5.11]{bongers2021foundations}. Let $g:\Xc_W\to \Xc_V$ be the (essentially) unique solution function of $M$. Note that there exists a measurable map $\tilde{g}:\Xc_{W\cap \anc_{G^a}(S)}\to \Xc_S$ such that $\tilde{g}=g_S$ $\Prb(X_W)$-a.s. So, exogenous variables that are not in $W\cap \anc_{G^a}(S)$ will not merge. Denote by $M_{\mathrm{merge}}$ the SCM obtained from the merging operation. This implies that $G(M_{\mathrm{merge}})$ is a subgraph of $G(M)_{\abe(S)}$. Since $\Prb_M(X_W\mid X_S\in \Si)\ll \Prb(X_W)$, $G((M_{\mathrm{merge}})_{\mathrm{update}})$ is a subgraph of $G(M_{\mathrm{merge}})$, where $M_{\mathrm{update}}$ denotes the SCM obtained from $M$ via updating the exogenous distribution to the posterior given $X_S\in \Si$.   Definition of the conditioned DMG and \citet[Proposition 5.11]{bongers2021foundations} yield
    \[
        \begin{aligned}
            G(M_{|X_S\in \Si})&=G(((M_{\mathrm{merge}})_{\mathrm{update}})_{\sm S})\\
            &\subseteq G((M_{\mathrm{merge}})_{\mathrm{update}})_{\sm S}
            \subseteq G(M_{\mathrm{merge}})_{\sm S} \subseteq (G(M)_{\abe(S)})_{\sm S}=G(M)_{|S}, 
        \end{aligned}
    \]
    since $(M_{\mathrm{merge}})_{\mathrm{update}}$ is simple.
    
    We now show the second claim. From the last part, we have that $G(M_{|X_{s_1}\in \Si_1})$ is a subgraph of $G(M)_{|s_1}$. Note that if $G_1$ is a subgraph of $G_2$ then $(G_1)_{|S}$ is a subgraph of $(G_2)_{|S}$. By \cref{prop:iterative_cond}, it holds that $G(M_{|X_{\{s_1,s_2\}}\in \Si_1\times \Si_2})$ is a subgraph of $G((M_{|X_{s_1}\in \Si_1})_{|X_{s_2}\in \Si_2})$ and is therefore a subgraph of $(G(M)_{|s_1})_{|s_2}$. Hence, we can conclude that $G(M_{|X_S\in\Si})$ is a subgraph of $((G(M)_{|s_1})_{\ldots})_{|s_n}$.
\end{proof}

\cievent*

\begin{proof}
Let $\{\Hc_n\}_{n=1}^\infty$ be a countable generator of $\Bc_{\Xc_S}$ such that $\Prb(X_S\in \Hc_n)>0$ for all $n$, which exists since $(\Xc_S,\Bc_{\Xc_S})$ is standard. Define 
\[
\begin{aligned}
    \Gc_n&\coloneqq \sigma(X_C)\vee\sigma(\Ibbm (X_S\in \Hc_1),\ldots,\Ibbm (X_S\in \Hc_n))\\
    \text{ and }
    \Gc_\infty&\coloneqq \sigma(X_C)\vee\sigma(\Ibbm (X_S\in \Hc_1),\Ibbm (X_S\in \Hc_2),\ldots).
\end{aligned}
\]
We suppose that $X_A\ind X_B \mid X_C,X_S\in \Hc$ for all $\Hc\in \Bc_{\Xc_S}$ such that $\Prb(X_S\in \Hc)>0$.  Then we have for all measurable subsets $\Cc\subseteq \Xc_A$ and $\Dc \subseteq \Xc_B$ and for all $n\in \Nb$
\[
    \Prb\left((X_A, X_B)\in \Cc\times \Dc\mid \Gc_n\right)
    =\Prb(X_A\in \Cc\mid \Gc_n)\cdot \Prb(X_B\in \Dc\mid \Gc_n) \quad \Prb\text{-a.s.}
\]
L\'{e}vy's upward theorem \citep{Williams91prob} implies that $\Prb(\ldots\mid \Gc_n)\overset{\text{a.s.}}{\to} \Prb(\ldots\mid \Gc_{\infty})=\Prb(\ldots\mid X_C,X_S)$ as $n\to \infty$. Therefore,
\[
    \Prb((X_A, X_B)\in \Cc\times \Dc\mid X_C,X_S)=\Prb(X_A\in \Cc\mid X_C,X_S)\cdot \Prb(X_B\in \Dc\mid X_C,X_S) \quad \Prb\text{-a.s.},
\]
which means that $X_A\ind X_B\mid X_C,X_S$.
\end{proof}

\begin{lemma}[The first step of \cref{def:cond_dmg} commutes with marginalization]\label{lem:abemarg}
    Let $G=(V, E^d, E^b)$ be a DMG and $S,L\subseteq V$ be two disjoint subsets. Then we have $(G_{\abe(S)})_{\sm L}=(G_{\sm L})_{\abe(S)}$.
\end{lemma}

\begin{proof}
    The two graphs $(G_{\abe(S)})_{\sm L}$ and $(G_{\sm L})_{\abe(S)}$ have the same set of nodes. We show that they have the same set of edges. It is easy to see that $v\tuh u$ is in $(G_{\abe(S)})_{\sm L}$ iff $v\tuh u$ is in $(G_{\sm L})_{\abe(S)}$. We show that for any $v,u\in V\sm L$, we have $v\huh u$ in $(G_{\abe(S)})_{\sm L}$ iff $v\huh u$ is in $(G_{\sm L})_{\abe(S)}$. 
    
    Suppose that $v\huh u$ is in $(G_{\abe(S)})_{\sm L}$. There are two possibilities: (i) $v\huh u$ is in $G_{\abe(S)}$ and (ii) $v\huh u$ is not in $G_{\abe(S)}$. We now consider the first case. We will show that if $v\huh u$ is in $G$, then it must be in $(G_{\sm L})_{\abe(S)}$. By the properties of graphical marginalization, we have $\anc_{G_{\sm L}}(S)=\anc_{G}(S)\sm L$ and $\sib_G(\anc_G(S))\sm L \subseteq \sib_{G_{}\sm L}(\anc_{G_{\sm L}}(S))$. This implies $\SA_G(S)\sm L\subseteq \SA_{G_{\sm L}}(S)$, where $\SA_G(S)\coloneqq \anc_{G}(S)\cup \sib_G(\anc_{G}(S))$. So, if $u,v\in \SA_G(S)$, then $v\huh u$ is in $(G_{\sm L})_{\abe(S)}$. We now consider the second case. The bidirected edge $v\huh u$ must come from some bifurcation in $G_{\abe(S)}$. If the bifurcation consists of only edges in $G$, then $v\huh u$ must also occur in $(G_{\sm L})_{\abe(S)}$. If the bifurcation contains edges not in $G$, then it must be of the form $v\hut l_1\hut \cdots \hut l_{k-1}\huh l_k\tuh \cdots \tuh l_{n-1}\tuh u$ with $l_{k-1},l_k\in \SA_G(S)$, where $l_i\in L$ for all $i$. This forces $v,u\in \SA_{G_{\sm L}}(S)$.  Therefore, $v\huh u$ is in $(G_{\sm L})_{\abe(S)}$.

    Suppose that $v\huh u$ is in $(G_{\sm L})_{\abe(S)}$. If there is 
    a bifurcation in $G$ between $v$ and $u$ with all intermediate nodes in $L$, then the same bifurcation exists in $G_{\abe(S)}$. Therefore, marginalizing out $L$ in $G_{\abe(S)}$ creates $v\huh u$ in $(G_{\abe(S)})_{\sm L}$. If $v\huh u$ comes from the ``abe'' operation on $G_{\sm L}$, then $v,u \in \SA_{G_{\sm L}}(S)$.  If $v,u \in \anc_{G_{\sm L}}(S)$, then $v,u\in \anc_{G}(S)$. Assume $v\in \sib_{G_{\sm L}}(\anc_{G_{\sm L}}(S))\sm \anc_{G_{\sm L}}(S)$. Then $v\huh \tilde{v}$ must be in $G_{\sm L}$ for some $\tilde{v}\in \anc_{G_{\sm L}}(S)$. We also have $u \huh \tilde{u}$ in $G_{\sm L}$ for some $\tilde{u}\in \anc_{G_{\sm L}}(S)$ if $u\notin \anc_{G_{\sm L}}(S)$. If these bidirected edges are present in $G$, then we are done. So we assume that $v\huh \tilde{v}$ is not present in $G$. Then there must exist bifurcations of the form $v\hut l_1\hut \cdots \hut l_{k-1}\suh l_k\tuh \cdots \tuh l_{n-1}\tuh \tilde{v}$ with $l_{i}\in L$ and  $\tilde{u}\hut \tilde{l}_1\hut \cdots \hut \tilde{l}_{j-1}\hus \tilde{l}_j\tuh \cdots \tuh \tilde{l}_{m-1}\tuh u$ with $\tilde{l}_{i}\in L$ ($m=1$ if $u\huh \tilde{u}$ is present in $G$ or $u=\tilde{u}$) in $G$. Then $l_{k-1},\tilde{l}_j \in \SA_{G}(S)$. Therefore, we have $v\hut l_1 \hut \cdots \hut l_{k-1} \huh \tilde{l}_j \tuh \cdots \tuh \tilde{l}_{m-1}\tuh u$ or $v\hut l_1 \hut \cdots \hut l_{k-1} \huh u$ in $G_{\abe(S)}$. Finally, $v \huh u $ is in $(G_{\abe(S)})_{\sm L}$. This finishes the proof.
\end{proof}

\begin{lemma}[The first step of \cref{def:cond_dmg} commutes with itself]\label{lem:abeabe}
    Let $G=(V, E^d, E^b)$ be a DMG and $S_1,S_2\subseteq V$ be two disjoint subsets. Then we have 
    \[
    (G_{\abe(S_1)})_{\abe (S_2)}=(G_{\abe(S_2)})_{\abe(S_1)}\subseteq G_{\abe(S_1\cup S_2)}.
    \]
\end{lemma}

\begin{proof}
To simplify the notation, we write $G_{ij}\coloneqq (G_{\abe(S_i)})_{\abe(S_j)}$. Write $D\coloneqq (\SA_G(S_1)\cap \SA_G(S_2))\cap (\anc_{G}(S_1)\cup \anc_{G}(S_2))$. There are two cases: (i) $D=\emptyset$, and (ii) $D\ne \emptyset$. In the first case, $\SA_G(S_1)\cap \SA_G(S_2)=\emptyset$ or $\SA_G(S_1)\cap \SA_G(S_2)$ does not contain ancestors of $S_1\cup S_2$ in $G$. Then it is not hard to see that adding bidirected edges to nodes in $\SA_G(S_1)$ and adding bidirect edges to nodes in $\SA_G(S_2)$ are independent. Therefore, we have $a\huh b$ for every $a,b\in \SA_G(S_1)$ and $c\huh d$ for every $c,d \in\SA_G(S_2)$ in both $G_{12}$ and $G_{21}$, and there are no other bidirected edges added to $G$. Hence, $G_{12}=G_{21}$ in this case. We now consider the second case. Pick an arbitrary node $a\in D$. By the definition of the set $D$, it holds $a\in \anc_G(S_1)\cup \anc_G(S_2)$. If $a\in \anc_G(S_1)\cap \anc_G(S_2)$, then we have $b\huh c$ for all $b,c\in \SA_G(S_1)\cup \SA_G(S_2)$ in both $G_{12}$ and $G_{21}$, and there are no other bidirected edges added to $G$. In the following, WLOG, we assume $a\in \anc_G(S_1)$ but $a\notin \anc_{G}(S_2)$ by the symmetry of $\anc_G(S_1)$ and $\anc_G(S_2)$. Note that we have $a\huh \tilde{a}$ for some $\tilde{a}\in \anc_{G}(S_2)$. Then we have $\tilde{a}\huh d$ for all $d\in \SA_G(S_1)$ in $G_1$. This implies $\SA_G(S_1)\subseteq \SA_{G_1}(S_2)$ and hence $\SA_{G_1}(S_2)=\SA_{G}(S_1)\cup \SA_G(S_2)$. Overall, in $G_{12}$ we have bidirected edges between any two nodes in $\SA_{G}(S_1)\cup \SA_{G}(S_2)$ added to $G$ and no other bidirected edges are added. By the symmetry of $S_1$ and $S_2$, we have the same argument and conclusion hold for $G_{21}$. This concludes that $G_{12}=G_{21}$.

We now show the second claim. For that, observe that $G_{\abe(S_1)}\subseteq G_{\abe(S_1\cup S_2)}$ and $G_{\abe(S_1\cup S_2)}=(G_{\abe(S_1\cup S_2)})_{\abe(S_2)}$.
\end{proof}

\begin{lemma}\label{lem:consistency}
Let $M=\SCM$ be a simple SCM. Let $g: \mathcal{X}_W \rightarrow \mathcal{X}_V$ be the (essentially unique) solution function of $M$. Let $A\subseteq V $ and write $B\coloneqq V\sm A$. Let $g^{B}: \mathcal{X}_A \times \mathcal{X}_W \rightarrow \mathcal{X}_B$ be the solution function of $M$ \wrt $B$. Then
$$
g_B\left(x_W\right)=g^{B}\left(g_A\left(x_W\right), x_W\right)
$$
for $\Prb(X_W)$-a.a.\ $x_W \in \mathcal{X}_W$.
\end{lemma}

\begin{proof}
For $\Prb(X_W)$-a.a.\ $x_W \in \mathcal{X}_W$ and all $x_V\in \Xc_V$, we have
$$
\left\{\begin{array} { l }
{ x _ { A } = f _ { A } ( x_V,x_W )  } \\
{ x _ { B } = f _ { B } ( x_V,x_W) }
\end{array} \Longleftrightarrow \left\{\begin{array}{ll}
x_A=g_A\left(x_W\right) \\
x_B=g_B\left(x_W\right) .
\end{array}\right.\right.
$$
Since $g^B$ is the essentially unique solution function of $M$ \wrt $B$, for $\Prb(X_W)$-a.a.\ $x_W \in \mathcal{X}_W$ and all $x_V\in \Xc_V$ we have
$$
x_B=f_B(x_V,x_W) \Longleftrightarrow x_B=g^{B}\left(x_A, x_W\right).
$$
Hence, for $\Prb(X_W)$-a.a.\ $x_W \in \mathcal{X}_W$ and all $x_V\in \Xc_V$ we have
\[
\left\{\begin{array} { l }
{ x _ { A } = f _ { A } ( x_V,x_W  ) } \\
{ x _ { B } = f _ { B } ( x_V,x_W  ) }
\end{array} \Longleftrightarrow \left\{\begin{array} { l l }
{ x _ { A } } { = g _ { A } (  x _ { W } )  } \\
{ x _ { B } = g ^ { B  } ( x _ { A } ,  x _ { W } ) }
\end{array} \Longleftrightarrow \left\{\begin{array}{ll}
x_A =g_A\left(x_W\right) \\
x_B =g^{B}\left(g_A\left(x_W\right), x_W\right) .
\end{array}\right.\right.\right.
\]
Since $M$ is uniquely solvable, we can conclude for $\Prb(X_W)$-a.a.\ $x_W \in \mathcal{X}_W$
$$
g_B\left(x_W\right)=g^{B}\left(g_A\left(x_W\right), x_W\right).
$$
\end{proof}

\section{Discussions on conditioning operation for SCMs}\label{app:discuss_cdSCM}

\subsection{SCMs are not flexible enough for representing s-SCMs}\label{app:impossible}

In this subsection, we show that, in general, it is impossible to find a simple SCM encoding all the causal semantics of a simple s-SCM. This gives an answer to Question \ref{Q1}: the causal semantics of the ancestors of selection nodes cannot be preserved in general under the framework of simple SCMs and \cref{thm:causal_semantics} shows that the causal semantics of the non-ancestors of the selection nodes can always be preserved when applying the conditioning operation.

We can see the rung-one and rung-two semantics of an SCM $M=\SCM$ as a collection of distributions $\Cc^V=\left\{\Prb^{[x_T]}(X_{V\sm T}):T\subseteq V, x_T\in \Xc_T\right\}$, which satisfies some constraints given by the SCM $M$, i.e.,  $\Prb^{[x_T]}(X_{V\sm T})=\Prb_M(X_{V\sm T}\mid \Do(X_T=x_T))$ for all $T\subseteq V$ and $x_T\in \Xc_T$.

Let
\[
\Cc^{V\sm S}\coloneqq\left(\Prb^{[x_T]}(X_{V\sm (S\cup T)}):T\subseteq V\sm S, x_T\in \Xc_T, \Prb_M(X_S\in \Si\mid \Do(X_T=x_T))>0\right)
\]
where
\[
\Prb^{[x_T]}(X_{V\sm (S\cup T)})\coloneqq\frac{\Prb_M(X_{V\sm (T\cup S)},X_{S}\in \Si\mid \Do(X_T=x_T))}{\Prb_M(X_{S}\in \Si\mid \Do(X_T=x_T))}
\]
for some simple s-SCM $(M,X_S\in \Si)$ with $M=\SCM$. Note that $\Cc^{V \sm S}$ is just the collection of observational and interventional distributions induced by the simple s-SCM $(M,X_S\in \Si)$. Now Question \ref{Q1} can be rephrased as: given $\Cc^{V\sm S}$ defined above, can we \emph{always} find a simple SCM $\tilde{M}$ with endogenous nodes $O\coloneqq V\sm S$ such that
\begin{equation}\label{eq:Q1}
   \forall T\subseteq O, x_T\in \Xc_T \quad \Prb^{[x_T]}(X_{O\sm T})= \Prb_{\tilde{M}}(X_{O}\mid \Do(X_T=x_T)) ?
\end{equation}
The answer is \emph{No}:\footnote{This is also related to \citet{Lauritzen98generating}, who shows that any hierarchical model can be generated from graphical models represented by a DAG with selection.}
\begin{proposition}\label{prop:impossible}
    There exists a simple s-SCM $(M,X_S\in \Si)$ such that it is impossible to find a simple SCM $\tilde{M}$ with endogenous nodes $O\coloneqq V\sm S$ such that \cref{eq:Q1} holds.
\end{proposition}

\begin{proof}
    We first show that in general there is no acyclic SCM $\tilde{M}$ such that equation~\eqref{eq:Q1} holds and second show that in general there is no simple (even cyclic) SCM satisfying equation~\eqref{eq:Q1}. This can be summarized by Figure~\ref{fig:venn_diagram_scms}.

\begin{figure}[t]\centering
  \begin{tikzpicture}[scale=0.8, transform shape]
      \begin{scope}
%        \draw[draw=green,fill=green,opacity=0.2] (0,0) ellipse (1.5cm and 0.5cm);
        \draw[draw,fill,opacity=0.2] (0,0.3) ellipse (2cm and 1cm);
        \draw[draw,fill,opacity=0.2] (0,0.6) ellipse (3cm and 1.5cm);
        \draw[draw,fill,opacity=0.2] (0,0.9) ellipse (4cm and 2cm);
        \node at (0,0.3) {acyclic SCMs};
        \node at (0,1.6) {simple SCMs};
        \node at (0,2.4) {simple s-SCMs};
      \end{scope}
  \end{tikzpicture}
  \caption{Venn diagram for different causal modeling classes with rung-two information.\label{fig:venn_diagram_scms}}
\end{figure}

\paragraph{No acyclic SCM $\tilde{M}$ satisfying equation~\eqref{eq:Q1}.} Let $M$ be an acyclic SCM with $X_S=X_A+X_B$, $X_A=E_A\sim \Uni[0,1]$, $X_B=E_B\sim \Uni[0,1]$, and selection mechanism $X_S\geq 0.8$, whose graph is shown in Figure~\ref{fig:disdef}.

\begin{figure}[ht]
\centering
\begin{tikzpicture}[scale=0.8, transform shape]
  \begin{scope}[xshift=0]
    \node[ndout] (X) at (1,3) {$A$};
    \node[ndout] (Y) at (3,3) {$B$};
    \node[ndsel] (S) at (2,2) {$S$};
    \draw[arout] (X) to (S);
    \draw[arout] (Y) to (S);
    \node at (2,1) {$G(M)$};
  \end{scope}

\end{tikzpicture}
\caption{The causal graphs of the SCM $M$ with selection $X_S\geq 0.8$.}
\label{fig:disdef}
\end{figure}

We shall explain the non-existence of an acyclic SCM $\tilde{M}$ satisfying equation~\eqref{eq:Q1} by contradiction. Assume that there exists an acyclic SCM $\tilde{M}$ satisfying equation~\eqref{eq:Q1}. For $x_A\in [0,1]$ and $x_B\in [0,1]$, we require $\Prb_{\tilde{M}}(X_B\mid \Do(X_A=x_A))=\Prb_M(X_B\mid \Do(X_A=x_A),X_S\geq 0.8)$ and $\Prb_{\tilde{M}}(X_A\mid \Do(X_B=x_B))=\Prb_M(X_A\mid \Do(X_B=x_B),X_S\geq 0.8)$. Since $\Prb_M(X_B\mid \Do(X_A=x_A),X_S\geq 0.8)$ and $\Prb_M(X_A\mid \Do(X_B=x_B),X_S\geq 0.8)$ are not constant in $x_A$ and $x_B$ respectively, we must have a directed edge from $A$ to $B$ and a directed edge from $B$ to $A$ in the causal graph $G(\tilde{M})$. Hence, $\tilde{M}$ cannot be an acyclic SCM.

\paragraph{No simple (even cyclic) SCM $\tilde{M}$ satisfying equation~\eqref{eq:Q1}.} Let $M$ be an SCM that satisfies
\[
\begin{aligned}
    \Prb_M(X_B)&=0.9\delta_0+0.1\delta_1,\\
    \Prb_M(X_A\mid \Do(X_B=0))&=0.6\delta_0+0.4\delta_1,\
    \Prb_M(X_A\mid \Do(X_B=1))=0.1\delta_0+0.9\delta_1,\\
    \Prb_M(X_S\mid \Do(X_A=0))&=0.9\delta_0+0.1\delta_1,\
    \Prb_M(X_S\mid \Do(X_A=1))=0.1\delta_0+0.9\delta_1,
\end{aligned}
\]
selection mechanism $X_S=1$, and graph $G(M)$ shown in Figure \ref{fig:impossible2}.

\begin{figure}[ht]
\centering
\begin{tikzpicture}[scale=0.8, transform shape]
  \begin{scope}[xshift=0]
    \node[ndout] (X) at (1.5,3) {$A$};
    \node[ndout] (Y) at (0,3) {$B$};
    \node[ndsel] (S) at (3,3) {$S$};
    \draw[arout] (Y) to (X);
    \draw[arout] (X) to (S);
    \node at (1.5,2) {$G(M)$};
  \end{scope}
\end{tikzpicture}
\caption{The causal graph of the SCM $M$ with selection $X_S=1$.}
\label{fig:impossible2}
\end{figure}

\citet{manski98bounding} proved a ``natural'' bound and we point out here that it also holds for simple SCMs. More specifically, for an arbitrary simple SCM $\tilde{M}$ with discrete endogenous variables $\{X_A,X_B\}$, it must hold that $\Prb_{\tilde{M}}(X_A=x_A,X_B)\leq \Prb_{\tilde{M}}(X_B\mid \Do(X_A=x_A))$ for all $x_A\in \Xc_A$.\footnote{This inequality is interpreted as $\Prb_{\tilde{M}}(X_A=x_A,X_B=x_B)\leq \Prb_{\tilde{M}}(X_B=x_B\mid \Do(X_A=x_A))$ for all $x_B\in \Xc_B$.} Indeed, by the consistency $X_B=X_B(X_A)$ a.s.\ (see also \citet[Proposition~7.5.1]{forre25causality}) and elementary probability theory, we have
\[
\begin{aligned}
    \Prb_{\tilde{M}}(X_A=x_A,X_B)&=\Prb_{\tilde{M}}(X_A=x_A,X_B(x_A))\\
    &\leq \Prb_{\tilde{M}}(X_B(x_A))=\Prb_{\tilde{M}}(X_B\mid \Do(X_A=x_A)).
\end{aligned}
\]

Assume that $\tilde{M}$ is a simple SCM satisfying equation~\eqref{eq:Q1}. Then by requiring 
\[
\Prb_{\tilde{M}}(X_A=1,X_B=1)=\Prb_{M}(X_A=1,X_B=1\mid X_S=1) 
\]
and 
\[
\Prb_{\tilde{M}}(X_B=1\mid \Do(X_A=1))=\Prb_{M}(X_B=1\mid \Do(X_A=1),X_S=1),
\]
we have
\[
\Prb_{\tilde{M}}(X_A=1,X_B=1)\approx 0.176> 0.1=\Prb_{\tilde{M}}(X_B=1\mid \Do(X_A=1)),
\]
which contradicts the natural bound $\Prb_{\tilde{M}}(X_A=1,X_B=1)\leq \Prb_{\tilde{M}}(X_B=1\mid \Do(X_A=1))$ that the simple SCM $\tilde{M}$ must satisfy.
\end{proof}

\begin{remark}[Interventions on ancestors of selection nodes]
  Let $T\subseteq \anc_{G(M)}(S)$.
  The above theorem tells us that in general $\Prb_{M_{|X_S\in \Si}}(X_{O\sm T}\mid \Do(X_T=x_T))\ne \Prb_M(X_{O\sm T}\mid \Do(X_T=x_T),X_S\in \Si)$, since $\Prb_{M_{|X_S\in \Si}}(X_{O\sm T}\mid \Do(X_T=x_T))=\Prb_M(X_{O\sm T}(x_T)\mid X_S\in \Si)$ and $\Prb_M(X_{O\sm T}(x_T)\mid X_S\in \Si)\ne \Prb_M(X_{O\sm T}\mid \Do(X_T=x_T),X_S\in \Si)$ in general. However, in some cases, we can infer that $\Prb_{M_{|X_S\in \Si}}(X_{O\sm T}\mid \Do(X_T=x_T))= \Prb_M(X_{O\sm T}\mid \Do(X_T=x_T),X_S\in \Si)$. For example, if $\Prb_M(X_S\in \Si\mid \Do(X_T=x_T))=\Prb_M(X_S\in \Si)$ then we can conclude that $\Prb_{M_{|X_S\in \Si}}(X_{O\sm T}\mid \Do(X_T=x_T))= \Prb_M(X_{O\sm T}\mid \Do(X_T=x_T),X_S\in \Si)$. As another example, if we know that the second rule or the third rule of do-calculus applies to $\Prb_{M_{|X_S\in\Si}}(X_{O\sm T}\mid \Do(X_T=x_T),X_C)$ and $\Prb_M(X_{O\sm T},X_S\in \Si\mid \Do(X_T=x_T),X_C)$ to reduce $\Do(X_T)$ to given $X_T$ or eliminate $\Do(X_T)$ entirely, then we have the equality $\Prb_{M_{|X_S\in \Si}}(X_{O\sm T}\mid \Do(X_T=x_T))= \Prb_M(X_{O\sm T}\mid \Do(X_T=x_T),X_S\in \Si)$ under the assumption of discreteness and positivity. Consider a concrete toy example where we have an SCM with causal graph $S\hut T\tuh O$. Under discreteness and positivity assumption, we have
  \[
    \begin{aligned}
      \Prb_{M_{|X_S\in \Si}}(X_O\mid \Do(X_T=x_T))&=\Prb_{M_{|X_S\in \Si}}(X_O\mid X_T=x_T)\\
      &=\Prb_{M}(X_O\mid X_T=x_T,X_S\in \Si)\\
      &=\frac{\Prb_{M}(X_O,X_S\in \Si\mid X_T=x_T)}{\Prb_{M}(X_S\in\Si\mid X_T=x_T)}\\
      &=\Prb_M(X_O \mid \Do(X_T=x_T),X_S\in \Si).
    \end{aligned}
  \]
  It might be possible to find some interesting conditions to guarantee this, and therefore in the given setting one does not need to treat the ancestors of selection nodes differently than the non-ancestors. 

\end{remark}

\subsection{Other variants of conditioning operation for SCMs}\label{app:variants_cond}

\cref{def:cdSCM} is not the only possible way to define a ``conditioned SCM''. In this section, we explore some other possibilities of conditioning operations such as different decompositions of exogenous nodes, and conditioning for Causal Bayesian Networks.\footnote{Conditioning on a variable that we want to observe is discussed in \cref{sec:3.2}.}

% \[
% \begin{aligned}
%     H_0\coloneqq \{ w\in W: [&g^{-1}_S(\Si)]_{x_w^\prime}=[g^{-1}_S(\Si)]_{x_w^{\prime\prime}}\text{ up to a } \Prb(X_{W\sm w})\text{-null set} \\
%     &\text{for all } x_w^\prime,x_w^{\prime\prime}\in C \subseteq \Xc_{w} \text{ with } \Prb(X_w\in C)=1\},
% \end{aligned}
% \]
% and $H\coloneqq W\sm H_0$, where $[\tilde{g}^{-1}_S(\Si)]_{x_w}\coloneqq \{\tilde{x}_{W\sm w}: (x_w,\tilde{x}_{W\sm w})\in \tilde{g}^{-1}_S(\Si)\}$ is the $x_w$-section of set $\tilde{g}^{-1}_S(\Si)$.

\subsubsection{Different decomposition of exogenous nodes}

Why do we care about decomposition of exogenous variables and make sure that the new coarse-grained variables are mutually independent given selection? That is because we want to have a Markov property of the causal graphs of our SCMs (so that we can apply do-calculus, and so on) and without mutual independence of the exogenous variables this may fail. 

For example, in some literature (such as \citet{bareinboim22pearl_hierarchy}), SCMs are not required to have mutually independent exogenous random variables but can have any exogenous probability distribution $\Prb(X_W)$ on $X_W$. A bidirected edge is drawn between two endogenous variables $X_{v_1}$ and $X_{v_2}$ if they share the same exogenous variables or their exogenous parents are correlated according to $\Prb(X_W)$. It seems that if we adopt this framework, then we just need to update the exogenous distribution and not to merge exogenous variables when defining a conditioning operation. Consider the ``SCM''
\[
    M: \left\{
    \begin{aligned}
        \Prb(E_1,E_2,E_3)&=\frac{1}{4}\delta_{000}+\frac{1}{4}\delta_{011}+\frac{1}{4}\delta_{101}+\frac{1}{4}\delta_{110}\\
        X_1 & = E_1  \\
        X_2 & = E_2  \\
        X_3 & = E_3.
    \end{aligned}
    \right.
\]
\begin{figure}[ht]
\centering
\begin{tikzpicture}[scale=0.8, transform shape]
  \begin{scope}[xshift=0]
    \node[ndexo] (E1) at (0,4.5) {$E_1$};
    \node[ndexo] (E2) at (1.5,4.5) {$E_2$};
    \node[ndexo] (E3) at (3,4.5) {$E_3$};
    \node[ndout] (X1) at (0,3) {$X_1$};
    \node[ndout] (X2) at (1.5,3) {$X_2$};
    \node[ndout] (X3) at (3,3) {$X_3$};
    \draw[arout] (E1) to (X1);
    \draw[arout] (E2) to (X2);
    \draw[arout] (E3) to (X3);
    \node at (1.5,2) {$G^a(M)$};
  \end{scope}
  \begin{scope}[xshift=6cm]
    \node[ndout] (X1) at (0,3) {$X_1$};
    \node[ndout] (X2) at (1.5,3) {$X_2$};
    \node[ndout] (X3) at (3,3) {$X_3$};
    \node at (1.5,2) {$G(M)$};
  \end{scope}
\end{tikzpicture}
\caption{The ``causal'' graph of the ``SCM'' $M$.}
\label{fig:wrongscm}
\end{figure}
According to the above definition, this ``SCM'' $M$ would have graph $G(M)$ shown in Figure~\ref{fig:wrongscm}. The graph is of this form because $E_i\ind E_j$ for $i,j=1,2,3$ and $i\ne j$. However, although the graph implies that $X_1,X_2\underset{G(M)}{\stackrel{d}{\perp}} X_3$, we know from the model $M$ that \[X_1,X_2\underset{\Prb_M(X_1,X_2,X_3)}{\notind} X_3.\] Therefore, the Markov property does not hold, and all the results based on it may not hold either, such as the back-door criterion and Pearl's do-calculus.

It is worth mentioning that one can have different decompositions of the exogenous nodes of the conditioned SCMs. One extreme is to merge all the exogenous nodes into one single node, which results in the ``coarsest'' model without much information. The other extreme is to consider the ``finest'' decomposition of the exogenous nodes under the current framework of SCMs, which results in ``most fine-grained'' models with as much information as possible.

In Definition~\ref{def:cdSCM}, we used a ``finest'' decomposition. We can consider any decomposition of label set $W$ given $X_S\in \Si$ that is coarser than the one given in Definition~\ref{def:cdSCM}. There are two such examples that appear natural, but one should note that the properties of these two operations are slightly different from the ones of Definition~\ref{def:cdSCM}.

\begin{definition}[Conditioned SCMs II]\label{def:cdSCM2}
    In the setting of Definition~\ref{def:cdSCM}, we define $M_{|X_S\in\Si}=(\hat{V},\bar{W},\hat{\Xc},\hat{f},\hat{\Prb})$ where $\hat{V},\hat{\Xc},\hat{f},\hat{\Prb}$ are the same as Definition~\ref{def:cdSCM}, while
    \[
    \bar{W}\coloneqq \{W\cap \anc_{G^a(M_{\sm (V\sm S)})}(S)\}\dcup (W\sm \anc_{G^a(M_{\sm (V\sm S)})}(S)).
    \]
\end{definition}

\begin{definition}[Conditioned SCMs III]\label{def:cdSCM3}
    In the setting of Definition~\ref{def:cdSCM}, we define $M_{|X_S\in\Si}=(\hat{V},\bar{\bar{W}},\hat{\Xc},\hat{f},\hat{\Prb})$ where $\hat{V},\hat{\Xc},\hat{f},\hat{\Prb}$ are the same as Definition~\ref{def:cdSCM}, while
    \[\bar{\bar{W}}\coloneqq \{W\cap \anc_{G^a(M)}(S)\}\dcup (W\sm \anc_{G^a(M)}(S)).\]
\end{definition}

 For \cref{def:cdSCM3}, we do not have Proposition~\ref{prop:cond_marg}, but we can obtain the commutativity of marginalization and conditioning up to counterfactual equivalence (see an old version of the current article \citep{Chen24Modeling}).

\subsubsection{conditioning operation for causal Bayesian networks}

Given a Causal Bayesian Network $N=(G=(V,E^d),\{\Prb(X_v\mid \Do(X_{\pa_G(v)}))\}_{v\in V})$, we can
use the deterministic representation of Markov kernels to construct an SCM $M_N=\SCM$ (i.e., there exist a uniformly distributed random variable $U_v$ and a measurable function $f_v$ such that $\Prb(X_v\mid \Do(X_{\pa_G(v)}))=f_v(X_{\pa_G(v)},U_v)_*\Prb(U_v)$, see, e.g., \citet[Proposition 10.7.6]{bogachev2007measure}). We can then apply the conditioning operation for $M_N$ and then transform the conditioned SCM back to get a conditioned Causal Bayesian Network with latent variables.

% Acknowledgements and Disclosure of Funding should go at the end, before appendices and references


% Note: in this sample, the section number is hard-coded in. Following
% proper LaTeX conventions, it should properly be coded as a reference:

%In this appendix we prove the following theorem from
%Section~\ref{sec:textree-generalization}:

\subsection{Conditioning operation for SCMs with inputs}\label{sec:cond_iSCM}

In this section, we extend the definition of the conditioning operation (Definition~\ref{def:cdSCM}) on SCMs to SCMs with input nodes, which we call \textbf{iSCMs} and are introduced in \citet{forre25causality}. The difference between iSCMs and SCMs is that iSCMs have exogenous (non-stochastic) input variables in addition to endogenous and exogenous random variables. Note that such an extension of conditioning operation is not straightforward due to interactions between non-stochastic variables and stochastic variables (\cf Remark~\ref{rem:interaction} and Definition~\ref{def:cdiSCM}).

\begin{definition}[SCMs with input nodes (iSCMs)]
    A \textbf{Structural Causal Model with input nodes (iSCM)} is a tuple $M=(J,V,W,\Xc,\Prb,f)$ where $J$ represents the label set for \textbf{exogenous input (non-stochastic) variables}. Other components have the same definitions as their counterparts in Definition~\ref{def:scm}, except for $\Xc=\prod_{i\in J\dcup V\dcup W} \Xc_i$.
\end{definition}

All the definitions in Section~\ref{sec:pre} can be extended to iSCMs with minor modifications (see \citep{forre25causality} for more details. For example, an iSCM could induce a Markov kernel $\Prb_M(X_V\mid \Do(X_J))$ in general and not merely a probability distribution $\Prb_M(X_V)$. A solution function $g:\Xc_J\times \Xc_W\to \Xc_V$ has also arguments from $x_J\in \Xc_J$ and all quantifiers ``for $\Prb(X_W)$-a.a.\ $x_W\in \Xc_W$, and all $x_V\in \Xc_V$'' used in the relevant definition (e.g., essentially unique solution function) are replaced by ``for all $x_J\in \Xc_J$, $\Prb(X_W)$-a.a.\ $x_W\in \Xc_W$, and all $x_V\in \Xc_V$'' (note that the ordering of these quantifiers is important). We define $M_{\Do(T)}$ as an iSCM where we transform $T$ from endogenous nodes to exogenous input nodes for $T\subseteq V$. This is in contrast to the case in SCMs where we always need to assign a specific value to the intervened variables. 

The input variables can model hard/soft interventions, (non-stochastic) parameters for models, context variables, and regime indicators \citep{Dawid02influence,Dawid21decision} or say policy variables \citep{spirtes2001causation}, etc. Conceptually, iSCMs model data-generating processes where first the values of $X_J$ are assigned (e.g., by the experimenter) and then the process in Algorithm~\ref{alg:sampler} is implemented (function $g$ also depends on $X_J$).  Mathematically, input variables can help rigorously develop a general measure-theoretic causal calculus \citep{forre2021transitional,forre25causality}.\footnote{When the variables are not discrete, a naive approach would not work. See the discussions in \citet{forre2020causal} and \citet{forre2021transitional}.} One common feature of all these concepts is that they are modeled by variables without probability distributions on them and, therefore, need a distinct treatment from ordinary random variables. 

One convenient foundational framework for dealing with such non-stochastic and ordinary stochastic variables universally is named \textbf{transitional probability theory} in \citep{forre2021transitional}.  We will use relevant concepts directly and refer the reader to \citet{forre2021transitional} for more details.

\begin{remark}[Interaction between non-stochastic variables and stochastic variables]\label{rem:interaction}
   Given $X_S\in \Si$, the exogenous distribution $\Prb(X_W)$ becomes an exogenous Markov kernel
   \[
   \Prb_M(X_W\mid X_S\in \Si, \Do(X_J))=\frac{\Prb_M(X_W,X_S\in \Si\mid \Do(X_J))}{\Prb_M(X_S\in \Si \mid \Do(X_J))},
   \] 
   where the exogenous distribution has a dependency on $X_J$. Since $\Prb_M(X_S\in \Si \mid \Do(X_J=x_J))=0$ might be possible for some $x_J\in \Xc_J$, this might require merging some exogenous input nodes and restricting $\Xc_J$.
\end{remark}
 
To eliminate the ``entanglement'' between $X_J$ and $X_W$ mentioned in Remark~\ref{rem:interaction}, we need \citet[Theorem Proposition~10.7.6]{bogachev2007measure}. It states that we can represent a Markov kernel as a deterministic map of a random variable.

\begin{defprop}[Deterministic representation of Markov kernels]\label{defprop:rep_mk}
Let $\mathcal{Z}$ be any measurable space, $\Xc$ be a standard measurable space, and $\mathrm{P}(X \mid Z): \mathcal{Z} \rightarrow \Pc(\mathcal{X})$ be a Markov kernel. Then there exists a measurable function $R : \mathcal{Z} \times  \Uc \rightarrow \mathcal{X}$ such that 
  $$
  \Prb(X\mid Z)=\Prb(R(Z,U)\mid Z),
  $$ 
  where $\Uc$ is a measurable space and $U$ a random variable taking values in $\Uc$. We call $(U,R)$ a \textbf{deterministic representation of the Markov kernel $\Prb(X\mid Z)$}. 
\end{defprop}

\begin{remark}
     One can take $\Prb(U)$ to be $\mathrm{Uni}([0,1]^n)$ or $\mathcal{N}(0,I_{n})$. After fixing $U$, in general, the measurable function $R$ is \emph{not} unique, injective, or surjective.
\end{remark}

The conditioning operation for iSCMs can be seen as the composition of (i) merging all the input nodes that are ancestors of the selection nodes and exogenous random nodes that are ancestors of the selection nodes, respectively, restricting the values that input variables can take given $X_S\in \Si$, and then representing the corresponding conditioned Markov kernel deterministically to eliminate the ``dependence'' between input variables and exogenous random variables, and (ii) marginalizing out the selection variables.

  % Here, $\ell:J\cup V\cup W\rightarrow J\cup V\cup W$ is a merging function that merges $W\cap\anc(S)$ into one single node,
  % $J\cap\anc(S)$ into one single node, respectively, and any other nodes are untouched. The mapping $(\cdot)_{\mathrm{ch}}:M\mapsto (M)_{\mathrm{ch}}$ changes the input space $\X_{\ell(J\cap \anc(S))}$ to $\X_{\star_J}$ and
  % $\Prb(X_{W\cap B})$ to $\Prb(X_{W\cap B}\mid X_S\in \mathcal{S},\doi(X_{J\cap B}))$, respectively, while other parts of $M$ are untouched. The measurable mapping $R:\X_{J\cap\anc(S)}\times[0,1]\rightarrow \X_{W\cap\anc(S)}$ is a reparameterization satisfying the property
  % \[
  %   \Prb(X_{W\cap\anc(S)}\mid X_{S}\in\Si,\doi(X_{J\cap \anc(S)}))=\delta(R\mid E, X_{J\cap \anc(S)})\circ\Prb(E),
  % \]
  % where $\Prb(E)=\Uni[0,1]$.

Let $M=\left( J, V, W, \Xc, \Prb, f\right)$ be a simple iSCM \citep[Definition~6.6.5]{forre25causality}. Let $S\subseteq V$ be a subset of endogenous nodes and $\mathcal{S} \subseteq \Xc_S$ a measurable subset of values $X_S$ may take where there exists $x_J\in\Xc_J$ such that $\Prb_M(X_S\in \Si\mid \Do(X_J=x_J))>0$. Let $g:\Xc_J \times \Xc_W\rightarrow \Xc_V$ be the (essentially unique) solution function of $M$. Note that we can see $g_S$ as a map from $\Xc_{J\cap \anc_{G(M)}(S)}\times \Xc_{W\cap \anc_{G^a(M)}(S)}$ to $\Xc_S$.  Write $O\coloneqq V\sm S$, $B_1\coloneqq \anc_{G(M)}(S)\cap J$, and $B_2\coloneqq \anc_{G^a(M)}(S)\cap W$.

% We define a new equivalence relation $\Cc\doteq \Dc$ for $\Cc,\Dc\subseteq \Xc_J\times \Xc_W$ if $\pr_{\Xc_J}(\Cc)=\pr_{\Xc_J}(\Dc)$ and for all $x_J\in \pr_{\Xc_J}(\Cc)$ we have $\Cc_{x_J}\pequ \Dc_{x_J}$ where $\pr_{\Xc_J}:\Xc_J\times \Xc_W\to \Xc_J$ is the projection map on $\Xc_J$. Define $\Pf_{\Si}$ as the set of partitions $\Ic=\{I_1,\ldots, I_k\}$ of the label set $(J\cup W)\cap \anc_{G^a(M)}(S)$ such that $g^{-1}_S(\Si)\doteq \bigtimes_{i=1}^k \pr_{\Xc_{I_i}}(g^{-1}_S(\Si))$. Similarly to the discussion before Definition~\ref{def:cdSCM}, we have that $(\Pf_{\Si},\vee)$ is a finite join semi-lattice with largest elements $\Hc=\{H_i\}_{i=1}^n$. Define $J_i\coloneqq H_i\cap J$ and $W_i\coloneqq H_i\cap W$ for all $1\leq i \leq n$. For $1\leq i \leq n$, let $(U_i,R_i)$ be a deterministic representation of the Markov kernel $\Prb_M(X_{W_i}\mid X_S\in \mathcal{S},\Do(X_{J_i}))$ coming from \cref{defprop:rep_mk} such that $U_i\sim \Uni([0,1])$.

\begin{definition}[Conditioned iSCM]\label{def:cdiSCM}  Under the above setting, we define a conditioned iSCM
$M_{|X_S\in\mathcal{S}}\coloneqq\left(\hat{J}, \hat{V}, \hat{W}, \hat{\Xc}, \hat{\Prb}, \hat{f}\right)$ by 
\begin{enumerate}
    \item $\hat{J}\coloneqq \{\star_J\}\dcup (J\sm B_1)$ where $\star_J=B_1$;

    \item $\hat{W}\coloneqq \{\star_W\}\dcup (W\sm B_2)$ where $\star_W=B_2$;
    
    \item $\hat{V}\coloneqq V\sm S$;

    \item $\Xc_{\hat{J}}\coloneqq \Xc_{\star_J}\times \Xc_{J\sm B_1}$ where 
    \[
    \begin{aligned}
        \Xc_{\star_J}\coloneqq \{x_{B_1}\in\Xc_{B_1}\mid    \Prb_M(X_{S}\in \Si\mid \Do(X_{B_1}=x_{B_1}))>0\}
    \end{aligned}
    \] 
    and $\Xc_{\hat{W}}\coloneqq \Xc_{\star_W}\times \Xc_{W\sm B_2}$ with $\Xc_{\star_W}=[0,1]$ and $\Xc_{\hat{V}}\coloneqq \Xc_{V\sm S}$;

    \item $\hat{\Prb}(X_{\hat{W}})\coloneqq \Prb(X_{W\sm B_2})\otimes \Prb(X_{\star_W})$ with $\Prb(X_{\star_W})=\mathrm{Uni}([0,1])$;

    \item causal mechanism:
\[
  \begin{aligned}
    &\hat{f}(x_{\hat{J}},x_{\hat{V}},x_{\hat{W}})\coloneqq\\
   &  f_O(x_{J\sm B_1},x_{\star_J}, x_{O}, g^S(x_{J\sm B_1},x_{\star_J}, x_{O},
  x_{W\setminus B_2}, R(x_{\star_J},x_{\star_W})), x_{W\setminus B_2}, R(x_{\star_J},x_{\star_W})),
  \end{aligned}
  \]
  where $(X_{\star_W},R)$ is a deterministic representation of the Markov kernel $\Prb_M(X_{B_2}\mid \Do(X_{B_1}),X_S\in \Si)$ and $g^S:\Xc_{J}\times \Xc_O\times \Xc_W\to \Xc_S$ is the (essentially unique) solution function of $M$ \wrt $S$.
\end{enumerate}
\end{definition}

\begin{remark}
    Here we simply merge all the input node ancestors of $S$ and exogenous node ancestors of $S$, respectively. One can also derive a finer merging scheme, similar to what we did for \cref{def:cdSCM}. For the sake of space, we did not spell out all the details. The essential point that we want to show in this subsection is how to eliminate the dependency between input variables and exogenous random variables given $X_S\in \Si$. 
\end{remark}

  One can develop a theory for the conditioning operation for iSCMs and show the corresponding properties in parallel to what we did in \cref{sec:cond}. We now define an operation on iSCMs called \textbf{exogenous (quasi-)pullback of iSCMs}, which defines formally the first step of the conditioning operation for iSCMs. Because of the nice properties of marginalization \citet[Section~6.8]{forre25causality}, to show properties of the conditioning operation for iSCMs it suffices to show that the corresponding properties of exogenous quasi-pullback of iSCMs hold. See \citet[Section~8.3.2]{forre25causality} for some properties of exogenous pullback of iSCMs. Note that merging exogenous nodes is a special case of exogenous (quasi-)pullback of iSCMs.

\begin{definition}[Exogenous quasi-pullback iSCMs]
  Let $M=(J, V, W, \Xc, \Prb, f)$ be an iSCM. Let $\tilde{M}=(\tilde{J}, V, \tilde{W}, \tilde{\Xc}, \tilde{\Prb}, \tilde{f})$ be an iSCM. Let $\Phi_J:\Xc_{\tilde{J}\cup\tilde{W}}\rightarrow \Xc_J$ be a measurable mapping that does not depend on the $X_{\tilde{W}}$-component, and let $\Phi_W: \Xc_{\tilde{J}\cup\tilde{W}}\rightarrow \Xc_W $ be a measurable mapping such that $\Qr(X_W\mid X_{\tilde{J}})\coloneqq (\Phi_W)_*\left(\delta(X_{\tilde{J}}\mid X_{\tilde{J}})\otimes\tilde{\Prb}(X_{\tilde{W}})\right)\ll \Prb(X_W)$, i.e., for all $x_{\tilde{J}}\in \Xc_{\tilde{J}}$, and for every measurable subset $\Ac\subseteq \Xc_W$, $\Prb(X_W\in \Ac)=0$ implies that $\Qr(X_W\in \Ac\mid X_{\tilde{J}}=x_{\tilde{J}})=0$. Assume that
  \[
    \tilde{f}(x_{\tilde{J}},x_V,x_{\tilde{W}})=f\left(\Phi_J(x_{\tilde{J}}),x_V,\Phi_W(x_{\tilde{J}},x_{\tilde{W}})\right).
  \]
  Then we call $\Phi=(\Phi_J,\Phi_W):\Xc_{\tilde{J}\cup\tilde{W}}\rightarrow \Xc_J\times \Xc_W$ an \textbf{exogenous quasi-pullback function of $M$} and $\tilde{M}$ an \textbf{exogenous quasi-pullback iSCM of $M$ associated with $(\Phi,\tilde{\Prb})$}. We  denote $\tilde{M}$ by $M_{\rp(\Phi,\tilde{\Prb})}$.

  % If $\Phi$ furthermore satisfies that
  % $(\Phi_W)_*\left(\delta(X_{\tilde{J}}\mid X_{\tilde{J}})\otimes\tilde{\Prb}(X_{\tilde{W}})\right)=\delta(X_{\tilde{J}}\mid X_{\tilde{J}})\otimes\Prb(X_W)$, then we call $\tilde{M}$ a \textbf{reparameterized iSCM of $M$ (associated with $(\Phi,\tilde{\Prb})$)}.
\end{definition}
  
  It is easy to see that $M_{|X_S\in \Si}=\Big(M_{\mathrm{ep}(\Phi,\hat{\Prb})}\Big)_{\sm S}$ is a marginalized exogenous quasi-pullback iSCM of $M$ associated with the exogenous quasi-pullback function $\Phi=(\Phi_J,\Phi_W)$ and distribution $\hat{\Prb}(X_{\hat{W}})$ where
\[
\begin{aligned}
    &\Phi_J:\Xc_{\hat{J}\cup \hat{W}}\to \Xc_J, \Phi_J(x_{\hat{J}},x_{\hat{W}})\coloneqq x_J\\ 
    \text{and } 
    &\Phi_W:\Xc_{\hat{J}\cup \hat{W}}\to \Xc_W,\Phi_W(x_{\hat{J}},x_{\hat{W}})\coloneqq (x_{W\sm B_2},R(x_{\star_J},x_{\star_W})).
\end{aligned}
\]
  
  The graphical conditioning operation for DMGs with input nodes can be defined similarly to Definition~\ref{def:cond_dmg} by (i) merging all the input node ancestors of $S$, (ii) adding bidirected edges to output nodes that are ancestors or siblings of ancestors of $S$, and (iii) marginalizing out $S$. In the graph, we make both the merged input node (that corresponds with those input nodes that were ancestors of $S$) and the output nodes that were ancestors of $S$ dashed. One can develop a theory for this operation and show the corresponding properties in parallel to what we did in \cref{sec:condDMG}. Note that one needs to replace stochastic conditional independence and the usual graphical $\sigma$-separation with transitional conditional independence \citep[Definition~3.1]{forre2021transitional} and a nuanced graphical separation \citep[Definition~5.9]{forre2021transitional}, respectively. 

We leave further exploration of the properties of the conditioning operations for iSCMs and on DMGs with input nodes for future work.

\newpage

\bibliographystyle{apalike}
\bibliography{bibfile_cdscm}

@article{Chen24Modeling,
  title = {Modeling Latent Selection with Structural Causal Models},
  author = {Leihao Chen and Onno Zoeter and Joris M. Mooij},
  year = {2024},
  journal= {arXiv.org preprint, arXiv:2401.06925v2 [cs.AI]},
  url = {https://arXiv.org/abs/2401.06925},
}

@article{mathur24simple_graphical,
    author = {Mathur, Maya B and Shpitser, Ilya},
    title = {Simple graphical rules for assessing selection bias in general-population and selected-sample treatment effects},
    journal = {American Journal of Epidemiology},
    volume = {194},
    number = {1},
    pages = {267-277},
    year = {2024},
    month = {06},
    issn = {0002-9262},
    doi = {10.1093/aje/kwae145},
    url = {https://doi.org/10.1093/aje/kwae145},
    eprint = {https://academic.oup.com/aje/article-pdf/194/1/267/59920542/kwae145.pdf},
}

@article{fritz2021_definetti_catprob_josa,
  title   = {De Finetti's Theorem in Categorical Probability},
  author  = {Tobias Fritz and Tomáš Gonda and Paolo Perrone},
  journal = {Journal of Stochastic Analysis},
  volume  = {2},
  number  = {4},
  year    = {2021},
  doi     = {10.31390/josa.2.4.06},
  url     = {https://repository.lsu.edu/josa/vol2/iss4/6/}
}

@article{lorenz2023causal_string_diagrams,
  title={Causal models in string diagrams},
  author={Robin Lorenz and Sean Tull},
  year={2023},
  journal={arXiv.org preprint,  arXiv:2304.07638 [cs.LO]},
  url={https://arXiv.org/abs/2304.07638},
}

@article{fritz2025empirical_slln_catprob,
  title={Empirical Measures and Strong Laws of Large Numbers in Categorical Probability},
  author={Tobias Fritz and Tomáš Gonda and Antonio Lorenzin and Paolo Perrone and Areeb Shah Mohammed},
  year={2025},
  journal={arXiv.org preprint,  arXiv:2503.21576 [math.PR]},
  url={https://arXiv.org/abs/2503.21576},
}

@article{mohammed2025partializations_markov_categories,
  title={Partializations of {Markov categories}},
  author={Areeb Shah Mohammed},
  year={2025},
  journal={arXiv.org preprint,  arXiv:2509.05094 [math.CT]},
  url={https://arXiv.org/abs/2509.05094},
}

@article{chen2024aldous_hoover_catprob,
  title={The {Aldous--Hoover} Theorem in Categorical Probability},
  author={Leihao Chen and Tobias Fritz and Tomáš Gonda and Andreas Klingler and Antonio Lorenzin},
  year={2024},
  journal={arXiv.org preprint,  arXiv:2411.12840 [math.ST]},
  url={https://arXiv.org/abs/2411.12840},
}

@article{fritz2023d,
  title={The d-separation criterion in categorical probability},
  author={Fritz, Tobias and Klingler, Andreas},
  journal={Journal of Machine Learning Research},
  volume={24},
  number={46},
  pages={1--49},
  year={2023}
}

@book{Lauritzen98generating,
title = "Generating mixed hierarchical interaction models by selection",
author = "Lauritzen, \{Steffen L.\}",
year = "1998",
language = "English",
series = "Research Report Series",
publisher = "Aalborg Universitetsforlag",
number = "R-98-2009",
address = "Denmark",
}

@article{Bhadane25Revisiting,
  title = {Revisiting the {Berkeley} Admissions data: Statistical Tests for Causal Hypotheses},
  author = {Sourbh Badhane and Joris M. Mooij and Philip Boeken and Onno Zoeter},
  journal = {arXiv.org preprint},
  volume = {arXiv:2502.10161 [stat.ME]},
  url = {https://arXiv.org/abs/2502.10161},
  month = 2,
  year = 2025,
}

@article{forre25causality,
  title={A Mathematical Introduction to Causality},
  author={Forr{\'e}, Patrick and Mooij, Joris M},
  year={2025},
  url={https://staff.fnwi.uva.nl/j.m.mooij/articles/causality_lecture_notes_2025.pdf}
}

@article{manski98bounding,
author = {Manski, Charles F. and Nagin, Daniel S.},
title = {Bounding Disagreements About Treatment Effects: A Case Study of Sentencing and Recidivism},
journal = {Sociological Methodology},
volume = {28},
number = {1},
pages = {99-137},
doi = {https://doi.org/10.1111/0081-1750.00043},
url = {https://onlinelibrary.wiley.com/doi/abs/10.1111/0081-1750.00043},
eprint = {https://onlinelibrary.wiley.com/doi/pdf/10.1111/0081-1750.00043},
year = {1998}
}

@article{Munafo2016ColliderSW,
  title={Collider scope: when selection bias can substantially influence observed associations},
  author={Marcus Robert Munafo and Kate Tilling and Amy E Taylor and David M. Evans and George Davey Smith},
  journal={International Journal of Epidemiology},
  year={2016},
  volume={47},
  pages={226 - 235},
}

@article{zhang2008causal,
  title={Causal reasoning with ancestral graphs},
  author={Zhang, Jiji},
  journal={Journal of Machine Learning Research},
  volume={9},
  pages={1437--1474},
  year={2008},
  publisher={Microtome Publishing}
}

@article{forre2021transitional,
  title={Transitional conditional independence},
  author={Forr{\'e}, Patrick},
  journal={arXiv.org preprint,  arXiv:2104.11547 [math.ST]},
  year={2021}
}

@techreport{wald1943method,
  author       = {Wald, Abraham},
  title        = {A Method of Estimating Plane Vulnerability Based on Damage of Survivors},
  institution  = {Statistical Research Group, Columbia University},
  year         = {1943},
  note         = {Reprinted as Center for Naval Analyses Research Contribution CRC-432, 1980}
}

@article{zhao21bets,
author = {Qingyuan Zhao and Nianqiao Ju and Sergio Bacallado and Rajen D. Shah},
title = {{BETS: The dangers of selection bias in early analyses of the coronavirus disease (COVID-19) pandemic}},
volume = {15},
journal = {The Annals of Applied Statistics},
number = {1},
publisher = {Institute of Mathematical Statistics},
pages = {363--390},
year = {2021},
doi = {10.1214/20-AOAS1401},
URL = {https://doi.org/10.1214/20-AOAS1401}
}

@article{fryer2019empirical,
  title={An empirical analysis of racial differences in police use of force},
  author={Fryer Jr, Roland G},
  journal={Journal of Political Economy},
  volume={127},
  number={3},
  pages={1210--1261},
  year={2019},
  publisher={The University of Chicago Press Chicago, IL}
}

@article{hernan2004structural,
 author = {Miguel A. Hernán and Sonia Hernández-Díaz and James M. Robins},
 journal = {Epidemiology},
 number = {5},
 pages = {615--625},
 publisher = {Lippincott Williams & Wilkins},
 title = {A Structural Approach to Selection Bias},
 volume = {15},
 year = {2004}
}

@book{hernan2020causal,
  author    = {Hern{\'a}n, Miguel A. and Robins, James M.},
  title     = {Causal Inference: What If},
  publisher={Boca Raton: Chapman \& Hall/CRC},
  year      = {2020}
}

@article{lu2022toward,
  title={Toward a clearer definition of selection bias when estimating causal effects},
  author={Lu, Haidong and Cole, Stephen R and Howe, Chanelle J and Westreich, Daniel},
  journal={Epidemiology},
  volume={33},
  number={5},
  pages={699--706},
  year={2022},
  publisher={Wolters Kluwer}
}

@InProceedings{bareinboim2012controlling,
  title = 	 {Controlling Selection Bias in Causal Inference},
  author = 	 {Bareinboim, Elias and Pearl, Judea},
  booktitle = 	 {Proceedings of the Fifteenth International Conference on Artificial Intelligence and Statistics},
  pages = 	 {100--108},
  year = 	 {2012},
  volume = 	 {22},
  series = 	 {Proceedings of Machine Learning Research},
  address = 	 {La Palma, Canary Islands},
  month = 	 {21--23 Apr},
  publisher =    {PMLR},
  url = 	 {https://proceedings.mlr.press/v22/bareinboim12.html}
}

@inproceedings{bareinboim2015recovering,
  title={Recovering causal effects from selection bias},
  author={Bareinboim, Elias and Tian, Jin},
  booktitle={Proceedings of the AAAI Conference on Artificial Intelligence},
  volume={29},
  number={1},
  pages={2410–2416},
  year={2015}
}

@article{richardson2002ancestral,
  title={Ancestral graph Markov models},
  author={Richardson, Thomas S. and Spirtes, Peter},
  journal={The Annals of Statistics},
  volume={30},
  number={4},
  pages={962--1030},
  year={2002},
  publisher={Institute of Mathematical Statistics}
}

@inproceedings{forre2018constraint,
  title={Constraint-based causal discovery for non-linear structural causal models with cycles and latent confounders},
  author={Forr{\'e}, Patrick and Mooij, Joris M},
  booktitle={Proceedings of the 34th Conference on Uncertainty in Artificial Intelligence},
  pages={269--278},
  year={2018},
}

@article{smith2020selection,
  title={Selection mechanisms and their consequences: understanding and addressing selection bias},
  author={Smith, Louisa H},
  journal={Current Epidemiology Reports},
  volume={7},
  pages={179--189},
  year={2020},
  publisher={Springer}
}

@book{pearl2009causality, place={Cambridge}, edition={2nd}, title={Causality: Models, Reasoning, and Inference}, DOI={10.1017/CBO9780511803161}, publisher={Cambridge University Press}, author={Pearl, Judea}, year={2009}}

@article{bongers2021foundations,
  title={Foundations of structural causal models with cycles and latent variables},
  author={Bongers, Stephan and Forr{\'e}, Patrick and Peters, Jonas and Mooij, Joris M},
  journal={The Annals of Statistics},
  volume={49},
  number={5},
  pages={2885--2915},
  year={2021},
  publisher={Institute of Mathematical Statistics}
}

@book{bogachev2007measure,
  title={Measure Theory},
  author={Bogachev, Vladimir Igorevich},
  volume={1},
  year={2007},
  publisher={Springer}
}

@article{richardson2023nested,
  title={Nested {Markov} properties for acyclic directed mixed graphs},
  author={Richardson, Thomas S and Evans, Robin J and Robins, James M and Shpitser, Ilya},
  journal={The Annals of Statistics},
  volume={51},
  number={1},
  pages={334--361},
  year={2023},
  publisher={Institute of Mathematical Statistics}
}

@inproceedings{forre2020causal,
  title={Causal calculus in the presence of cycles, latent confounders and selection bias},
  author={Forr{\'e}, Patrick and Mooij, Joris M},
  booktitle={Proceedings of the 36th Conference on Uncertainty in Artificial Intelligence},
  pages={71--80},
  year={2020},
}

@article{berkson1946limitations,
  title={Limitations of the application of fourfold table analysis to hospital data},
  author={Berkson, Joseph},
  journal={Biometrics Bulletin},
  volume={2},
  number={3},
  pages={47--53},
  year={1946},
  publisher={JSTOR}
}

@article{fritz2020synthetic,
  title={A synthetic approach to Markov kernels, conditional independence and theorems on sufficient statistics},
  author={Fritz, Tobias},
  journal={Advances in Mathematics},
  volume={370},
  pages={107239},
  year={2020},
  publisher={Elsevier}
}

@misc{richardson2013single,
  title = {Single world intervention graphs: a primer},
  howpublished = {\url{https://citeseerx.ist.psu.edu/document?repid=rep1&type=pdf&doi=07bbcb458109d2663acc0d098e8913892389a2a7}},
  author={Thomas S. Richardson and James M. Robins},
  year={2013},
  url={https://api.semanticscholar.org/CorpusID:13292974}
}

@article{richardson2013swig,
  title={Single world intervention graphs ({SWIGs}): A unification of the counterfactual and graphical approaches to causality},
  author={Richardson, Thomas S and Robins, James M},
  year={2013},
  publisher={Citeseer}
}

@inproceedings{hyttinen2014constraint,
author = {Hyttinen, Antti and Eberhardt, Frederick and J\"{a}rvisalo, Matti},
title = {Constraint-Based Causal Discovery: Conflict Resolution with Answer Set Programming},
year = {2014},
booktitle = {Proceedings of the 30th Conference on Uncertainty in Artificial Intelligence},
pages = {340–349},

}

@article{daniel2012using,
  title={Using causal diagrams to guide analysis in missing data problems},
  author={Daniel, Rhian M and Kenward, Michael G and Cousens, Simon N and De Stavola, Bianca L},
  journal={Statistical methods in medical research},
  volume={21},
  number={3},
  pages={243--256},
  year={2012},
  publisher={Sage Publications Sage UK: London, England}
}

@book{Reichenbach1956-REITDO-2,
	address = {Mineola, N.Y.},
	author = {Hans Reichenbach},
	editor = {Maria Reichenbach},
	publisher = {Dover Publications},
	title = {The Direction of Time},
	year = {1956}
}

@article{von2021simpson,
author = {Von K{\"u}gelgen, Julius and Gresele, Luigi and Sch{\"o}lkopf, Bernhard},
title = {{Simpson's paradox in COVID-19 case fatality rates: a mediation analysis of age-related causal effects}},
journal={IEEE Transactions on Artificial Intelligence},
  volume={2},
  number={1},
  pages={18--27},
  year={2021},
  publisher={IEEE}
}

@article{Bec2019abstract, title={Abstracting Causal Models}, volume={33}, url={https://ojs.aaai.org/index.php/AAAI/article/view/4117}, DOI={10.1609/aaai.v33i01.33012678}, abstractNote={&lt;p&gt;We consider a sequence of successively more restrictive definitions of abstraction for causal models, starting with a notion introduced by Rubenstein et al. (2017) called &lt;em&gt;exact transformation&lt;/em&gt; that applies to probabilistic causal models, moving to a notion of &lt;em&gt;uniform transformation&lt;/em&gt; that applies to deterministic causal models and does not allow differences to be hidden by the “right” choice of distribution, and then to &lt;em&gt;abstraction&lt;/em&gt;, where the interventions of interest are determined by the map from low-level states to high-level states, and &lt;em&gt;strong abstraction&lt;/em&gt;, which takes more seriously all potential interventions in a model, not just the allowed interventions. We show that procedures for combining micro-variables into macro-variables are instances of our notion of strong abstraction, as are all the examples considered by Rubenstein et al.&lt;/p&gt;}, number={01}, journal={Proceedings of the AAAI Conference on Artificial Intelligence}, author={Beckers, Sander and Halpern, Joseph Y.}, year={2019}, month={Jul.}, pages={2678-2685} }

@inproceedings{kusner2017counterfactual,
author = {Kusner, Matt and Loftus, Joshua and Russell, Chris and Silva, Ricardo},
title = {Counterfactual fairness},
year = {2017},
isbn = {9781510860964},
publisher = {Curran Associates Inc.},
address = {Red Hook, NY, USA},
booktitle = {Proceedings of the 31st International Conference on Neural Information Processing Systems},
pages = {4069–4079},
numpages = {11},
location = {Long Beach, California, USA},
series = {NIPS'17}
}

@inproceedings{zhang2018fairness,
  title={Fairness in decision-making—the causal explanation formula},
  author={Zhang, Junzhe and Bareinboim, Elias},
  booktitle={Proceedings of the AAAI Conference on Artificial Intelligence},
  volume={32},
  number={1},
  year={2018}
}

@incollection{peters2022causal,
  title={Causal models for dynamical systems},
  author={Peters, Jonas and Bauer, Stefan and Pfister, Niklas},
  booktitle={Probabilistic and Causal Inference: The Works of Judea Pearl},
  pages={671--690},
  year={2022}
}

@article{bongers2018causal,
  author = "Stephan Bongers and Tineke Blom and Joris M. Mooij",
  journal = "arXiv.org preprint",
  title = "Causal Modeling of Dynamical Systems",
  volume = "arXiv:1803.08784v4 [cs.AI]",
  year = 2022,
  url = "https://arXiv.org/abs/1803.08784v4",
}

@inproceedings{rubenstein17,
  title     = {Causal Consistency of Structural Equation Models},
  author    = {Paul K. Rubenstein and Sebastian Weichwald and Stephan Bongers and Joris M. Mooij and Dominik Janzing and Moritz Grosse-Wentrup and Bernhard Sch{\"o}lkopf},
  booktitle = {Proceedings of the 33rd Annual Conference on {U}ncertainty in {A}rtificial {I}ntelligence ({UAI}-17)},
  year      = 2017,
}

@inproceedings{correa2021nest,
 author = {Correa, Juan and Lee, Sanghack and Bareinboim, Elias},
 booktitle = {Advances in Neural Information Processing Systems},
 editor = {M. Ranzato and A. Beygelzimer and Y. Dauphin and P.S. Liang and J. Wortman Vaughan},
 pages = {6856--6867},
 publisher = {Curran Associates, Inc.},
 title = {Nested Counterfactual Identification from Arbitrary Surrogate Experiments},
 url = {https://proceedings.neurips.cc/paper_files/paper/2021/file/36bedb6eb7152f39b16328448942822b-Paper.pdf},
 volume = {34},
 year = {2021}
}

@InProceedings{versteeg2022local,
  title = 	 {Local Constraint-Based Causal Discovery under Selection Bias},
  author =       {Versteeg, Philip and Mooij, Joris and Zhang, Cheng},
  booktitle = 	 {Proceedings of the First Conference on Causal Learning and Reasoning},
  pages = 	 {840--860},
  year = 	 {2022},
  editor = 	 {Schölkopf, Bernhard and Uhler, Caroline and Zhang, Kun},
  volume = 	 {177},
  series = 	 {Proceedings of Machine Learning Research},
  month = 	 {11--13 Apr},
  publisher =    {PMLR},
  url = 	 {https://proceedings.mlr.press/v177/versteeg22a.html}
}

@article{Forre2017markov,
  author = "Patrick Forr{\'e} and Joris M. Mooij",
  journal = "arXiv.org preprint",
  title = "Markov Properties for Graphical Models with Cycles and Latent Variables",
  volume = "arXiv:1710.08775 [math.ST]",
  year = 2017,
  url = "https://arXiv.org/abs/1710.08775",
}

@article{rubin1974estimating,
  title={Estimating causal effects of treatments in randomized and nonrandomized studies},
  author={Rubin, Donald B},
  journal={Journal of Educational Psychology},
  volume={66},
  number={5},
  pages={688-701},
  year={1974},
  publisher={American Psychological Association}
}

@book{Williams91prob, place={Cambridge}, title={Probability with Martingales}, publisher={Cambridge University Press}, author={Williams, David}, year={1991}}

@book{spirtes2001causation,
  title={Causation, Prediction, and Search},
  author={Spirtes, Peter and Glymour, Clark and Scheines, Richard},
  year={2001},
  publisher={MIT Press}
}

@inproceedings{spirtes95causal,
author = {Spirtes, Peter and Meek, Christopher and Richardson, Thomas},
title = {Causal inference in the presence of latent variables and selection bias},
year = {1995},
booktitle = {Proceedings of the 11th Conference on Uncertainty in Artificial Intelligence},
pages = {499–506},
numpages = {8},
}

@incollection{spirtes99alogorithm,
    author = {Spirtes, P and Meek, C and Richardson, T},
    isbn = {9780262315821},
    title = "{An algorithm for causal inference in the presence of latent variables and selection bias}",
    booktitle = "{Computation, Causation, and Discovery}",
    publisher = {AAAI Press},
    year = {1999},
    month = {05},
    doi = {10.7551/mitpress/2006.003.0009},
    url = {https://doi.org/10.7551/mitpress/2006.003.0009},
    eprint = {https://direct.mit.edu/book/chapter-pdf/278764/9780262315821\_cah.pdf},
}

@inproceedings{MooijClaassen20constraint,
  title = {Constraint-Based Causal Discovery using Partial Ancestral Graphs in the presence of Cycles},
  author = {Joris M. Mooij and Tom Claassen},
  booktitle = {Proceedings of the 36th Conference on {U}ncertainty in {A}rtificial {I}ntelligence ({UAI}-20)},
  editor = {Jonas Peters and David Sontag},
  publisher = {PMLR},
  volume = 124,
  pages = {1159--1168},
  year = 2020,
  month = 8,
  url = {http://proceedings.mlr.press/v124/m-mooij20a/m-mooij20a-supp.pdf}
}

@article{geiger90identifying,
author = {Geiger, Dan and Verma, Thomas and Pearl, Judea},
title = {Identifying independence in bayesian networks},
journal = {Networks},
volume = {20},
number = {5},
pages = {507-534},
doi = {https://doi.org/10.1002/net.3230200504},
url = {https://onlinelibrary.wiley.com/doi/abs/10.1002/net.3230200504},
eprint = {https://onlinelibrary.wiley.com/doi/pdf/10.1002/net.3230200504},
year = {1990}
}

@article{Correa17causal, title={Causal Effect Identification by Adjustment under Confounding and Selection Biases}, volume={31}, url={https://ojs.aaai.org/index.php/AAAI/article/view/11060}, DOI={10.1609/aaai.v31i1.11060}, abstractNote={ &lt;p&gt; Controlling for selection and confounding biases are two of the most challenging problems in the empirical sciences as well as in artificial intelligence tasks. Covariate adjustment (or, Backdoor Adjustment) is the most pervasive technique used for controlling confounding bias, but the same is oblivious to issues of sampling selection. In this paper, we introduce a generalized version of covariate adjustment that simultaneously controls for both confounding and selection biases. We first derive a sufficient and necessary condition for recovering causal effects using covariate adjustment from an observational distribution collected under preferential selection. We then relax this setting to consider cases when additional, unbiased measurements over a set of covariates are available for use (e.g., the age and gender distribution obtained from census data). Finally, we present a complete algorithm with polynomial delay to find all sets of admissible covariates for adjustment when confounding and selection biases are simultaneously present and unbiased data is available. &lt;/p&gt; }, number={1}, journal={Proceedings of the AAAI Conference on Artificial Intelligence}, author={Correa, Juan and Bareinboim, Elias}, year={2017}, month={Feb.} }

@inproceedings{shpitser06identification,
author = {Shpitser, Ilya and Pearl, Judea},
title = {Identification of conditional interventional distributions},
year = {2006},
booktitle = {Proceedings of the 22d Conference on Uncertainty in Artificial Intelligence},
pages = {437–444},
numpages = {8},

}

@article{pearl95causal,
  author  = {Pearl, Judea},
  title   = {Causal Diagrams for Empirical Research},
  journal = {Biometrika},
  year    = {1995},
  volume  = {82},
  number  = {4},
  pages   = {669--710},
  note    = {With discussion}
}

@inproceedings{tian02general,
author = {Tian, Jin and Pearl, Judea},
title = {A general identification condition for causal effects},
year = {2002},
isbn = {0262511290},
publisher = {American Association for Artificial Intelligence},
address = {USA},
booktitle = {Eighteenth National Conference on Artificial Intelligence},
pages = {567–573},
numpages = {7},
location = {Edmonton, Alberta, Canada}
}

@inproceedings{huang06pearl,
author = {Huang, Yimin and Valtorta, Marco},
title = {Pearl's calculus of intervention is complete},
year = {2006},
booktitle = {Proceedings of the 22nd Conference on Uncertainty in Artificial Intelligence},
pages = {217–224},
numpages = {8},
}

@article{Huang2008OnTC,
  title={On the completeness of an identifiability algorithm for semi-Markovian models},
  author={Yimin Huang and Marco G. Valtorta},
  journal={Annals of Mathematics and Artificial Intelligence},
  year={2008},
  volume={54},
  pages={363-408},
  url={https://api.semanticscholar.org/CorpusID:52818662}
}

@inproceedings{shpister06joint,
title = "Identification of joint interventional distributions in recursive semi-Markovian causal models",
author = "Ilya Shpitser and Judea Pearl",
year = "2006",
month = nov,
day = "13",
language = "English (US)",
isbn = "1577352815",
series = "Proceedings of the National Conference on Artificial Intelligence",
pages = "1219--1226",

}

@inproceedings{yaroslav23identifiability,
author = {Kivva, Yaroslav and Etesami, Jalal and Kiyavash, Negar},
title = {On identifiability of conditional causal effects},
year = {2023},
booktitle = {Proceedings of the 39th Conference on Uncertainty in Artificial Intelligence},
}

@article{Abouei24sID, title={s-ID: Causal Effect Identification in a Sub-population}, volume={38}, url={https://ojs.aaai.org/index.php/AAAI/article/view/30011}, DOI={10.1609/aaai.v38i18.30011}, abstractNote={Causal inference in a sub-population involves identifying the causal effect of an intervention on a specific subgroup, which is distinguished from the whole population through the influence of systematic biases in the sampling process. However, ignoring the subtleties introduced by sub-populations can either lead to erroneous inference or limit the applicability of existing methods. We introduce and advocate for a causal inference problem in sub-populations (henceforth called s-ID), in which we merely have access to observational data of the targeted sub-population (as opposed to the entire population). Existing inference problems in sub-populations operate on the premise that the given data distributions originate from the entire population, thus, cannot tackle the s-ID problem. To address this gap, we provide necessary and sufficient conditions that must hold in the causal graph for a causal effect in a sub-population to be identifiable from the observational distribution of that sub-population. Given these conditions, we present a sound and complete algorithm for the s-ID problem.}, number={18}, journal={Proceedings of the AAAI Conference on Artificial Intelligence}, author={Abouei, Amir Mohammad and Mokhtarian, Ehsan and Kiyavash, Negar}, year={2024}, month={Mar.}, pages={20302-20310} }

@inproceedings{pearl95instrumental,
author = {Pearl, Judea},
title = {On the testability of causal models with latent and instrumental variables},
year={1995},
booktitle = {Proceedings of the 11th Conference on Uncertainty in Artificial Intelligence},
pages = {435–443},
numpages = {9},
}

@inproceedings{balke94counterfactual,
author = {Balke, Alexander and Pearl, Judea},
title = {Counterfactual probabilities: computational methods, bounds and applications},
year = {1994},
isbn = {1558603328},
publisher = {Morgan Kaufmann Publishers Inc.},
address = {San Francisco, CA, USA},
booktitle = {Proceedings of the 10th International Conference on Uncertainty in Artificial Intelligence},
pages = {46–54},
numpages = {9},
location = {Seattle, WA},
series = {UAI'94}
}

@article{balke97bounds,
author = {Alexander Balke and Judea Pearl},
title = {Bounds on Treatment Effects from Studies with Imperfect Compliance},
journal = {Journal of the American Statistical Association},
volume = {92},
number = {439},
pages = {1171--1176},
year = {1997},
publisher = {Taylor \& Francis},
doi = {10.1080/01621459.1997.10474074},
URL = { 
    
        https://doi.org/10.1080/01621459.1997.10474074  
},
eprint = {     
        https://doi.org/10.1080/01621459.1997.10474074
}
}

@article{VanHimbeeck19quantum,
  doi = {10.22331/q-2019-09-16-186},
  url = {https://doi.org/10.22331/q-2019-09-16-186},
  title = {Quantum violations in the {I}nstrumental scenario and their relations to the {B}ell scenario},
  author = {Van Himbeeck, Thomas and Bohr Brask, Jonatan and Pironio, Stefano and Ramanathan, Ravishankar and Sainz, Ana Bel{\'{e}}n and Wolfe, Elie},
  journal = {{Quantum}},
  issn = {2521-327X},
  publisher = {{Verein zur F{\"{o}}rderung des Open Access Publizierens in den Quantenwissenschaften}},
  volume = {3},
  pages = {186},
  month = sep,
  year = {2019}
}

@article{Zhang21bounding, title={Bounding Causal Effects on Continuous Outcome}, volume={35}, url={https://ojs.aaai.org/index.php/AAAI/article/view/17449}, DOI={10.1609/aaai.v35i13.17449}, abstractNote={We investigate the problem of bounding causal effects from experimental studies in which treatment assignment is randomized but the subject compliance is imperfect. It is well known that under such conditions, the actual causal effects are not point-identifiable due to uncontrollable unobserved confounding. In their seminal work, Balke and Pearl (1994) derived the tightest bounds over the causal effects in this settings by employing an algebra program to derive analytic expressions. However, Pearl’s approach assumes the primary outcome to be discrete and finite. Solving such a program could be intractable when high-dimensional context variables are present. In this paper, we present novel non-parametric methods to bound causal effects on the continuous outcome from studies with imperfect compliance. These bounds could be generalized to settings with a high-dimensional context.}, number={13}, journal={Proceedings of the AAAI Conference on Artificial Intelligence}, author={Zhang, Junzhe and Bareinboim, Elias}, year={2021}, month={May}, pages={12207-12215} }

@article {imbens00instrumental,
	title = {The Interpretation of Instrumental Variables Estimators in Simultaneous Equations Models with an Application to the Demand for Fish},
	journal = {Review of Economic Studies},
	volume = {67, July},
	year = {2000},
	pages = {499-527},
	author = {Guido Imbens and Joshua Angrist and Kathryn Graddy}
}

@article{Heckman79SampleSB,
 ISSN = {00129682, 14680262},
 URL = {http://www.jstor.org/stable/1912352},
 author = {James J. Heckman},
 journal = {Econometrica},
 number = {1},
 pages = {153--161},
 publisher = {[Wiley, Econometric Society]},
 title = {Sample selection bias as a specification error},
 urldate = {2025-11-15},
 volume = {47},
 year = {1979}
}

@inproceedings{Blom19beyond,
  title     = {Beyond Structural Causal Models: Causal Constraints Models},
  author    = {Tineke Blom and Stephan Bongers and Joris M. Mooij},
  pages     = {585-594},
  url       = {https://proceedings.mlr.press/v115/blom20a.html},
  booktitle = {Proceedings of the 35th Uncertainty in Artificial Intelligence Conference ({UAI}-19)},
  year      = 2020,
  publisher = {PMLR},
  series    = {Proceedings of Machine Learning Research},
  volume    = 115,
}

@article{Geiger23CausalAF,
  title={Causal Abstraction for Faithful Model Interpretation},
  author={Atticus Geiger and Chris Potts and Thomas F. Icard},
  journal={arXiv.org preprint},
  year={2023},
  volume={arXiv:2301.04709 [cs.AI]},
}

@article{Evans16graphs_margin,
author = {Evans, Robin J.},
title = {Graphs for Margins of {Bayesian} Networks},
journal = {Scandinavian Journal of Statistics},
volume = {43},
number = {3},
pages = {625-648},
doi = {https://doi.org/10.1111/sjos.12194},
url = {https://onlinelibrary.wiley.com/doi/abs/10.1111/sjos.12194},
eprint = {https://onlinelibrary.wiley.com/doi/pdf/10.1111/sjos.12194},
year = {2016}
}

@InProceedings{cooper95causal_discovery_selecion,
  title = 	 {Causal Discovery from Data in the Presence of Selection Bias},
  author =       {Cooper, Gregory F.},
  booktitle = 	 {Pre-proceedings of the Fifth International Workshop on Artificial Intelligence and Statistics},
  pages = 	 {140--150},
  year = 	 {1995},
  editor = 	 {Fisher, Doug and Lenz, Hans-Joachim},
  volume = 	 {R0},
  series = 	 {Proceedings of Machine Learning Research},
  month = 	 {04--07 Jan},
  publisher =    {PMLR},
  url = 	 {https://proceedings.mlr.press/r0/cooper95a.html},
  note =         {Reissued by PMLR on 01 May 2022.}
}

@article{richardson03markov_admg,
author = {Richardson, Thomas},
title = {Markov Properties for Acyclic Directed Mixed Graphs},
journal = {Scandinavian Journal of Statistics},
volume = {30},
number = {1},
pages = {145-157},
doi = {https://doi.org/10.1111/1467-9469.00323},
url = {https://onlinelibrary.wiley.com/doi/abs/10.1111/1467-9469.00323},
eprint = {https://onlinelibrary.wiley.com/doi/pdf/10.1111/1467-9469.00323},
year = {2003}
}

@article{Maathuis15generalized,
author = {Marloes H. Maathuis and Diego Colombo},
title = {{A generalized back-door criterion}},
volume = {43},
journal = {The Annals of Statistics},
number = {3},
publisher = {Institute of Mathematical Statistics},
pages = {1060 -- 1088},
year = {2015},
doi = {10.1214/14-AOS1295},
URL = {https://doi.org/10.1214/14-AOS1295}
}

@inbook{Jaynes03probability, place={Cambridge}, title={Paradoxes of probability theory}, booktitle={Probability Theory: The Logic of Science}, publisher={Cambridge University Press}, author={Jaynes, E. T.}, editor={Bretthorst, G. LarryEditor}, year={2003}, pages={451–489}}

@article{Dawid21decision,
url = {https://doi.org/10.1515/jci-2020-0008},
title = {Decision-theoretic foundations for statistical causality},
author = {A. Philip Dawid},
pages = {39--77},
volume = {9},
number = {1},
journal = {Journal of Causal Inference},
doi = {doi:10.1515/jci-2020-0008},
year = {2021},
lastchecked = {2024-06-03}
}

@article{pearl2015conditioning,
  title={Conditioning on post-treatment variables},
  author={Pearl, Judea},
  journal={Journal of Causal Inference},
  volume={3},
  number={1},
  pages={131--137},
  year={2015},
  publisher={De Gruyter}
}

@article{wang23discovery,
  author  = {Y. Samuel Wang and Mathias Drton},
  title   = {Causal Discovery with Unobserved Confounding and Non-Gaussian Data},
  journal = {Journal of Machine Learning Research},
  year    = {2023},
  volume  = {24},
  number  = {271},
  pages   = {1--61},
  url     = {http://jmlr.org/papers/v24/21-1329.html}
}

@article{Robins1992IdentifiabilityAE,
  title={Identifiability and Exchangeability for Direct and Indirect Effects},
  author={James M. Robins and Sander Greenland},
  journal={Epidemiology},
  year={1992},
  volume={3},
  pages={143-155},
  url={https://api.semanticscholar.org/CorpusID:10757981}
}

@article{Nabi18fair, title={Fair Inference on Outcomes}, volume={32}, url={https://ojs.aaai.org/index.php/AAAI/article/view/11553}, DOI={10.1609/aaai.v32i1.11553}, abstractNote={ &lt;p&gt; In this paper, we consider the problem of fair statistical inference involving outcome variables. Examples include classification and regression problems, and estimating treatment effects in randomized trials or observational data. The issue of fairness arises in such problems where some covariates or treatments are &quot;sensitive,&quot; in the sense of having potential of creating discrimination. In this paper, we argue that the presence of discrimination can be formalized in a sensible way as the presence of an effect of a sensitive covariate on the outcome along certain causal pathways, a view which generalizes (Pearl 2009). A fair outcome model can then be learned by solving a constrained optimization problem. We discuss a number of complications that arise in classical statistical inference due to this view and provide workarounds based on recent work in causal and semi-parametric inference. &lt;/p&gt; }, number={1}, journal={Proceedings of the AAAI Conference on Artificial Intelligence}, author={Nabi, Razieh and Shpitser, Ilya}, year={2018}, month={Apr.} }

@article{Chiappa19path-specific, title={Path-Specific Counterfactual Fairness}, volume={33}, url={https://ojs.aaai.org/index.php/AAAI/article/view/4777}, DOI={10.1609/aaai.v33i01.33017801}, abstractNote={&lt;p&gt;We consider the problem of learning fair decision systems from data in which a sensitive attribute might affect the decision along both fair and unfair pathways. We introduce a counterfactual approach to disregard effects along unfair pathways that does not incur in the same loss of individual-specific information as previous approaches. Our method corrects observations adversely affected by the sensitive attribute, and uses these to form a decision. We leverage recent developments in deep learning and approximate inference to develop a VAE-type method that is widely applicable to complex nonlinear models.&lt;/p&gt;}, number={01}, journal={Proceedings of the AAAI Conference on Artificial Intelligence}, author={Chiappa, Silvia}, year={2019}, month={Jul.}, pages={7801-7808} }

@article{shpitser2013counterfactual,
  title={Counterfactual graphical models for longitudinal mediation analysis with unobserved confounding},
  author={Shpitser, Ilya},
  journal={Cognitive Science},
  volume={37},
  number={6},
  pages={1011--1035},
  year={2013},
  publisher={Wiley Online Library}
}

@article {pearl14interpretation,
	Title = {Interpretation and identification of causal mediation},
	Author = {Pearl, Judea},
	DOI = {10.1037/a0036434},
	Number = {4},
	Volume = {19},
	Month = {December},
	Year = {2014},
	Journal = {Psychological methods},
	ISSN = {1082-989X},
	Pages = {459—481},
	URL = {https://doi.org/10.1037/a0036434},
}

@inproceedings{pearl2001direct,
  title={Direct and indirect effects},
  author={Pearl, Judea},
  booktitle={Proceedings of the 17th Conference on Uncertainty in Artificial Intelligence},
  pages={411--420},
  year={2001},
  publisher = {PMLR},
}

@article{Correa2020stochastic, title={A Calculus for Stochastic Interventions:Causal Effect Identification and Surrogate Experiments}, volume={34}, url={https://ojs.aaai.org/index.php/AAAI/article/view/6567}, DOI={10.1609/aaai.v34i06.6567}, abstractNote={&lt;p&gt;Some of the most prominent results in causal inference have been developed in the context of atomic interventions, following the semantics of the &lt;em&gt;do&lt;/em&gt;-operator and the inferential power of the &lt;em&gt;do&lt;/em&gt;-calculus. In practice, many real-world settings require more complex types of interventions that cannot be represented by a simple atomic intervention. In this paper, we investigate a general class of interventions that covers some non-trivial types of policies (conditional and stochastic), which goes beyond the atomic class. Our goal is to develop general understanding and formal machinery to be able to reason about the effects of those policies, similar to the robust treatment developed to handle the atomic case. Specifically, in this paper, we introduce a new set of inference rules (akin to &lt;em&gt;do&lt;/em&gt;-calculus) that can be used to derive claims about general interventions, which we call σ-calculus. We develop a systematic and efficient procedure for finding estimands of the effect of general policies as a function of the available observational and experimental distributions. We then prove that our algorithm and σ-calculus are both sound for the tasks of identification (Pearl, 1995) and z-identification (Bareinboim and Pearl, 2012) under this class of interventions.&lt;/p&gt;}, number={06}, journal={Proceedings of the AAAI Conference on Artificial Intelligence}, author={Correa, Juan and Bareinboim, Elias}, year={2020}, pages={10093-10100} }

@article{Pearl19interpretation,
url = {https://doi.org/10.1515/jci-2019-2002},
title = {On the Interpretation of do(x)},
author = {Judea Pearl},
volume = {7},
number = {1},
journal = {Journal of Causal Inference},
doi = {doi:10.1515/jci-2019-2002},
year = {2019},
lastchecked = {2024-06-10}
}

@article{abouei24sIDlatent,
      title={Causal Effect Identification in a Sub-Population with Latent Variables}, 
      author={Amir Mohammad Abouei and Ehsan Mokhtarian and Negar Kiyavash and Matthias Grossglauser},
      journal={arXiv.org preprint,  arXiv:2405.14547 [cs.LG]},
      year={2024},
  
}

@article{WolfeSpekkensFritz+2019,
url = {https://doi.org/10.1515/jci-2017-0020},
title = {The Inflation Technique for Causal Inference with Latent Variables},
author = {Elie Wolfe and Robert W. Spekkens and Tobias Fritz},
pages = {20170020},
volume = {7},
number = {2},
journal = {Journal of Causal Inference},
doi = {doi:10.1515/jci-2017-0020},
year = {2019},
lastchecked = {2024-06-13}
}

@inbook{bareinboim22pearl_hierarchy,
author = {Bareinboim, Elias and Correa, Juan D. and Ibeling, Duligur and Icard, Thomas},
title = {On Pearl’s Hierarchy and the Foundations of Causal Inference},
year = {2022},
isbn = {9781450395861},
publisher = {Association for Computing Machinery},
address = {New York, NY, USA},
edition = {1},
url = {https://doi.org/10.1145/3501714.3501743},
booktitle = {Probabilistic and Causal Inference: The Works of Judea Pearl},
pages = {507–556},
numpages = {50}
}

@article{vanderweele13defconfounder,
          volume = {41},
          number = {1},
           month = {February},
          author = {Tyler J. VanderWeele and Ilya Shpitser},
           title = {On the definition of a confounder},
         journal = {The Annals of Statistics},
           pages = {196--220},
            year = {2013}
}

@article{Dawid02influence,
author = {Dawid, A. P.},
title = {Influence Diagrams for Causal Modelling and Inference},
journal = {International Statistical Review},
volume = {70},
number = {2},
pages = {161-189},
doi = {https://doi.org/10.1111/j.1751-5823.2002.tb00354.x},
url = {https://onlinelibrary.wiley.com/doi/abs/10.1111/j.1751-5823.2002.tb00354.x},
eprint = {https://onlinelibrary.wiley.com/doi/pdf/10.1111/j.1751-5823.2002.tb00354.x},
year = {2002}
}

\end{document}